\pdfoutput=1  

\documentclass[11pt]{article}
\usepackage[margin=.8in,left=.8in]{geometry}
\usepackage{amsmath}
\usepackage{amsfonts}
\usepackage{amssymb}
\usepackage{amsthm}
\usepackage{mathtools}
\usepackage{subcaption}
\usepackage{float}
\usepackage{graphicx}
\usepackage{color}
\usepackage{rotating}
\usepackage{xcolor}
\usepackage[utf8]{inputenc}
\usepackage{hyperref}
\usepackage[all]{xy}

\usepackage{tikz}
\usepackage{tikz-cd}
\usepackage{lipsum}
\usepackage{adjustbox}

\usepackage{multirow}

\usepackage{stmaryrd}    

\usepackage{enumerate} 

\usepackage{helvet}   

\usepackage{bigints} 

\usepackage{mathptmx}
\usepackage{amsmath}
\usepackage{graphicx}

\definecolor{darkblue}{rgb}{0.05,0.25,0.65}
\definecolor{greenii}{RGB}{20,140,10}
\definecolor{orangeii}{RGB}{200,100,5}

\makeatletter
\newcommand{\raisemath}[1]{\mathpalette{\raisem@th{#1}}}
\newcommand{\raisem@th}[3]{\raisebox{#1}{$#2#3$}}
\makeatother

\usepackage{tabularx}   

\usepackage[new]{old-arrows}   

\newdir{> }{{}*!/10pt/@{>}}

\DeclareRobustCommand{\rchi}{{\mathpalette\irchi\relax}}
\newcommand{\irchi}[2]{\raisebox{\depth}{$#1\chi$}} 

\newcommand{\mapsup}{\mbox{$\;\,$\begin{rotate}{90}$\!\!\!\mapsto$\end{rotate}}}
\newcommand{\mapsdown}{\mbox{$\!\!\!\!$\begin{rotate}{-90}$\!\!\!\!\!\!\!\!\mapsto$\end{rotate}}}

\makeatletter
\newif\if@sup
\newtoks\@sups
\def\append@sup#1{\edef\act{\noexpand\@sups={\the\@sups #1}}\act}%
\def\reset@sup{\@supfalse\@sups={}}%
\def\mk@scripts#1#2{\if #2/ \if@sup ^{\the\@sups}\fi \else%
  \ifx #1_ \if@sup ^{\the\@sups}\reset@sup \fi {}_{#2}%
  \else \append@sup#2 \@suptrue \fi%
  \expandafter\mk@scripts\fi}
\def\tensor#1#2{\reset@sup#1\mk@scripts#2_/}
\def\multiscripts#1#2#3{\reset@sup{}\mk@scripts#1_/#2%
  \reset@sup\mk@scripts#3_/}
\makeatother

\makeatletter
\newbox\slashbox \setbox\slashbox=\hbox{$/$}
\def\itex@pslash#1{\setbox\@tempboxa=\hbox{$#1$}
  \@tempdima=0.5\wd\slashbox \advance\@tempdima 0.5\wd\@tempboxa
  \copy\slashbox \kern-\@tempdima \box\@tempboxa}
\def\slash{\protect\itex@pslash}
\makeatother

\def\clap#1{\hbox to 0pt{\hss#1\hss}}
\def\mathllap{\mathpalette\mathllapinternal}
\def\mathrlap{\mathpalette\mathrlapinternal}
\def\mathclap{\mathpalette\mathclapinternal}
\def\mathllapinternal#1#2{\llap{$\mathsurround=0pt#1{#2}$}}
\def\mathrlapinternal#1#2{\rlap{$\mathsurround=0pt#1{#2}$}}
\def\mathclapinternal#1#2{\clap{$\mathsurround=0pt#1{#2}$}}

\let\oldroot\root
\def\root#1#2{\oldroot #1 \of{#2}}
\renewcommand{\sqrt}[2][]{\oldroot #1 \of{#2}}

\DeclareSymbolFont{symbolsC}{U}{txsyc}{m}{n}
\SetSymbolFont{symbolsC}{bold}{U}{txsyc}{bx}{n}
\DeclareFontSubstitution{U}{txsyc}{m}{n}

\DeclareSymbolFont{stmry}{U}{stmry}{m}{n}
\SetSymbolFont{stmry}{bold}{U}{stmry}{b}{n}

\DeclareFontFamily{OMX}{MnSymbolE}{}
\DeclareSymbolFont{mnomx}{OMX}{MnSymbolE}{m}{n}
\SetSymbolFont{mnomx}{bold}{OMX}{MnSymbolE}{b}{n}
\DeclareFontShape{OMX}{MnSymbolE}{m}{n}{
    <-6>  MnSymbolE5
   <6-7>  MnSymbolE6
   <7-8>  MnSymbolE7
   <8-9>  MnSymbolE8
   <9-10> MnSymbolE9
  <10-12> MnSymbolE10
  <12->   MnSymbolE12}{}

\makeatletter
\def\Decl@Mn@Delim#1#2#3#4{%
  \if\relax\noexpand#1%
    \let#1\undefined
  \fi
  \DeclareMathDelimiter{#1}{#2}{#3}{#4}{#3}{#4}}
\def\Decl@Mn@Open#1#2#3{\Decl@Mn@Delim{#1}{\mathopen}{#2}{#3}}
\def\Decl@Mn@Close#1#2#3{\Decl@Mn@Delim{#1}{\mathclose}{#2}{#3}}
\Decl@Mn@Open{\llangle}{mnomx}{'164}
\Decl@Mn@Close{\rrangle}{mnomx}{'171}
\Decl@Mn@Open{\lmoustache}{mnomx}{'245}
\Decl@Mn@Close{\rmoustache}{mnomx}{'244}
\makeatother

\makeatletter
\DeclareRobustCommand\widecheck[1]{{\mathpalette\@widecheck{#1}}}
\def\@widecheck#1#2{%
    \setbox\z@\hbox{\m@th$#1#2$}%
    \setbox\tw@\hbox{\m@th$#1%
       \widehat{%
          \vrule\@width\z@\@height\ht\z@
          \vrule\@height\z@\@width\wd\z@}$}%
    \dp\tw@-\ht\z@
    \@tempdima\ht\z@ \advance\@tempdima2\ht\tw@ \divide\@tempdima\thr@@
    \setbox\tw@\hbox{%
       \raise\@tempdima\hbox{\scalebox{1}[-1]{\lower\@tempdima\box
\tw@}}}%
    {\ooalign{\box\tw@ \cr \box\z@}}}
\makeatother

\makeatletter
\def\udots{\mathinner{\mkern2mu\raise\p@\hbox{.}
\mkern2mu\raise4\p@\hbox{.}\mkern1mu
\raise7\p@\vbox{\kern7\p@\hbox{.}}\mkern1mu}}
\makeatother

\newcommand{\Z}{\ensuremath{\mathbb Z}}

\renewcommand{\(}{\begin{equation}}
\renewcommand{\)}{\end{equation}}
\newcommand{\bea}{\begin{eqnarray*}}
\newcommand{\eea}{\end{eqnarray*}}

\usepackage{cleveref}

\crefformat{section}{\S#2#1#3} 
\crefformat{subsection}{\S#2#1#3}
\crefformat{subsubsection}{\S#2#1#3}

\theoremstyle{italics}
\newtheorem{theorem}{Theorem}[section]
\newtheorem{lemma}[theorem]{Lemma}
\newtheorem{prop}[theorem]{Proposition}

\theoremstyle{definition}
\newtheorem{defn}[theorem]{Definition}

\newtheorem{example}[theorem]{Example}

\newtheorem{remark}[theorem]{Remark}
\newtheorem{note[theorem]}{Note}

\usepackage{amsfonts}

\begin{document}

\title{Twisted Cohomotopy implies level quantization
 \\ of the full 6d Wess-Zumino term of the M5-brane}

 \author{Domenico Fiorenza, \; Hisham Sati, \; Urs Schreiber}

\maketitle

\begin{abstract}
  The full
  6d Hopf-Wess-Zumino term in the action functional for the M5-brane
  is anomalous as traditionally defined.
  What has been missing is a condition implying
  the higher analogue of level quantization
  familiar from the 2d Wess-Zumino term.
  We prove that the anomaly cancellation condition is implied by
  the hypothesis
  that the C-field is charge-quantized in twisted Cohomotopy theory.
  The proof follows by a twisted/parametrized generalization of
  the Hopf invariant, after identifying the
  full 6d Hopf-Wess-Zumino term with
  a twisted homotopy Whitehead integral formula, which we establish.
\end{abstract}

\medskip

\tableofcontents

\medskip

\section{Introduction and results}

The expected but elusive quantum theory of M5-branes in M-theory
(see \cite[\S 3]{Duff99}\cite[\S 2]{HSS18})
has come to be widely regarded as a core open problem in string theory,
already in its decoupling limit of an expected 6-dimensional superconformal quantum field theory
(see \cite{Moore12}\cite{HeckmanRudelius18}).
Most attempts to understand at least aspects of this theory have
been based on analogies (such as with the known M2-brane theory) and consistency checks (such as from
implications of the expected superconformal structure). But a systematic derivation of the theory from
deeper principles has not been possible, since these deeper principles must be those of the ambient M-theory, whose formulation
is itself a wide open problem
(\cite[\S 6]{Duff96}\cite[\S 12]{Moore14}\cite[$@$21:15]{Witten19}\cite[$@$17:14]{Duff19}).

\medskip
Recently in \cite{FSS19a}, following \cite{Sati13}, we motivated, from rigorous analysis of the super
homotopy theory of super $p$-branes initiated in \cite{FSS13},
a hypothesis about the mathematical
foundations of microscopic M-theory:

\vspace{.2cm}

\hypertarget{HypothesisH}{}
\noindent
\hspace{-.18cm}
\def\arraystretch{1.3}
\begin{tabular}{|l|}
  \hline
  \hspace{-.2cm}
  {\bf Hypothesis H.}
  {\it
  The M-theory C-field is charge-quantized
  in
  J-twisted Cohomotopy theory}
  (\cite{FSS20b}, Def. \ref{BackgroundFieldsSatisfyingHypothesisH}).
  \hspace{-.2cm}
  \\
  \hline
\end{tabular}

\vspace{.2cm}

We proved in \cite{FSS19b}\cite{SS19a}\cite{SS19b}\cite{SS20a}\cite{FSS20a}
that this hypothesis implies a list of subtle consistency conditions that had informally been argued to be
necessary for M-theory to exist. This suggests that
{\it Hypothesis H}
could indeed be
a correct assumption about the
mathematical principles underlying microscopic M-theory. If this is the case, further aspects of M-theory
must be systematically derivable, by rigorous mathematical deduction.

\medskip
Here we prove that \hyperlink{HypothesisH}{\it Hypothesis H} implies global consistency of the
full Hopf-Wess-Zumino term that appears
in the Green-Schwarz-type action functional of the M5-brane.
This used to be an open problem:

\newpage

\noindent {\bf The open problem.}
The full 6d Hopf-Wess-Zumino-term (Hopf-WZ-term) of the
single M5-brane (see \eqref{IntroductionWZTermWithLevel} below
for multiple M5-branes), originally proposed in \cite[p. 10]{Aharony96},
fully established by \cite[(1)]{BLNPST97},
is a functional of fields on a 6d \emph{worldvolume} manifold $\Sigma^6$ that may be expressed
in terms of auxiliary extended fields on a cobounding {\it extended worldvolume} manifold
$\widehat \Sigma^{\, 7}$, as follows
(full details and generality below in \cref{The6dWZWTerm}, see \hyperlink{TableA}{Table A}):

\vspace{.14cm}
\begin{tabular}{cc}
\hspace{-1cm}
\begin{minipage}[l]{8.4cm}
\begin{align}
  \label{IntroductionWZTerm}
  \underset{
    \mathclap{
    \raisebox{-3pt}{
      \tiny
      \color{darkblue}
      \bf
      Hopf-WZ term
    }
    }
  }{
    \widehat {S}^{\; 1 \,\rm M5}_{\mathrm{WZ}}=2\widehat {S}^{\; \rm M5}_{\mathrm{WZ}}
  }
  \;:=\;
  2
  \underset{\widehat \Sigma^{\, 7}}{\int}
  \Big(
    \tfrac{1}{2}
   \widehat H_3 \wedge \widehat f^\ast \widetilde G_4
   \;+\;
   \widehat f^\ast G_7
  \Big)
  \\
  \label{ConsistencyCondition}
  \exp
  \Big(
    2\pi i
    \big(
       \widehat S^{\; 1\, \rm  M5}_{\mathrm{WZ}}
    \big)
  \Big)
  \;\in\; U(1)
\end{align}
\end{minipage}
&
  \hspace{1.8cm}
  {\footnotesize
    \def\arraystretch{1.3}
    \begin{tabular}{|l||l|}
      \hline
      $\widehat \Sigma^{\, 7}$ & Extended worldvolume
      \\
      \hline
      $\widehat f$ & Extended sigma-model field
      \\
      \hline
      $\widehat H_3$ & Extended worldvolume higher gauge field
      \\
      \hline
      $\widetilde G_4$ & Shifted background C-field flux
      \\
      \hline
      $G_7$ & Dual background C-field flux
      \\
      \hline
    \end{tabular}
    }
\end{tabular}

\vspace{.2cm}
\noindent
The open problem is to show that this expression
\eqref{IntroductionWZTerm} is actually well-defined,
in that it is independent of the choice of extensions,
or at least independent up to integer shifts, so that
at least the exponentiated Wess-Zumino action functional
\eqref{ConsistencyCondition} is well-defined.

\medskip
This means, equivalently, that
the difference of \eqref{IntroductionWZTerm}
over any two extensions, hence the integral
as in \eqref{IntroductionWZTerm} over
{\it closed} 7-manifolds $\widetilde \Sigma^7$, is an integer.
In supergravity, this closed integral is also known as
the {\it Page charge} \cite[(8)]{Page83}\cite[(43)]{DuffStelle91}\cite[(1.5.2)]{BLMP12},
and a proof of its integrality
would be a proof of {\it Page charge quantization}
(full details and generality are given below in Def. \ref{TheAnomalyFunctional}).
In this form, the problem was raised in \cite{Moore05},
but the details of the solution were left open.

\medskip

\noindent {\bf Partial solution in the literature.}
A suggestive partial solution to this problem was proposed in \cite{Intriligator00}, by
\begin{enumerate}[{\bf (i)}]
\vspace{-2.9mm}
\hypertarget{Assumption1}{}
\item   assuming that $G_4$ is not only the form datum underlying
a topological cocycle in rational Cohomotopy, \footnote{This is an after the fact
statement, in that we recast it that way in our formulation.
In \cite{Intriligator00}, the 4-sphere was part of spacetime,
whereas in our Cohomotopy
formulation it serves as a classifying space
that receives maps out of spacetime;
and we view the formulation in \cite{Intriligator00} as a special case where part of spacetime
is identified with that classifying space.
The sphere as a classifying space
for Cohomotopy cohomology theory
generalizes the Eilenberg-MacLane spaces
that classify ordinary cohomology.
See \cite[\S 2]{FSS20b} for these general concepts
and \cite{GS20} for a detailed comparison between cohomology and Cohomotopy.
}
but even that of an actual
smooth function $c_{\mathrm{smth}}$ to the smooth 4-sphere
\cite[(5.3)]{Intriligator00};

\vspace{-3mm}
\hypertarget{Assumption2}{}
\item
focusing on the first summand \cite[(2.4)]{Intriligator00}
and disregarding the second summand in \eqref{IntroductionWZTerm},
leaving its understanding for later \cite[top of p. 16]{Intriligator00}.
\end{enumerate}
\vspace{-2mm}

\noindent With these simplifications imposed, expression \eqref{IntroductionWZTerm} reduces
on oriented difference manifolds $\widetilde \Sigma^7 := \widehat \Sigma^{\, 7}_1 - \widehat \Sigma^{\, 7}_2$ \eqref{DifferenceOriented}
to the classical Whitehead integral formula \cite{Whitehead47} (see \cite[Prop. 17.22]{BT82})
for the {\it Hopf invariant} $\mathrm{HI}( c_{\mathrm{smth}} \circ \widetilde f )$
of maps to the 4-sphere
(recalled as Def. \ref{TheHopfInvariant} below).
Since the Hopf invariant is an integer by its homotopy-theoretic definition,
\cite{Intriligator00} suggests that \eqref{ConsistencyCondition}
is satisfied and thus refers to the first summand in \eqref{IntroductionWZTerm} as the \emph{Hopf-Wess-Zumino term},
a terminology
that was used for other sigma-models before \cite{WilczekZee83}\cite{TN},
and which has become widely adopted for the M5 since
(e.g.
\cite[\S 3.2]{KalkkinenStelle03}\cite[(2)]{HuNanopoulos11}\cite[(4.13)]{Arvanitakis18}).
But, since assumption \hyperlink{Assumption1}{\bf (i)} is not
supposed to be generally satisfied, so that
disregarding  the second term \hyperlink{Assumption2}{\bf  (ii)} is not
generally possible, this is only a partial solution, and
the full problem of showing general consistency of \eqref{IntroductionWZTerm}
by demonstrating \eqref{ConsistencyCondition} had remained open.

\medskip

\noindent {\bf Solution by homotopy periods in Cohomotopy.}
We observe here
that the full Hopf-WZ-term \eqref{IntroductionWZTerm},
including the previously neglected summand $\widehat f^\ast G_7$,
has the form of a
secondary characteristic class descending from the intersection pairing,
originally called a ``functional cup product'' by Steenrod \cite{Steenrod49}
and more recently discussed under the name
{\it homotopy period} in \cite[Ex. 1.9]{SinhaWalter08}.
That these {\it homotopy Whitehead integrals}
are the proper homotopy-theoretic formulation of the
original Whitehead integral formula \cite{Whitehead47}
for the Hopf invariant,
was remarked already by Haefliger \cite[p. 17]{Haefliger78}.
For the analogous lower-dimensional case of
maps from the 3-sphere to the 2-sphere,
this had been worked out in \cite[\S 14.5]{GrMo13}.

\medskip
Our {\bf first main result} here (Theorem \ref{AnomalyFunctionalAsLiftInCohomotopy} below)
is a transparent proof that the full 6d Hopf-WZ-term \eqref{IntroductionWZTerm}
(including both summands) is a homotopy period/homotopy Whitehead integral in this sense, which
reduces to the Whitehead integral formula for the Hopf invariant in the respective special cases
(Remark \ref{WhiteheadIntegralFormulasInTheLiterature} below).
In fact, we prove a more general
twisted version of the homotopy Whitehead integral,
which incorporates also the topological twists that
account for the half-integral shift by $\tfrac{1}{4}p_1$ demanded by flux quantization
of the background C-field (Remark \ref{ShiftedFluxQuantizationCorrectionsToTheHopfWZTerm} below)
thus generalizing the 6d Hopf-WZ term \eqref{IntroductionWZTerm}
to curved backgrounds (Def. \ref{DefinitionOfWZWTermByFieldExtensionToCoboundary} below).

\medskip
This shows, in particular, that the two summands in \eqref{IntroductionWZTerm} can not be invariantly
separated,  and hence that it is really the full term \eqref{IntroductionWZTerm}  which deserves
to be called the \emph{Hopf-Wess-Zumino term}. Thereby the puzzlement expressed in \cite[top of p. 8]{Intriligator00} is resolved:
The first summand of \eqref{IntroductionWZTerm} by itself does not actually qualify as a Wess-Zumino term,
since it is not (the pullback of) a cocycle.
The full term \emph{is} a cocycle
(Remark \ref{TheHopfWZTermOverBFivebrane}), and in fact a cocycle in integral
cohomology if \hyperlink{HypothesisH}{\it Hypothesis H} is satisfied, by the proof of our second main result:

\medskip
Our {\bf second main result} (Theorem \ref{AnomalyIsIntegral} below)
shows that under \hyperlink{HypothesisH}{\it Hypothesis H}
the 6d Hopf-Wess-Zumino term \eqref{IntroductionWZTerm} is generally integral,
even in its topologically twisted generalization.
This topologically twisted/parametrized generalization of the
Hopf invariant
thus establishes \eqref{ConsistencyCondition} and hence
proves in generality that the 6d Hopf-Wess-Zumino term of the
M5-brane is well-defined (namely integral, level-quantized);
see \eqref{SummaryDiagram} for the conclusion.

\medskip

\noindent {\bf Consequences.} We briefly highlight
some consequences of and conclusions drawn from this result:

\vspace{-.4cm}
\paragraph{1. Level quantization.}
A key argument of  \cite[(2.8)]{Intriligator00} was that the mathematical incarnation
of $N$ coinciding M5-branes is in the bare Hopf-WZ term \eqref{IntroductionWZTerm}
$\widehat S_{\rm WZ}^{\; \rm M5}
  =
 \int \tfrac{1}{2} H_3 \wedge \widehat f^\ast G_4 + \cdots$
being multiplied by $N(N+1)$, at least in its first summand.
Since, by our result here, the two summands cannot be invariantly separated,
this means that the full term has to be multiplied this way,
hence that for $N$ coincident M5-branes the expression \eqref{IntroductionWZTerm}
generalizes to
\vspace{-2mm}
\begin{equation}
  \label{IntroductionWZTermWithLevel}
  \widehat {S}^{\; N \,\rm M5}_{\mathrm{WZ}}
  \;:=\;
  N(N+1)
  \underset{\widehat \Sigma^{\, 7}}{\int}
  \Big(
    \tfrac{1}{2}
   \widehat H_3 \wedge \widehat f^\ast \widetilde G_4
   \;+\;
   \widehat f^\ast G_7
  \Big)
  \phantom{AAAAAA}
  \mbox{
    \begin{tabular}{|l||l|}
      \hline
      $N$ & Number of coincident M5-branes
      \\
      \hline
    \end{tabular}
  }
\end{equation}

\vspace{-1mm}
\noindent
with the factor of 2 in \eqref{IntroductionWZTerm} being the case of $N = 1$.
Since $N(N+1)$ is even for all $N$, the condition that \eqref{ConsistencyCondition}
is well-defined {up to an integral shift}
(by Theorems \ref{AnomalyFunctionalAsLiftInCohomotopy} and
\ref{AnomalyIsIntegral}) implies that
\vspace{-3mm}
\begin{equation}
  \label{ConsistencyConditionWithLevel}
  \exp
  \Big(
    2\pi i
    \big(
       \widehat S^{\; N \, \rm M5}_{\mathrm{WZ}}
    \big)
  \Big)
  \;\in\; U(1)
\end{equation}

\vspace{-2mm}
\noindent
is also well-defined, for all $N$.
Thus the factor $N(N+1)$
plays the role of the \emph{level} of the 6d Wess-Zumino term
of the M5-brane;
and its even integral form is the \emph{level quantization}
for the 6d Hopf-Wess-Zumino term of the M5-brane,
in analogy with integral levels of
ordinary Wess-Zumino terms \cite{Witten83}.

\vspace{-3mm}
\paragraph{2. Dimensional generalization and the Hopf invariant one theorem.}
The full 6d Wess-Zumino term of the M5-brane \eqref{IntroductionWZTerm}
is evidently the special case $k = 1$ of a sequence of Wess-Zumino terms
$S^{ 1 \, B(4k+1)}_{\mathrm{WZ}}$ that exist for all $k \in \mathbb{N}$
on higher gauged $p$-brane sigma-model fields with $p = 4k+1$, hence the
notation $B(4k+1)$. It is precisely these worldvolume dimensions that admit self-dual higher gauge fields.
For trivial topological twist $\tau$ in \eqref{GeneralBackgroundFieldDataAsCohomotopyCocycle},
the proof of Theorem \ref{AnomalyFunctionalAsLiftInCohomotopy}
generalizes verbatim to this infinite hierarchy, simply by generalizing the degree of the generator
$\omega_4$ in \eqref{GeneralBackgroundFieldDataAsCohomotopyCocycle} to $2(k+1)$
and the degree of the generator $\omega_{\, 7}$  to $4k+3$. Similarly, Prop.
\ref{GWIIsHomotopyInvariantOfBackgroundFieldData} generalizes verbatim and shows that
for all $k \in \mathbb{N}$
the anomaly functionals $\widetilde S^{\; 1 \, B(4k+1)}_{\mathrm{WZ}}$
(Def. \ref{TheAnomalyFunctional})
of these Wess-Zumino terms compute, in the absence of topological twists and under \hyperlink{HypothesisH}{\it Hypothesis H},
the Hopf invariant of the composite of the brane's sigma-model field with the cocycle of the background field in Cohomotopy.

\medskip
It is interesting to note that, from this perspective,
we may take the classical \emph{Hopf invariant one theorem}
\cite{Adams60}
to say that if the oriented difference of extended worldvolumes
is the $(4k+1)$-sphere $\widetilde \Sigma^{4k+1} = S^{4k+1}$,
then for almost all values of $k \in \mathbb{N}$
the anomaly functional $\widetilde S^{\; 1 B(4k+1)}_{\mathrm{WZ}}$
(Def. \ref{TheAnomalyFunctional})
is an even integer,
in that the only values of $k$
for which it may take odd integer values
are precisely those that correspond to branes which actually appear in
string/M-theory:
\vspace{-2mm}
\begin{center}
\begin{tabular}{|c||c|c|c|}
  \hline
  $k =$ & $0$ & $1$ & $2$
  \\
  \hline
  \hspace{-.1cm}$(4k+1)$-brane\hspace{-.1cm}
    &
  \hspace{-.1cm}string\hspace{-.1cm}
    &
  \hspace{-.1cm}five-brane\hspace{-.1cm}
    &
  \hspace{-.1cm}nine-brane\hspace{-.1cm}
  \\
  \hline
\end{tabular}
\hspace{.2cm}
\begin{minipage}[l]{9.7cm}
  {\bf \footnotesize \hyperlink{HypothesisH}{\it Hypothesis H} with the
  \emph{Hopf invariant one theorem}} \footnotesize singles out
  the worldvolume dimensions $p+1 \in \{2,6,10\}$
  among $p$-branes admitting self-dual higher gauge fields,
  as those whose
  Wess-Zumino
  anomaly functional $\widetilde S^{\; 1 \, B (4k+1)}_{\mathrm{WZ}}$
  is integrally \emph{in}divisible.
\end{minipage}
\end{center}

\paragraph{3. Unifying role of the quaternionic Hopf fibration.}
It is noteworthy that the proofs of our main results
(Theorem \ref{AnomalyFunctionalAsLiftInCohomotopy} and Theorem \ref{AnomalyIsIntegral})
proceed entirely by characterizing lifts in Cohomotopy through the quaternionic Hopf fibration,
observing that it is such lifts which reflect, under \hyperlink{HypothesisH}{\it Hypothesis H},
the higher gauge field $H_3$ on the worldvolume of the M5-brane \cite[Prop. 3.20]{FSS19b}.
This tightly connects the discussion of the
6d Wess-Zumino term here to
the analogous cohomotopical discussion of its supersymmetric completion in  \cite{FSS15}\cite{FSS19d}
and to the anomaly cancellation conditions on the background fields
in \cite{FSS19b}\cite{SS20a}, all rigorously derived from first principles;
and thus suggests that a complete derivation of the
elusive quantum M5-brane may exist guided by
\hyperlink{HypothesisH}{\it Hypothesis H}.

{
\paragraph{4. Outlook -- Refinement to differential Cohomotopy.}
It is well known
(see \cite{FSS12b}, review in \cite{FSS13a})
that the definition of Wess-Zumino- and Chern-Simons-terms by field extensions over a cobounding manifold,
while an elegant method when
it applies, is not the most general definition of these terms. In cases
where such field extensions do not exist, the WZ- and CS-terms may
still exist, now defined as hypervolume holonomies of cocycles in
a \emph{differential} cohomology theory (see \cite[\S 4.3]{FSS20b}).
For the ordinary WZ- and CS-term this differential cohomology theory is differential
ordinary cohomology, represented equivalently as Cheeger-Simons
differential characters or as Deligne cohomology or as bundle gerbes with connections, or as $B^n U(1)$-principal connections.

\medskip
But for the case of the 6dWZW term of the M5-branes,
our results here show
that the appropriate differential cohomology theory that
generalizes the construction by field extension presented here must be
a differential refinement of Cohomotopy cohomology theory. We had
constructed one version of such a \emph{differential Cohomotopy cohomology theory} in
\cite[\S 4]{FSS15}\cite[\S 5.3]{FSS20b}, further discussed in \cite[\S 3]{GS20}.
Ultimately one should use
this
to generalize the results we present
here to situations where extensions of fields over cobounding manifolds
may not exist.}

\medskip
\noindent {\bf Outline.}
In \cref{The6dWZWTerm} we make precise the
6d Hopf-Wess-Zumino term and its anomaly, including topological twisting.
In \cref{InTermsOfCohomotopy} we establish that the full WZ term is
a homotopy period/homotopy Whitehead integral.
In \cref{HypothesisHImpliesAnomalyCancellation} we prove
that {\it Hypothesis H} implies that
the full 6d Hopf-Wess-Zumino term is well-defined.

\section{The full 6d Hopf-WZ term of the M5-brane}
 \label{The6dWZWTerm}

In this section we present a precise definition, paraphrasing from the informal literature,
of the 6d Hopf-Wess-Zumino term of the M5-brane, refine it to include topological
twists reflecting the shifted quantization condition on the C-field flux, and then prove
that the corresponding anomaly functional is a homotopy invariant.

 \medskip
First we state (in Def. \ref{LocalDefinitionOfWZWTerm} below) the
6d WZ term for ``small''
sigma-model fields as found in the original articles \cite[p. 11]{Aharony96}\cite[(1)]{BLNPST97},
then we consider its globalization via extension to cobounding extended worldvolumes as in
\cite[(5.4)]{Intriligator00} (Def. \ref{DefinitionOfWZWTermByFieldExtensionToCoboundary} below).
Throughout, we include the half-integral shift of $G_4$ by $\tfrac{1}{4}p_1$,
demanded by the flux quantization of the C-field \cite{Witten97a};
see Remark \ref{ShiftedFluxQuantizationCorrectionsToTheHopfWZTerm} below.
Finally, we discuss
the corresponding anomaly functional (Def. \ref{TheAnomalyFunctional} below) and show that
it is a homotopy invariant on the space of gauged sigma-model fields
(Lemma \ref{AnomalyFunctionalIsHomotopyInvariant}).

\medskip

To be precise, we begin by introducing the relevant ingredients:

 \newpage

\begin{defn}[Background C-field and higher gauged sigma-model fields]
  \label{BasicSetup}
   $\,$
   \vspace{-1mm}
   \item {\bf (i)}
   Let  $X^8
   $
   be a smooth 8-manifold \hspace{-1pt}which \hspace{-1pt}is \hspace{-1pt}connected,
   \hspace{-1pt}simply
   \hspace{-1pt}connected\footnote{
         All results in the following readily generalize
         to non-connected $X$, but nothing essential is
         gained thereby. The assumption that $X$ is simply connected
         is to allow the use of Sullivan model analysis
         in \cref{InTermsOfCohomotopy} and \cref{HypothesisHImpliesAnomalyCancellation}
         (as in {\cite[Rem 2.6]{FSS19b}\cite[Rem. 3.53]{FSS20b}}).
         For this it would be sufficient to assume that $X$ is
         \emph{nilpotent} \cite[Def. 3.52]{FSS20b} in that it has nilpotent fundamental group
         acting nilpotently on homotopy and homology groups
         of its universal cover. This assumption should not
         be necessary, but without it all proofs
         will become much more involved.
       } and spin, \hspace{-1pt}to \hspace{-1pt}be \hspace{-1pt}called \hspace{-1pt}the \hspace{-1pt}\emph{target spacetime}\footnote{
         This pertains to M-theory on 8-manifolds, see
         \cite[Remark 3.1]{FSS19b}. We will often just write $X$ for $X^8$.
       }.

          \vspace{-1mm}
   \item {\bf (ii)}
   Let  $\Sigma$ be a smooth manifold,
        which is compact and oriented, to be called
     \begin{enumerate}[{\bf (a)}]
     \vspace{-2mm}
     \item the \emph{worldvolume} if it is 6-dimensional
        $\Sigma := \Sigma^6$ without boundary;
      \vspace{-2mm}
       \item the \emph{extended worldvolume} if it is 7-dimensional
           $\Sigma := \widehat \Sigma^{\, 7}$,
         with collared boundary
         \vspace{-2mm}
         \begin{equation}
           \label{Collar}
           \xymatrix{
             \Sigma^6
             =
             \partial \widehat \Sigma^{\, 7}
            \; \ar@{^{(}->}[rr]^-{ (\mathrm{id},0) }
             &&
             \big(\partial \widehat \Sigma^{\, 7}\big)
             \times [0,1)
            \; \ar@{^{(}->}[r]
             &
             \widehat \Sigma^{\, 7}.
           }
         \end{equation}
       \vspace{-9mm}
       \item the \emph{oriented difference of extended worldvolumes}
         if it is 7-dimensional
         $\Sigma := \widetilde \Sigma^{\, 7}$ and
         arising as the oriented difference
         \begin{equation}
           \label{DifferenceOriented}
           \widetilde \Sigma^{\, 7}
           \;=\;
           \widehat \Sigma^{\, 7}_1
           -
           \widehat \Sigma^{\, 7}_2
           \;:=\;
           \widehat \Sigma^{\, 7}_1
           \cup_{\Sigma^6}
           \big(\, \widehat \Sigma^{\, 7}_2 \big)^{\mathrm{op}}
         \end{equation}
         (where $(-)^{\mathrm{op}}$ denotes orientation reversal)
         of two collared coboundary extension $\widehat \Sigma^{\, 7}_{1,2}$
         \eqref{Collar}
         of the same worldvolume
         $\partial \widehat \Sigma^{\, 7}_{1,2} = \Sigma^6$; in particular $\widetilde \Sigma^{\, 7}$ has no boundary.
     \end{enumerate}

   \item{\bf (iii)}
   A \emph{background field} configuration on $X^8$
   is
  \begin{enumerate}[{\bf (a)}]
  \vspace{-2mm}
   \item
   an affine $\mathrm{Spin}(8)$-connection $\nabla$ on
   the tangent bundle\footnote{
     The theorems below hold, as general statements
     about the 6d WZ term, for $\nabla$ a
     connection on any $\mathrm{Spin}$ bundle. But application to the
     actual M5-brane system requires $\nabla$ to be a tangent
     connection on spacetime.
   }
   $T X^8$;
   \vspace{-3mm}
   \item
   a pair of differential forms
    \vspace{-2mm}
  \begin{equation}
    \label{SpacetimeFormAssumption}
    \begin{aligned}
      \phantom{2}G_4 & \in\;  \Omega_{\mathrm{dR}}^4\big( X^8 \big)
      \\
      2G_7 & \in\; \Omega_{\mathrm{dR}}^7\big( X^8 \big)
    \end{aligned}
    \phantom{AAA}
    \mbox{such that}
    \phantom{AAA}
    \begin{aligned}
      d \; \phantom{2}G_4 & = 0\,,
      \\
      d \; 2G_7
        & =
        - G_4 \wedge G_4
        + \big( \tfrac{1}{4}p_1(\nabla) \big)
          \wedge
          \big( \tfrac{1}{4}p_1(\nabla) \big)
          \,,
    \end{aligned}
  \end{equation}
   \vspace{-2mm}
\noindent  where the Pontrjagin 4-form (e.g. \cite[\S XII.4]{KobayashiNomizu63})
 \vspace{0mm}
  \begin{equation}
    \label{Pontrjagin4Form}
    p_1(\nabla)
      :=
    \langle R_{\nabla} \wedge R_\nabla \rangle
      \end{equation}

   \vspace{-3mm}
\noindent
is the value of the
  curvature 2-form $R$ of $\nabla$ in the normalized
  Killing form invariant polynomial
  $\langle -,-\rangle$ on $\mathfrak{so}(8)$.
  Notice that in terms of the shifted flux form \cite[(1.2)]{Witten97a}\cite[(1.2)]{Witten97b}\cite[\S 3.4]{FSS19b}\cite{FSS20a}
  \begin{equation}
    \label{G4Plus}
    \widetilde G_4 \;:=\; G_4 + \tfrac{1}{4}p_1(\nabla)
  \end{equation}
  the second condition in
  \eqref{SpacetimeFormAssumption}
  \cite[(3.6)]{Witten97b}\cite[\S 3.3]{FSS19b}
  equivalently reads
  \begin{equation}
    \label{SecondVersionOfdG7Equation}
       d\,2 G_7
       =
      -
      \big(
        \widetilde G_4 \wedge \widetilde G_4
        -
        \tfrac{1}{2}p_1(\nabla) \wedge \widetilde G_4
      \big).
  \end{equation}
  \end{enumerate}

   \vspace{-3mm}
  \item {\bf (iv)}
    A \emph{higher gauged} \footnote{
    This is the higher analog of abelian gauging of 2d WZW
    model fields (e.g. \cite[(5)]{Forste03}),
    making the 6d Wess-Zumino term the action functional of
    a \emph{higher gauged Wess-Zumino model} \cite{FSS13}.
    }
    \emph{sigma-model field} is a pair
 \vspace{-1mm}
    \begin{equation}
      \label{GaugedExtendedSigmaModelFields}
      \big(
        f,
        \,
        H_3
      \big)
      \;=\;
      \Big(
        \,
        \Sigma \xrightarrow{\;f\;{\color{darkblue}\mathrm{smooth}}} X,
        \;\;
        d H_3
        \;=\;
        f^\ast
        \big(
          G_4 - \tfrac{1}{4} p_1(\nabla)
      \big)
      \Big)
    \end{equation}

 \vspace{-2mm}
\noindent
    consisting of

     \vspace{-2mm}
    \begin{enumerate}[{\bf (a)}]
     \item a smooth function $f$ from the {(extended)} worldvolume
      to spacetime,

\vspace{-2mm}
    \item a smooth differential 3-form
    $H_3$ on the {(extended)} worldvolume, which trivializes
    the pullback along $f$ of the difference between $G_4$ from \eqref{SpacetimeFormAssumption}
    and $\tfrac{1}{4}p_1(\nabla)$ from \eqref{Pontrjagin4Form},
   \end{enumerate}
   \vspace{-2mm}
   both required to have \emph{sitting instants} on any
   collared boundary \eqref{Collar},
   in that in some neighborhood of the boundary
     they are constant in the direction
      perpendicular to it \cite[Def. 4.2.1]{FSS10};

\item {\bf (v)}
  A \emph{gauge transformation} or \emph{homotopy}
  between two higher gauged sigma-model fields \eqref{GaugedExtendedSigmaModelFields}
 \vspace{-2mm}
  \begin{equation}
    \label{GaugedHomotopy}
    \xymatrix{
      \big( f_0,\; (H_3)_0 \big)
      \ar@{=>}[rr]^-{ \left( \eta,\; (H_3)_{[0,1]}\right) }
      &&
      \big( f_1,\; (H_3)_1 \big)
    }
  \end{equation}

 \vspace{-2mm}
\noindent
  is a pair consisting
  of a smooth homotopy $\eta$ from $f_0$ to $f_1$
  and a differential 3-form $(H_3)_{[0,1]} \in \Omega_{\mathrm{dR}}^3\big( \Sigma \times [0,1]\big)$ gauging $\eta$ and restricting to $(H_3)_{0,1}$
  at the boundaries of the interval:
   \vspace{-2mm}
  \begin{equation}
    \label{DataHomotopyGauged}
    \raisebox{47pt}{
    \xymatrix@C=4em{
      \Sigma
      \ar[d]_-{ (\mathrm{id}, 0) }
      \ar[drr]^-{ f_0 }
      &&&
      (\widetilde H_3)_0
      \\
      \Sigma \times [0,1]
      \ar[rr]^<<<<<<<<<{ \eta \;\, {\color{darkblue}\mathrm{smooth}} }
      &&
      X
      &
      (H_3)_{[0,1]}
      \ar@{|->}[u]_-{ (\mathrm{id},0)^\ast }
      \ar@{|->}[d]^-{ (\mathrm{id},1)^\ast }
      &
      d
      \big(
      (H_3)_{[0,1]}
      \big)
      =
      \eta^\ast
      \big(
        G_4
        -
        \tfrac{1}{4}p_1(\nabla)
      \big).
      \\
      \Sigma
      \ar[u]^-{ (\mathrm{id}, 1) }
      \ar[urr]_-{ f_1 }
      &&&
      (H_3)_1
    }
    }
  \end{equation}

 \vspace{-2mm}
  \item {\bf (vi)}
  We write
 \vspace{-.8cm}
  \begin{equation}
    \label{GaugeableFunctions}
    \mathrm{Maps}_{\mathrm{smth}}^{\mathrm{ggd}}
    (\Sigma,X)
    \;:=\;
    \big\{
      \big( f, H_3 \big)
    \big\}
    \,,
    \phantom{AAAA}
    \pi_0
    \big(
      \mathrm{Maps}_{\mathrm{smth}}^{\mathrm{ggd}}
      (\Sigma,X)
    \big)
    \;:=\;
    \big\{
      \big( f, H_3 \big)
    \big\}_{\big/\sim_{\mathrm{homotopy}}}
      \end{equation}

       \vspace{-2mm}
\noindent
  for the sets\footnote{
    The inclined reader will notice
    (see \cite{FSS13} for exposition) that
    the set $\mathrm{Maps}_{\mathrm{smth}}^{\mathrm{ggd}}(\Sigma,X)$
    is, of course, the underlying set of global sections of
    the atlas for the smooth \emph{moduli 2-stack} of higher gauged sigma-model fields on $\Sigma$
    \cite[\S 4.3]{FSS20b}, and
    $\pi_0\big(\mathrm{Maps}_{\mathrm{smth}}^{\mathrm{ggd}}(\Sigma,X)\big)$
    is the set of connected components of the geometric realization
    of this moduli 2-stack. All of the
    following discussion lifts to the higher differential geometry
    of moduli stacks of fields, but for the sake of brevity we
    will not further consider this here.
  }
  of
  higher gauged sigma-model fields \eqref{GaugedExtendedSigmaModelFields}
  and of their \emph{homotopy classes} \eqref{GaugedHomotopy},
  respectively.
\end{defn}

\begin{remark}[Shifted flux quantization corrections to the Hopf-WZ term]
  \label{ShiftedFluxQuantizationCorrectionsToTheHopfWZTerm}
  The existing literature on the Hopf-WZ term
  \cite[p. 11]{Aharony96}\cite[(1)]{BLNPST97}\cite[(2.4)]{Intriligator00}
  \cite[\S 3.2]{KalkkinenStelle03}\cite[(2)]{HuNanopoulos11}\cite[(4.13)]{Arvanitakis18}
  disregards any topological correction terms
  to the C-field flux proportional to $p_1(\nabla)$,
  shown in \eqref{G4Plus}, \eqref{SecondVersionOfdG7Equation}.
  Hence, in comparing to this literature
  (as in Def. \ref{LocalDefinitionOfWZWTerm}, Def. \ref{DefinitionOfWZWTermByFieldExtensionToCoboundary} below),
  one has to restrict to the
  special case that $p_1(\nabla) = 0$ (for instance in that $\nabla = 0$,
  hence that spacetime is assumed to be flat).
  Beyond this special case,
  such correction terms are famously thought to be required, by
  circumstantial arguments provided in \cite{Witten97a}\cite{Witten97b}.
  The result of \cite[\S 3.3]{FSS19b} (recalled as Remark \ref{BackgroundFieldsAsRationalCohomotopy} below)
  is that charge quantization of the C-field in J-twisted Cohomotopy theory
  \cite[\S 5.3]{FSS20b} implies exactly these corrections
  \eqref{SecondVersionOfdG7Equation}; and the main Theorem
  \ref{AnomalyIsIntegral} below says that with these $p_1$-corrections
  and this cohomotopical charge quantization,
  the full Hopf-WZ term is actually guaranteed to be anomaly-free (namely: integral, level-quantized).
\end{remark}

As an important example of Def. \ref{BasicSetup}, we offer the following:

\begin{lemma}[The 7-sphere as an extended worldvolume]
  \label{GaugebleSigmaModelMapsOn7Sphere}
  In the situation of Def. \ref{BasicSetup},
  let the oriented difference of extended worldvolumes
  be the 7-sphere:
  $\Sigma := \widetilde \Sigma^7 := S^7$.
  Then the set \eqref{GaugeableFunctions} of homotopy classes of extended gauged sigma-model fields
  is the set underlying the 7th homotopy group of
  target spacetime $X$:
  \begin{equation}
    \label{GaugeableFieldsOn7SphereFormpi7X}
    \pi_0
    \Big(
    \mathrm{Maps}_{\mathrm{smth}}^{\mathrm{ggd}}
    \big(
      S^7
      ,
      X
    \big)
    \Big)
    \;\simeq\;
    \pi_7(X)
    \,.
  \end{equation}
\end{lemma}
\begin{proof}
  Since homotopy classes of continuous
  functions
  between smooth manifolds
  are given by smooth homotopy classes of smooth functions
  (e.g. \cite[Cor. 17.8.1]{BT82}) it follows that
  already smooth homotopy classes of ungauged
  sigma-model fields are in bijection to $\pi_7(X)$
  (since the target spacetime $X$ is assumed to be connected
  there is no dependence on a basepoint).
  Hence
  it only remains to show that,
  for any extended sigma-model field
  $\widetilde f$, there exists at least one
  gauging $\widetilde H_3$ \eqref{GaugedExtendedSigmaModelFields}
  and that, for any two such gaugings
  $\big( \widetilde f, (\widetilde H_3)_0\big)$ and
  $\big( \widetilde f, (\widetilde H_3)_1\big)$
  of the same extended sigma-model field
  $\widetilde f$,
  there exists a gauged homotopy \eqref{GaugedHomotopy}
   \vspace{-2mm}
  $$
    \xymatrix{
      \big( \widetilde f, (\widetilde H_3)_0 \big)
      \ar@{=>}[rr]^{
        \left(
          \widetilde \eta, (\widetilde H_3)_{[0,1]}
        \right)
      }
      &&
      \big( \widetilde f, (\widetilde H_3)_1 \big)
    }
  $$

   \vspace{-2mm}
\noindent
  between them.
  {For the existence of the gauging $\widetilde{H}_3$ for a given $\tilde{f}$, we only
  need to notice that because $H^4_{\mathrm{dR}}(S^7)\cong H^4(S^7;\mathbb{R})=0$,
  we have $f^*[ G_4 - \tfrac{1}{4} p_1(\nabla)]=0$ and so there exists $\widetilde{H}_3\in \Omega_{\mathrm{dR}}^3\big( S^7\big)$ such that $d\widetilde{H}_3=f^*\left(G_4 - \tfrac{1}{4} p_1(\nabla)\right)$. Similarly, given two gaugings $(\widetilde H_3)_0$ and $(\widetilde H_3)_1$ of $\tilde{f}$,
  }
  since $H^3_{\mathrm{dR}}\big( S^7\big) = 0$
  and $(\widetilde H_3)_1 - (\widetilde H_3)_0 \in \Omega_{\mathrm{dR}}^3\big( S^7\big)$
  is closed by assumption, there exists
  $$
    \alpha \in \Omega_{\mathrm{dR}}^2\big( S^7 \big)
    \phantom{AA}
    \mbox{\rm such that}
    \phantom{AA}
    d \alpha = (\widetilde H_3)_1 - (\widetilde H_3)_0
    \,.
  $$
  Thus
  $$
    \Big(
      \widetilde \eta : (x,s) \longmapsto \widetilde f(x)
      \,,
      \;\,
      (\widetilde H_3)_{[0,1]}
      \;:=\;
      (\widetilde H_3)_1
      +
      (s-1) \cdot d \alpha
      +
      (d s) \wedge \alpha
    \Big)
  $$
  constitutes a homotopy as required.
\end{proof}

We now consider the 6d Hopf-WZ term in its various incarnations,
surveyed in \hyperlink{TableA}{Table A}.

\medskip
\medskip
\hypertarget{TableA}{}
\hspace{-1cm}
\begin{tabular}{ll}
\scalebox{.9}{
\begin{tabular}{|c|l|c|}
  \hline
  $
    \xymatrix{
      \mathrm{Maps}_{\mathrm{smth}}^{\mathrm{ggd}}
      \big(
        \Sigma^6,
        X^{11}
      \big)
      \ar[r]^-{
        S
      }
      &
      \mathbb{R}
    }
  $
  &
  \begingroup
\setlength{\tabcolsep}{6pt} 
\renewcommand{\arraystretch}{1} 
  \begin{tabular}{l}
    Hopf-WZ action functional
    \\
    on worldvolume
    \\
    $\Sigma^6$
  \end{tabular}
  \endgroup
  &
  Def. \ref{LocalDefinitionOfWZWTerm}
  \\
  \hline
  $
    \xymatrix{
      \mathrm{Maps}_{\mathrm{smth}}^{\mathrm{ggd}}
      \big(
        \widehat \Sigma^{\, 7},
        X^{11}
      \big)
      \ar[r]^-{
        \widehat S
      }
      &
      \mathbb{R}
    }
  $
  &
    \begingroup
\setlength{\tabcolsep}{6pt} 
\renewcommand{\arraystretch}{1} 
  \begin{tabular}{l}
    Extended Hopf-WZ functional
    \\
    on coboundary
    \hspace{-.2cm}
    \\
    $
      \partial \widehat \Sigma^{\, 7} = \Sigma^6
    $
  \end{tabular}
  \endgroup
  &
  Def. \ref{DefinitionOfWZWTermByFieldExtensionToCoboundary}
  \\
  \hline
  $
    \xymatrix{
      \mathrm{Maps}_{\mathrm{smth}}^{\mathrm{ggd}}
      \big(
        \widetilde \Sigma^{\, 7},
        X^{11}
      \big)
      \ar[r]^-{
        \widetilde S
      }
      &
      \mathbb{R}
    }
  $
  &
    \begingroup
\setlength{\tabcolsep}{6pt} 
\renewcommand{\arraystretch}{1} 
  \begin{tabular}{l}
    Hopf-WZ anomaly functional
    \\
    on oriented difference
    \\
    $
      \widetilde \Sigma^{\, 7}
      =
      \widehat \Sigma_1^{\, 7}
      -
      \widehat \Sigma_2^{\, 7}
    $
  \end{tabular}
  \endgroup
  &
  Def. \ref{TheAnomalyFunctional}
  \\
  \hline
\end{tabular}
}
&
\hspace{-.3cm}
\begin{minipage}[l]{6.1cm}
  {\bf \footnotesize Table A -- Incarnations of the Hopf-WZ term.}
  {\footnotesize
    The 6d Hopf-WZ term functional $S := S^{\rm M5}_{\mathrm{WZ}}$
    is a priori defined on gauged sigma-model fields on $\Sigma^6$.
    Its global definition involves an extension $\widehat S$ to
    extended fields
    on a coboundary $\widehat \Sigma^{\, 7}$.
    The difference of any two extensions is
    the anomaly functional $\widetilde S$
    on fields on the oriented difference
    $\widetilde \Sigma^{\, 7}= \widehat \Sigma^{\, 7}_1 - \widehat \Sigma^{\, 7}_2$.
  }
\end{minipage}
\end{tabular}

\medskip
\begin{defn}[6d Hopf-WZ term for small sigma-model fields]
\label{LocalDefinitionOfWZWTerm}
In the setting of Def. \ref{BasicSetup}, 
let $U \subset X^{8}$ be a chart (a contractible open subset).
For $\Sigma^6$ any closed orientable 6-manifold, write
$
  \mathrm{Maps}^{\mathrm{ggd}}_{\mathrm{smth}}
  \big(
    \Sigma^6,
    U
  \big)
  \;\subset\;
  \mathrm{Maps}^{\mathrm{ggd}}_{\mathrm{smth}}
  \big(
    \Sigma^6,
    X
  \big)
$
for the subset of those higher gauged sigma-model fields \eqref{GaugeableFunctions}
which factor through $U \subset X$ (the ``$U$-small sigma-model fields'').
{As the de Rham cohomology of $U$ is trivial in positive degree,
we may choose local potentials
 $C^{\, U}_3\in \Omega_{\mathrm{dR}}^3(U)$ for
 $\iota_U^\ast (G_4+\frac{1}{4}p_1(\nabla))$  and
 $2C^{\, U}_6 \in \Omega_{\mathrm{dR}}^6(U)$ for
 $\iota_U^\ast 2G_7 + C^{\, U}_3 \wedge \iota_U^\ast (G_4- \tfrac{1}{4}p_1(\nabla) \big)$.}
 \vspace{-2mm}
\begin{equation}
  \label{LocalPotentials}
  \hspace{-7cm}
  \xymatrix{
    &&
    U_{\phantom{A}}
    \ar@{^{(}->}[d]^{ \iota_U }
    &
    \mathrlap{\footnotesize
      \begin{array}{lcl}
        d\, C^{\, U}_3
        & = &
        \iota_U^\ast \big( G_4 + \tfrac{1}{4}p_1(\nabla) \big)
        \\
        d\,2C^{\, U}_6
        & = &
        \iota_U^\ast 2G_7
        +
        C^{\, U}_3 \wedge \iota_U^\ast \big(G_4{- \tfrac{1}{4}p_1(\nabla) \big)}
      \end{array}
    }
    \\
    \Sigma^6
    \ar[rr]_-{
      f
    }
    \ar@{-->}[urr]^-{ f_U }
    &&
    X^{8}
    &
    \mathrlap{\footnotesize
      \begin{array}{clc}
        d\,G_4 &=  \quad 0
        \\
        d\,2G_7 &= - G_4 \wedge G_4 {+ \big( \tfrac{1}{4}p_1(\nabla) \big)
          \wedge
          \big( \tfrac{1}{4}p_1(\nabla) \big)}
      \end{array}
    }
  }
\end{equation}
 Then the \emph{M5 6d Wess-Zumino term action functional} on these small fields is the function
\begin{equation}
  \label{Local6dWZWFunctional}
  \xymatrix@R=-6pt{
    \mathrm{Maps}^{\mathrm{ggd}}_{\mathrm{smth}}
    \big(
      \Sigma^6,
      U
    \big)
    \ar[rr]^-{ S^{\mathrm{M5}}_{\mathrm{WZ}} }
    &&
    \mathbb{R}
    \\
    \big(
      f_{U}, H_3
    \big)
    \ar@{}[rr]|-{\longmapsto}
    &&
    S^{\rm M5}_{\mathrm{WZ}}
    \big(
       {f_U},
       H_3
    \big)
    \;:=\;
    \tfrac{1}{2}
    \underset{\Sigma^6}{\bigintsss}
    \big(
      - H_3 \wedge f_U^\ast C_3^{\, U}
      +
      f_U^\ast 2C_6^{\, U}
    \big).
    }
\end{equation}
\end{defn}
{
\begin{lemma}[Independence of choices]
The functional $S^{\rm M5}_{\mathrm{WZ}}
    \big(f_U,
       H_3
    \big)$
    \eqref{Local6dWZWFunctional}
    is indeed well defined,
    in that it does not depend on the choice
    of the local potentials $C^{\, U}_3$ and $C^{\, U}_6$ \eqref{LocalPotentials}.
\end{lemma}
\begin{proof}
A different choice of local potentials is of the form
$
  (C_3^{\, U}+\alpha_3^U, 2C_6^{\, U}+2\alpha^U_6)
$,
with differentials
$
  d\alpha_U^3=0
$
and
$
  d 2\alpha_6^U
    =
  \alpha^U_3 \wedge \iota_U^\ast (G_4- \tfrac{1}{4}p_1(\nabla) \big)
$.
As the local chart $U$ is contractible, this implies
$\alpha_3^U = d\alpha_2^U$
and
$
  2\alpha_6^U
  =
  \alpha^U_2
    \wedge
  \iota_U^\ast
  \big(
    G_4- \tfrac{1}{4}p_1(\nabla)
  \big)
  +
  d\alpha_5^U
$.
Therefore, we have
\begin{align*}
  &
  \int_{\Sigma^6}
    \Big(
      -
      H_3
       \wedge
      f_U^\ast
      \big( C_3^{\, U} + \alpha_3^U \big)
      +
      2 f_U^\ast
      \big(
        C_6^{\, U} + \alpha_6^U
      \big)
    \Big)
    -
    \int_{\Sigma^6}
    \Big(
      - H_3 \wedge f_U^\ast \big( C_3^{\, U} \big)
      +
      f_U^\ast \big( 2 C_6^{\, U} \big)
    \Big)
    \\
    & =
    \int_{\Sigma^6}
      - H_3 \wedge d f_U^\ast\big( \alpha_2^U \big)
      +
      f_U^\ast
      \Big(
        \alpha^U_2
          \wedge
        \iota_U^\ast
        \big(
          G_4 - \tfrac{1}{4}p_1(\nabla)
        \big)
      \Big)
      +
      d f_U^*
      \big(
        \alpha_5^U
      \big)
      \\
      &=
      \int_{\Sigma^6}
      -
      H_3 \wedge d f_U^\ast \big( \alpha_2^U \big)
      +
      f_U^\ast \big( \alpha^U_2 \big)
        \wedge
      f^\ast
      \big(
        G_4- \tfrac{1}{4}p_1(\nabla)
      \big)
      +
      df_U^*
      \big(
        \alpha_5^U
      \big)
      \\
      &=
      \int_{\Sigma^6}
      - H_3 \wedge d f_U^\ast \big( \alpha_2^U \big)
      +
      f_U^\ast \big(\alpha^U_2 \big) \wedge dH_3
      +
      df_U^*\big( \alpha_5^U \big)
      \\
      &=
     \int_{\Sigma^6}
     d
     \Big(
     H_3 \wedge
     f_U^\ast \big( \alpha_2^U \big)
     +
     f_U^*\big(\alpha_5^U\big)
     \Big)
     =
     0.
\end{align*}

\vspace{-5mm}
\end{proof}
}

Now we globalize this definition,
following the well-known procedure originally introduced in the 2-dimensional case in \cite{Witten83}.

\begin{defn}[Global 6d Hopf-Wess-Zumino term via extended worldvolumes]
  \label{DefinitionOfWZWTermByFieldExtensionToCoboundary}
  In the situation of Def. \ref{BasicSetup},  for $\Sigma^6$ a given worldvolume,
  let $\widehat \Sigma^{\, 7}$ be a compact oriented smooth
  collared cobounding 7-manifold\footnote{This always exists, since the oriented cobordism ring in
  dimension 6 is trivial and by the collar neighbourhood theorem.} according to \eqref{Collar}
   \begin{equation}
    \label{ChosenCoboundary}
    \Sigma^6 := \partial \widehat \Sigma^{\, 7}\;.
  \end{equation}

   \vspace{-2mm}
\noindent
    Then we say that the corresponding \emph{extended action functional} for the
    6d Hopf-WZ-term on the closed manifold $\Sigma^6$ is the function
     \vspace{-2mm}
  \begin{equation}
    \label{TheGlobalActionFunctional}
    \xymatrix@R=-6pt{
      \mathrm{Maps}^{\mathrm{ggd}}_{\mathrm{smth}}
      \big(\,
        \widehat \Sigma^{\, 7},
        X
      \big)
      \ar[rr]^-{ \widehat S^{\; \rm M5}_{\mathrm{WZ}} }
      &&
      \mathbb{R}
      \\
      \big(
        \widehat f, \, \widehat H_3
      \big)
      \ar@{}[rr]|-{\longmapsto}
      &&
      \widehat S^{\; \rm M5}_{\mathrm{WZ}}
      \big(
        \widehat f, \widehat H_3
      \big)
      \;:=\;
      \tfrac{1}{2}
      \underset
      {\widehat \Sigma^{\, 7}}
      {\bigintsss}
      \Big(
        \widehat H_3
          \wedge
        {\widehat f}^\ast
        \big(
          G_4
          +
          \tfrac{1}{4}p_1(\nabla)
        \big)
        +
        \widehat f^\ast 2G_7
      \Big)
    }
  \end{equation}

   \vspace{-2mm}
\noindent
  on the set of extended gauged sigma-model fields
  \eqref{GaugedExtendedSigmaModelFields}.
  (For flat backgrounds, $\nabla = 0$, this reduces to
  the  \eqref{IntroductionWZTerm}, see Remark \ref{ShiftedFluxQuantizationCorrectionsToTheHopfWZTerm}.)
\end{defn}

\begin{lemma}[Global Hopf-WZ-term restricts to local Hopf-WZ-term]
\label{LocalWZWTermViaExtendedActionFunctional}
In the situation of Def. \ref{LocalDefinitionOfWZWTerm},
consider a worldvolume $\Sigma^6$.
Then, for every choice of extended worldvolume $\widehat \Sigma^{\, 7}$
\eqref{ChosenCoboundary}
the corresponding \emph{extended action functional}
$\widehat S$ (Def. \ref{DefinitionOfWZWTermByFieldExtensionToCoboundary})
coincides, for any chart $U \subset X$,
on $U$-small extended sigma-model fields
$\widehat f = \iota_U \circ \widehat f_U$ \eqref{LocalPotentials}
with the local action functional $S$ (Def. \ref{LocalDefinitionOfWZWTerm})
evaluated on the boundary values $f := \widehat f\;\vert_{\Sigma^6}$ of the extended fields:
  \vspace{-2mm}
$$
  \raisebox{20pt}{
  \xymatrix@C=3em{
    \mathrm{Maps}_{\mathrm{smth}}^{\mathrm{ggd}}
    \big(\,
      \widehat \Sigma^{\, 7},
      U
    \big)
    \ar[d]_{ (-)\vert_{\partial \widehat \Sigma^{\, 7}} }
    \ar[rr]^-{
      \widehat S^{\; \rm M5}_{\mathrm{WZ}}
    }
    &&
    \mathbb{R}
    \\
    \mathrm{Maps}_{\mathrm{smth}}^{\mathrm{ggd}}
    \big(
      \Sigma^6,
      U
    \big)
    \ar[urr]_-{ S^{\rm M5}_{\mathrm{WZ}} }
  }
  }
  \phantom{AAAAAAA}
  \widehat S^{\; \rm M5}_{\mathrm{WZ}}\big( \widehat f, \widehat H_3\big)
  \;=\;
  S^{\rm M5}_{\mathrm{WZ}}
  \Big(
    f := {\widehat f}\;\vert_{\partial \Sigma^7},
    \;
    H_3 := \big( \widehat H_3 \big)\vert_{\partial \widehat \Sigma^{\, 7}}
  \Big)
  \,.
$$
\end{lemma}
\begin{proof}
Observe, with \eqref{GaugedExtendedSigmaModelFields} and \eqref{LocalPotentials}, that
\begin{equation}
  \label{ComputingLocalExactness}
  \hspace{-1cm}
  \begin{aligned}
    d
    \big(
      -
      \widehat H_3 \wedge \widehat f_U^\ast C_3^{\, U}
      +
      \widehat f_U^\ast 2C_6^{\, U}
    \big)
    &
    \;=\;
      -
      \underset{
        \mathclap{
          \scalebox{.6}{$
          \widehat f^\ast
          \big(
            G_4 - \tfrac{1}{4}p_1(\nabla)
          \big)
          $}
        }
      }
      {
        \underbrace{
          d \widehat H_3
        }
      }
    \;
    \wedge
    \;
    \widehat f_U^\ast C_3^{\, U}
    \;+\;
    \widehat H_3
    \;\wedge\;
    \underset{
      \mathclap{
        \scalebox{.6}{$
          \widehat f^\ast (G_4{+\frac{1}{4}p_1(\nabla)})
        $}
        \;
      }
    }{
      \underbrace{
        f_U^\ast d C^{\, U}_3
      }
    }
    \;\;+\;\;
    \underset{
      \mathclap{
      \;
      \scalebox{.6}{$
        \begin{array}{l}
        \qquad \widehat f^\ast (G_4{-\frac{1}{4}p_1(\nabla)})
        \wedge
        \widehat f_U^\ast C_3^{\, U}
        \\ \qquad +
        \widehat f^\ast 2G_7
        \end{array}
      $}
      }
    }{
      \underbrace{
        \widehat f_U^\ast d 2C_6^{\, U}
      }
    }
    \\
    & =
      \widehat H_3
          \wedge
        {\widehat f}^\ast
        \big(
          G_4
          {+
          \tfrac{1}{4}p_1(\nabla)}
        \big)
    \;+\;
    \widehat f^\ast 2G_7\;.
  \end{aligned}
\end{equation}

  \vspace{-2mm}
\noindent
With this, the claim follows by Stokes' theorem:
  \vspace{-2mm}
\begin{equation}
  \begin{aligned}
    \widehat S^{\; \rm M5}_{\mathrm{WZ}}
    \big(
      \widehat f,\, \widehat H_3
    \big)
    & :=
    \tfrac{1}{2}
    \underset{\widehat \Sigma^7}{\int}
    \Big(
      \widehat H_3
          \wedge
        {\widehat f}^\ast
        \big(
          G_4
          {+
          \tfrac{1}{4}p_1(\nabla)}
        \big)
      +
      \widehat f_U^\ast 2G_7
    \Big)
    \\
    & =
    \tfrac{1}{2}
    \underset{\widehat \Sigma^7}{\int}
    d
    \big(
      -
      \widehat H_3 \wedge \widehat f_U^\ast C_3^{\, U}
      +
      \widehat f_U^\ast 2C_6^{\, U}
    \big)
    \\
    & =
    \tfrac{1}{2}
    \underset{\partial \widehat \Sigma^7}{\int}
    \big(
      -
      \widehat H_3 \wedge \widehat f_U^\ast C_3^{\, U}
      +
      \widehat f_U^\ast 2C_6^{\, U}
    \big)\vert_{\partial \Sigma^7}
    \\
    & =
    \tfrac{1}{2}
    \underset{\Sigma^6}{\int}
    \big(
      -
      H_3 \wedge f_U^\ast C_3^{\, U}
      +
      f_U^\ast 2C_6^{\, U}
    \big)
    \\
    & =:
    S^{\rm M5}_{\mathrm{WZ}}\big( 
    {f_U}
    ,\, H_3\big)
    \,.
  \end{aligned}
\end{equation}

\vspace{-7mm}
\end{proof}

\begin{example}[Coboundaries for $\Sigma^6 = S^3 \times S^3$]
\label{CoboundariesForS3xS3}
In the situation of Def. \ref{BasicSetup}, consider
as worldvolume the product manifold of two 3-spheres (this is considered in \cite[Example 2]{MS}
in the non-commutative setting):
  $
  \Sigma^6
  \;=\;
  S^3 \times S^3
 $.
In this case, there is a canonical choice of cobounding manifold
$\widehat \Sigma^{\, 7}$ \eqref{ChosenCoboundary}
given by the Cartesian product of
the 4-disk $D^4$ (the closed 4-dimensional ball) with the 3-sphere, in either order (as in \cite{Sati13}):
  \vspace{0mm}
\begin{equation}
  \label{CoboundaryForS3xS3}
  \widehat \Sigma^{\, 7}_L
  \;:=\;
  D^4 \times S^3
  \phantom{AA}\mbox{and}\phantom{AA}
  \widehat \Sigma^{\, 7}_R
  \;:=\;
  \big(
    S^3 \times D^4
  \big)^{\mathrm{op}}
  \,.
\end{equation}

  \vspace{-2mm}
\noindent
Here we are equipping each of
  \vspace{-2mm}
$$
  \left.
  \begin{array}{ccc}
    S^3 \times S^3,
    \\
    D^4 \times S^3 & = & \phantom{(\partial}D^4\phantom{)} \times \big( \partial D^4 \big),
    \\
    S^3 \times D^4 & = & \big(\partial D^4\big) \times \phantom{(\partial}D^4\phantom{)},
  \end{array}
  \right\}
  \;\subset\;
    D^4 \times D^4
  \;\subset\;
  \mathbb{R}^8
$$

  \vspace{-1mm}
\noindent
with the orientation induced from the canonical embedding into
$\mathbb{R}^8$, which implies, by the odd-dimensionality of
$S^3$, that the boundary of $S^3 \times D^4$ is
$\big( S^3 \times S^3 \big)^{\mathrm{op}}$ (opposite orientation).
This way, with \eqref{CoboundaryForS3xS3} we indeed have
  \vspace{-1mm}
$$
  \partial
  \widehat \Sigma^{\, 7}_{L,R}
  \;=\;
  \Sigma^6 := S^3 \times S^3
$$

  \vspace{-1mm}
\noindent
as oriented manifolds.
Observe that the union of one of these coboundaries with the orientation reversal of the other is
the 7-sphere (as considered in Lemma \ref{GaugebleSigmaModelMapsOn7Sphere}):
\footnote{Note that a different manipulation treats these, untraditionally, as manifolds with corners \cite{framed}\cite{Sati13}.}
\vspace{0mm}
\begin{equation}
  \begin{aligned}
    \widetilde \Sigma^{\, 7}
    \;:=\;
    \widehat \Sigma^{\, 7}_L
    \cup
    \big(\,
      \widehat \Sigma^{\, 7}_R
    \big)^{\mathrm{op}}
    & =
    D^4 \times \big(\partial D^4\big)
    \;\cup\;
    \big( \partial D^4 \big) \times D^4
    \\
    & =
    \partial
    \big(
      D^4 \times D^4
    \big)
    \\
    & \simeq
    \partial D^8
    \\
    & =
    S^7\;.
  \end{aligned}
  \phantom{AAAAAA}
  \raisebox{26pt}{
  \xymatrix{
    S^3 \times S^3
    \ar@{}[dr]|-{ \mbox{\tiny (po)} }
    \ar[d]
    \ar[r]
    &
    D^4 \times S^3
    \ar[d]
    \\
    S^3 \times D^4
    \ar[r] & S^7
  }
  }
\end{equation}
\end{example}

While Def. \ref{DefinitionOfWZWTermByFieldExtensionToCoboundary}
gives global meaning to the local Hopf-WZ term (Def. \ref{LocalDefinitionOfWZWTerm}),
by Lemma
\ref{LocalWZWTermViaExtendedActionFunctional},
this potentially comes
at the cost that the global definition depends on the choice of coboundary \eqref{ChosenCoboundary}. The following definition measures this
potential dependency:

\begin{defn}[Hopf-WZ anomaly functional/Page charge]
\label{TheAnomalyFunctional}
In the situation of Def. \ref{BasicSetup}, with given worldvolume $\Sigma^6$,
consider in Def. \ref{DefinitionOfWZWTermByFieldExtensionToCoboundary}
two choices $\widehat \Sigma_{L,R}^7$ of
collared cobounding extended worldvolumes \eqref{Collar}
$
  \partial \Sigma^7_{L,R}
  =
  \Sigma^6
$.
This makes their oriented difference
\eqref{DifferenceOriented}
a smooth closed 7-manifold
$
 \widetilde \Sigma^{\, 7}
  \;:=\;
  \widehat \Sigma^{\, 7}_L
  -
  \widehat \Sigma^{\, 7}_R
  \;:=\;
  \widehat \Sigma^{\, 7}_L
  \cup_{\Sigma^6}
  \big(\,
    \widehat \Sigma^{\, 7}_R
  \big)^{\mathrm{op}}
  \,.
$
Then for
$$
  \xymatrix@C=3em{
    \widehat \Sigma^{\, 7}_L
    \ar[drr]^-{ \widehat f_L }
    && &
    d \big( \widehat H_3 \big)_L = \widehat f_L^\ast G_4
    \ar@{|->}[d]^-{ \iota_{\partial L}^\ast }
    \\
    \Sigma^6_{\phantom{A}}
    \ar@{^{(}->}[u]^-{\iota_{\partial L}}
    \ar@{_{(}->}[d]_-{ \iota_{\partial R} }
    \ar[rr]|-{\, f  \,}
    &&
    X
    &
    d \, H_3 \,  = f^\ast G_4
    \\
    \widehat \Sigma^{\, 7}_R
    \ar[urr]_-{ \widehat f_R }
    && &
    d \big( \widehat H_3 \big)_R = \widehat f_R^\ast G_4
    \ar@{|->}[u]^-{ \iota_{\partial R}^\ast }
  }
$$
any pair of gauged extended sigma-model fields \eqref{GaugedExtendedSigmaModelFields},
extending the same ordinary sigma-model field $f$ over the
two choices of coboundaries, respectively,
we obtain a gauged extended sigma-model field
$\big( \widetilde f, \widetilde H_3 \big)$
on the closed 7-manifold $\widetilde \Sigma^7$ \eqref{DifferenceOriented}
(which is smooth by the assumption of sitting instants in \eqref{GaugedExtendedSigmaModelFields}):
\begin{equation}
  \label{GluingOfTwoCompatibleExtendedSigmaModelFields}
  \xymatrix@C=3em{
    \widehat \Sigma^{\, 7}_{L_{\phantom{A}}}
    \ar@{^{(}->}[d]^-{\iota_L}
    \ar[drr]^-{ \widehat f_L }
    && &
    d \big( \widehat H_3 \big)_L = \widehat f_L^\ast G_4
    \ar@{<-|}[d]^-{ \iota_L^\ast }
    \\
    \widetilde \Sigma^7
    \ar[rr]|-{\, \widetilde f \,}
    &&
    X
    &
    d \, \widetilde H_3 \,  = \widetilde f^\ast G_4
    \\
    \widehat \Sigma^{\, 7}_R
    \ar[urr]_-{ \widehat f_R }
    \ar@{_{(}->}[u]_-{ \iota_R }
    && &
    d \big( \widehat H_3 \big)_R = \widehat f_R^\ast G_4
    \ar@{<-|}[u]^-{ \iota_R^\ast }
  }
\end{equation}
In terms of this, the difference between the two extended action functionals
(Def. \ref{DefinitionOfWZWTermByFieldExtensionToCoboundary})
corresponding to the two choices of coboundaries
may be expressed as a single integral over $\widetilde \Sigma^7$:
\begin{equation}
  \label{DifferenceBetweenAnyTwoExtendedActionFunctionals}
  \begin{aligned}
  \widetilde S
  \big(
    \widetilde f,
    \,
    \widetilde H_3
  \big)
  & :=
  \widehat S
  \big(
    \widehat f_L, (\widehat H_3)_L
  \big)
  -
  \widehat S
  \big(
    \widehat f_R, (\widehat H_3)_R
  \big)
  \\
  & =
  \tfrac{1}{2}
  \underset{ \widetilde \Sigma^7 }{\int}
  \big(
    \widetilde H_3
      \wedge
    \widetilde f^\ast
    \big(
      G_4
      +
      \tfrac{1}{4}p_1(\nabla)
    \big)
    +
    \widetilde f^\ast 2G_7
  \big).
  \end{aligned}
\end{equation}
We call expression \eqref{DifferenceBetweenAnyTwoExtendedActionFunctionals}
the \emph{anomaly functional} of the 6d Hopf-Wess-Zumino term.
In supergravity this expressin is
also known as the {\it Page charge}
(traditionally disregarding the shift by $\tfrac{1}{4}p_1$, see
Remark \ref{ShiftedFluxQuantizationCorrectionsToTheHopfWZTerm}).
\end{defn}

\begin{lemma}[Hopf-WZ anomaly functional is homotopy invariant]
  \label{AnomalyFunctionalIsHomotopyInvariant}
  In the situation of Def. \ref{BasicSetup},
  let $\Sigma := \widetilde \Sigma^7$ be a closed 7-manifold.
  Then the
  Hopf-WZ anomaly functional \eqref{DifferenceBetweenAnyTwoExtendedActionFunctionals}
  is well-defined on the set \eqref{GaugeableFunctions}
  of homotopy classes of
  higher gauged sigma-model fields:
    \vspace{-2mm}
  \begin{equation}
    \label{GWIFunction}
    \xymatrix@R=-8pt{
      \pi_0
      \Big(
        \mathrm{Maps}_{\mathrm{smth}}^{\mathrm{ggd}}
        \big(\,
          \widetilde \Sigma^7
          ,
          X
        \big)
      \Big)
      \ar[rr]^{
        \widetilde S
      }
      &&
      \mathbb{R}
      \\
      \big[
        \widetilde f, \widetilde H_3
      \big]
       \ar@{}[rr]|-{\quad \longmapsto}
      &&
      \tfrac{1}{2}
      \underset{\widehat \Sigma^7 }{\bigintsss}
      \left(
        \widetilde H_3
          \wedge
        \widetilde f^\ast
        \big(
          G_4
          +
          \tfrac{1}{4}
          p_1(\nabla)
        \big)
        +
        \widetilde f^\ast 2G_7
      \right)
    }
  \end{equation}

  \vspace{-3mm}
\noindent  in that the integral on the right is independent of the choice of
  representative $(\widetilde f, \widetilde H_3)$ in its homotopy class.
\end{lemma}
\begin{proof}
  Given a gauge transformation/homotopy \eqref{GaugedHomotopy}
  between two extended gauged sigma-model fields
    \vspace{-2mm}
  $$
    \xymatrix{
      \big( \widetilde f_0, (\widetilde H_3)_0 \big)
      \ar@{=>}[rr]^-{ \left( \widetilde \eta, (\widetilde H_3)_{[0,1]}\right) }
      &&
      \big( \widetilde f_1, (\widetilde H_3)_1 \big)
    }
  $$

    \vspace{-2mm}
\noindent
  we need to show that
  $
    \widetilde S
    \big(
      \big[ \widetilde f_1, (\widetilde H_3)_1 \big]
    \big)
    \;=\;
    \widetilde S
    \big(
      \big[ \widetilde f_0, (\widetilde H_3)_0 \big]
    \big)
      $.
  With the data \eqref{DataHomotopyGauged}
  and using Stokes' theorem
  we directly compute as follows:
  $$
    \begin{aligned}
      \widetilde S
      \Big(
        \big[ \widetilde f_1, (\widetilde H_3)_1 \big]
      \Big)
      -
      \widetilde S
      \Big(
        \big[ \widetilde f_0, (\widetilde H_3)_0 \big]
      \Big)
      & =
      \tfrac{1}{2}
      \underset{
        \partial\big(
          \widetilde \Sigma^7 \times [0,1]
        \big)
      }{
        \int
      }
      \!\!
      \Big(
        \big(
          \widetilde H_3
        \big)_{[0,1]}
        \wedge
        \widetilde \eta^\ast \big( G_4 + \tfrac{1}{4}p_1(\nabla) \big)
        +
        \widetilde \eta^\ast 2G_7
      \Big)
      \\
      & =
      \tfrac{1}{2}
      \underset{
        \widetilde \Sigma^7 \times [0,1]
      }{
        \int
      }
      \!\!
      d
      \Big(
        \big(
          \widetilde H_3
        \big)_{[0,1]}
        \wedge
        \widetilde \eta^\ast \big( G_4 + \tfrac{1}{4}p_1(\nabla) \big)
        +
        \widetilde \eta^\ast 2G_7
      \Big)
      \\
      & =
      \tfrac{1}{2}
      \underset{
        \widetilde \Sigma^7 \times [0,1]
      }{
        \int
      }
            \Big(
        \underset{
          \widetilde \eta^\ast
          \big(
            G_4
            -
            \frac{1}{4}p_1(\nabla)
          \big)
        }
        {
          \underbrace{
            d
            \Big(
            \big(
              \widetilde H_3
            \big)_{[0,1]}
            \Big)
          }
        }
        \wedge
        \widetilde \eta^\ast \big( G_4 + \tfrac{1}{4}p_1(\nabla) \big)
        +
        \widetilde \eta^\ast d 2G_7
      \Big)
      \\
      & =
      \tfrac{1}{2}
      \underset{
        \widetilde \Sigma^7 \times [0,1]
      }{
        \int
      }
      \!\!
      \widetilde \eta^\ast
      \Big(
        \underset{
          = 0
        }{
        \underbrace{
          \big(G_4 - \tfrac{1}{4}p_1(\nabla) \big)
          \wedge
          \big( G_4 + \tfrac{1}{4}p_1(\nabla) \big)
          +
          d 2 G_7
        }
        }
      \Big)
      \\
      & = 0
      \,,
    \end{aligned}
  $$
  where in the last step, under the brace, we used the condition \eqref{SpacetimeFormAssumption}.
\end{proof}

\section{The full M5 Hopf-WZ anomaly is a homotopy Whitehead integral}
\label{InTermsOfCohomotopy}

We first recall from \cite{FSS19b}\cite[\S 5.3]{FSS20b}
how the background C-field $(G_4, 2G_7)$ is a cocycle in twisted rational Cohomotopy; this is Remark \ref{BackgroundFieldsAsRationalCohomotopy} below. Then we prove, in Theorem
\ref{AnomalyFunctionalAsLiftInCohomotopy},
that the Hopf-WZ anomaly functional from \cref{The6dWZWTerm} is equivalently
a lift in rational Cohomotopy through
the equivariant quaternionic Hopf fibration,
hence is in particular a homotopy invariant of both the
gauged sigma-model fields and the background fields in Cohomotopy.
Further below, in \cref{HypothesisHImpliesAnomalyCancellation},
this allows us to identify the anomaly functional
as a twisted/parametrized generalization
of a homotopy Whitehead integral.

\medskip

\noindent {\bf Notions from rational homotopy theory.}
In the following, we make free use of Sullivan model
dgc-algebras in rational homotopy theory (i.e., what in supergravity
are called ``FDA''s \cite{FSS13}\cite{FSS19a});
see \cite{Sullivan77}\cite{BousfieldGugenheim76} for the original accounts,
\cite{Hess06} for introduction,
\cite{GrMo13} for a standard textbook account,
and see \cite{FSS16a}\cite{FSS17} \cite{FSS19a}\cite{FSS20b} for review
streamlined towards our application.
As in these references, for $X$ a simply-connected topological space of finite rational type, we write
$\mathrm{CE}( \mathfrak{l} X)$ for its minimal Sullivan model
differential graded-commutative algebra (dgc-algebra),
indicating that this is the Chevalley-Eilenberg
algebra of the minimal Whitehead $L_\infty$-algebra $\mathfrak{l}X$
corresponding to the loop group of $X$ (see \cite[Prop. 3.64]{FSS20b}).
For making Sullivan models explicit, we display the list
of differential relations on each generator, thereby
declaring what the generators are (see \cite[(99)]{FSS20b}),
as shown in the following examples.

\begin{example}[Quaternionic Hopf fibration]
We denote the minimal relative Sullivan model for the
plain quaternionic Hopf fibration $h_{\mathbb{H}}$
 as follows (see \cite[Lemma 3.18]{FSS19b}):

\vspace{-.3cm}
\begin{equation}
  \label{TheQuaternionicHopfFibration}
  \hspace{-2cm}
  \raisebox{40pt}{
  \xymatrix@R=3pt{
    S^7
    \ar[dddd]^-{
      h_{\mathbb{H}}
    }_-{
      \mbox{
        \tiny
        \color{darkblue}
        \bf
        {\begin{tabular}{c}
          quaternionic
          \\
          Hopf fibration
        \end{tabular}}
      }
    }
    &
    {\phantom{AAAA}}
    &
    \mathclap{\phantom{\vert^{\vert^{\vert}}}}
    \mathrm{CE}\big( \mathfrak{l}S^7 \big)
    \ar@{=}[r]
    \ar@{<-^{)}}[dddd]^{
      \mathrm{CE}(\mathfrak{l} h_{\mathbb{H}} )
    }_-{
      \mbox{
        \tiny
        \color{darkblue}
        \bf
        \begin{tabular}{c}
          minimal relative
          \\
          Sullivan model
          \\
          (``FDA'')
        \end{tabular}
      }
    }
    &
    \big( \!\!\!
      {\begin{array}{lcl}
        d\,\omega_{\, 7} &=& 0
      \end{array}}
    \!\!\!  \big)
    \ar@{<-}[dddd]^-{
      \scalebox{.7}{
        $
        \arraycolsep=1.4pt
        \begin{array}{cc}
          0 & \omega_{\, 7}
          \\
          \mapsup & \mapsup
          \\
          \omega_4 & \omega_{\, 7}
        \end{array}
        $
      }
    }
    \\
    \\
    \\
    \\
    S^4
    &&
    \mathrm{CE}\big( \mathfrak{l}S^4 \big)
    \ar@{=}[r]
    &
    \left(
      \!\!\!
      {\begin{array}{lcl}
        d\,\omega_4 &=& 0
        \\
        d\,\omega_{\, 7} &=& - \omega_4 \wedge \omega_4
      \end{array}}
      \!\!\!
    \right)
    \\
    \mbox{
      \color{darkblue}
      \begingroup
\setlength{\tabcolsep}{3pt} 
\renewcommand{\arraystretch}{.8} 
  \small
  \bf
      \begin{tabular}{c}
        topological
        homotopy theory
      \end{tabular}
\endgroup
    }
    &&
    \mathrlap{
    \mbox{
      \color{darkblue}
      \begingroup
      \setlength{\tabcolsep}{3pt} 
\renewcommand{\arraystretch}{.8} 
  \small
  \bf
      \begin{tabular}{c}
        dgc-algebraic
        rational homotopy theory
      \end{tabular}
      \endgroup
    }
    }
  }
  }
\end{equation}
\end{example}
\begin{example}[Classifying space of Fivebrane-extended Spin-group]
  \label{ClassifyingSpaceOfHigherExtendedQuaternionUnitaryGroup}
The minimal relative Sullivan model for the
homotopy fiber $B \widehat {\mathrm{Sp}(2)}$
of the classifying
map for the Euler 8-class $\rchi_8$ on the classifying space
for the quaternionic unitary group
$\mathrm{Sp}(2)\hookrightarrow \mathrm{Spin}(8)$
(see \cite[\S A]{FSS20a})
is as follows (using \cite[(97)]{FSS19b}\cite[Lem. 4.24]{FSS20b}):

\vspace{-.2cm}
\begin{equation}
  \label{HomotopyFiberSpace}
\hspace{-1.8cm}
  \raisebox{23pt}{
  \xymatrix@C=15pt{
    B \widehat{\mathrm{Sp}(2)}
    \ar[d]|-{
      \mathclap{\phantom{\vert^{\vert^{\vert}}}}
     \scalebox{.6}{$  \mathrm{hofib}(\rchi_8)$}
      \mathclap{\phantom{\vert_{\vert_{\vert}}}}
    }
    \ar[r]^-{\simeq}
    &
    B \widehat{\mathrm{Spin}(5)}
    \ar[d]|-{
      \mathclap{\phantom{\vert^{\vert^{\vert}}}}
      \scalebox{.6}{$
      \mathrm{hofib}
      \big(
        - \tfrac{1}{4}p_2
        +
        \big(
          \tfrac{1}{4}p_1
        \big)^2
      \big)
      $}
      \mathclap{\phantom{\vert_{\vert_{\vert}}}}
    }
    &
    \ar@{}[d]|<{
      \mbox{
        \tiny
        \color{darkblue}
        \bf
        {\begin{tabular}{c}
          homotopy fiber
          \\
          trivializing
          \\
          obstructing 8-class
        \end{tabular}}
      }
    }
    &
    &&
    \left(
    \scalebox{.8}{$
    {\begin{aligned}
      d\, {\color{darkblue}\theta_7} & = \rchi_8
      \\
      d \, \rchi_8 & = 0
      \\
      d \, \tfrac{1}{2} p_1 & = 0
    \end{aligned}}
    $}
    \right)
    \ar@{<-^{)}}[d]_-{
      \mbox{
        \tiny
        \color{darkblue}
        \bf
        \begin{tabular}{c}
          minimal relative
          \\
          Sullivan model
        \end{tabular}
      }
    }
    \ar@{<-}[r]^-{\simeq}
    &
    \left(
    \scalebox{.8}{$
    {\begin{aligned}
      d\, {\color{darkblue}\theta_7}
        & =
      -\tfrac{1}{4}p_2 + \big( \tfrac{1}{4}p_1 \big)^2
      \\
      d\, p_2 & = 0
      \\
      d\, \tfrac{1}{2}p_1 & = 0
    \end{aligned}}
    $}
    \right)
    \ar@{<-^{)}}[d]
    \\
    B \mathrm{Sp}(2)
    \ar@/_1.5pc/[rrr]|-{
      \;\rchi_8\;
    }
    \ar[r]_-{ \simeq }^-{
      \scalebox{.6}{$
        \mathrm{tri}
      $}
    }
    &
    B \mathrm{Spin}(5)
    \ar[rr]^-{
      \scalebox{.6}{$
        -\tfrac{1}{4}p_2
        + \big(\tfrac{1}{4}p_1\big)^2
      $}
    }
    &{\phantom{AA}}&
    K(\mathbb{Z},8)
    &&
    \left(
    \scalebox{.8}{$
    {\begin{aligned}
      d \, \rchi_8 & = 0
      \\
      d \, \tfrac{1}{2} p_1 & = 0
    \end{aligned}}
    $}
    \right)
    \ar@{<-}@/_1.6pc/[rrr]_-{
      \scalebox{.6}{$
        \rchi_8 \;\;\mapsfrom\;\; c_8
      $}
    }
    \ar@{<-}[r]_-{ \simeq }^{ \mathrm{tri}^\ast }
    &
    \left(
    \scalebox{.8}{$
    {\begin{aligned}
      d \, p_2 & = 0
      \\
      d \, \tfrac{1}{2} p_1 & = 0
    \end{aligned}}
    $}
    \right)
    \ar@{<-}[rr]^-{
      \scalebox{.6}{$
        -\tfrac{1}{4}p_2
        + \big(\tfrac{1}{4}p_1\big)^2
        \;\mapsfrom\;
        c_8
      $}
    }
    &&
    \mathrlap{
    \left(
    \,
    \scalebox{.8}{$
    {\begin{aligned}
      d \, c_8 & = 0
    \end{aligned}}
    $}
    \!\!\!\!\!\!\!
    \right)
    .
    }
  }
  }
\end{equation}
\vspace{-.3cm}

\noindent Notice that the  higher extension
$\widehat {\mathrm{Sp}(2)}$ in \eqref{HomotopyFiberSpace}
is a version of the {\it Fivebrane 6-group},
and tangential $\widehat {\mathrm{Sp}(2)}$-structure
trivializing this 8-class is a kind of
{\it Fivebrane structure} according to \cite{SSS09}\cite{SSS12}:
a higher
analog of {\it String structure}
(trivializing $\tfrac{1}{2}p_1$), which itself is
a higher analog of $\mathrm{Spin}$-structure
(trivializating $w_2$).
In higher analogy to how
Spin-structure is the topological condition on target spacetime
needed for anomaly cancellation of the spinning particle \cite{Witten85},
and String-structure is the topological condition
for anomaly cancellation of the (heterotic) string
(i.e. for the Green-Schwarz mechanism, see \cite{SSS12}\cite{FSS20}\cite{FSS20a}),
so Fivebrane structure is meant to be the topological
condition needed for anomaly cancellation of the
(heterotic) five-brane.
That this is the case for $\widehat {\mathrm{Sp}(2)}$-structure,
as concerns the Hopf-WZ term of the M5-brane,
is brought out by our main
Theorems \ref{AnomalyFunctionalAsLiftInCohomotopy} and \ref{AnomalyIsIntegral} below,
see Remark \ref{TheHopfWZTermOverBFivebrane} below.
\end{example}

\noindent {\bf Rationalization over the real numbers.}
In order to have a {\it smooth} non-abelian de Rham theorem
(\cite[Thm. 3.87]{FSS20b})
involving the real de Rham dg-algebras $\Omega_{\mathrm{dR}}^\bullet(-)$
of smooth differential forms,
we take the rational
base field to be $\mathbb{R}$ instead of $\mathbb{Q}$
(as in \cite{GrMo13}\cite[Rem. 3.51]{FSS20b}), so that
our ``rational homotopy groups'' are actually ``real homotopy groups'' $\pi(X) \otimes_{\mathbb{Z}} \mathbb{R}$;
which makes no essential difference
(by \cite[Lem. 11.7]{BousfieldGugenheim76}).
Accordingly, for $X$ a simply-connected topological space, we write
  \vspace{-4mm}
\begin{equation}
  \label{Rationalization}
  \xymatrix{
    X
    \ar[rr]^-{ \eta^{\mathbb{R}}_X }_-{
      \mbox{
        \tiny
        \color{greenii}
        rationalization
      }
    }
    &&
    L_{\mathbb{R}}X
  }
\end{equation}

  \vspace{-3mm}
\noindent
for its rationalization
(localization over the real numbers, see \cite[Def. 3.55]{FSS20b}).

\begin{example}[Background C-field is cocycle in rational twisted Cohomotopy]
  \label{BackgroundFieldsAsRationalCohomotopy}
  The minimal Sullivan model (``FDA'') of the 4-sphere is
  free on generators $\omega_4$ and $\omega_{\, 7}$
  (in degrees 4 and 7, respectively), subject to differential relations
  shown on the bottom of \eqref{TheQuaternionicHopfFibration}.
This  means
  (\cite[\S 2.5]{Sati13},
  see also \cite[\S 2]{FSS16a}\cite[(59)]{FSS19a}\cite[Ex. 3.81]{FSS20b})
  that the background
  C-field data \eqref{SpacetimeFormAssumption}
  in the case that $p_1(\nabla) = 0$ \eqref{Pontrjagin4Form},
  is equivalently a flat $L_\infty$-algebra valued differential form
  \cite[Def. 3.77]{FSS20b} with values in the
  Whitehead $L_\infty$-algebra $\mathfrak{l}S^4$
  \cite[Ex. 3.68]{FSS20b}, namely a
  dg-algebra homomorphism from
  $\mathrm{CE}\big( \mathfrak{l}S^4\big)$
  to the de Rham algebra of $X$:

  \vspace{-.6cm}
  \begin{equation}
    \label{BackgroundFieldsAsdgcAlgebraHomomorphism}
    \xymatrix@C=24pt@R=0pt{
      X
      \ar[rr]^-{ (G_4, 2 G_7) }_>>>>>>>>>{\ }="s"
        &&
      L_{\mathbb{R}}S^4
      &
      \xleftrightarrow{
        \mbox{
          \tiny
          {\begin{tabular}{c}
            \color{greenii}
            \bf
            non-abelian
            de Rham theorem
            \\
            \cite[Thm. 3.87]{FSS20b}
          \end{tabular}}
        }
      }
      &
      \Omega_{\mathrm{dR}}^\bullet(X)
      \ar@{<-}[rr]^-{
        \mbox{
          \tiny
          $
          \begin{array}{rcl}
            G_4 &\!\!\!\!\!\!\mapsfrom\!\!\!\!\!\!& \omega_4
            \\
            2G_7 &\!\!\!\!\!\!\mapsfrom\!\!\!\!\!\!& \omega_{\, 7}
          \end{array}
          $
        }
      }
      &&
      \left(
        \!\!\!\!\!\!
        \mbox{
        \small
        $
        {\begin{array}{ll}
          d\,\omega_4 & = 0
          \\
          d\,\omega_{\, 7}
            & =
            - \omega_4 \wedge \omega_4
        \end{array}}
        $
        }
        \!\!\!\!\!\!
      \right)
      \\
      \ar@{}[rr]|-{
        \mathclap{
        \raisebox{10pt}{
          \footnotesize
          \color{darkblue}
          \bf
          \begin{tabular}{c}
            cocycle in
            \\
            rational 4-Cohomotopy
          \end{tabular}
        }
        }
      }
      &&
      &&
      \ar@{}[rr]_-{
        \mbox{
          \footnotesize
          \color{darkblue}
          \bf
          dg-algebra homomorphism
        }
      }
      &&
    }
  \end{equation}

  \vspace{-2mm}
\noindent
  More generally \cite[Prop. 3.20]{FSS19b}, consider
  $X = X^8$ a spin 8-manifold
  for M-theory compactified on 8-manifolds \cite[Rem. 3.1]{FSS19b},
  hence such that:

  \noindent
  {\bf (i)}
  the tangent bundle of $X^8$
  is equipped with tangential
  $\mathrm{Sp}(2) \to \mathrm{Spin}(8)$-structure
  $\tau$ (reflecting M2-brane background, see \cite[p. 8 \& \S 2.3]{FSS19b})
  with compatible connection $\nabla$
  (see \cite[Def. 5.25]{FSS20b}):

  \vspace{-.6cm}
  \begin{equation}
    \label{TangentialSp2Structure}
    \xymatrix@C=30pt{
      X^8
      \ar[dr]_{
        \underset{
          \mathllap{
            \raisebox{-3pt}{
              \tiny
              \color{darkblue}
              \bf
              \begin{tabular}{c}
                classifying map
                \\
                of spin-structure
              \end{tabular}
            }
            \;\;\;\;\;\;
          }
        }{
          T X^8
          \;
        }
      }^-{\ }="t"
      \ar[rr]^-{
        \overset{
          \mathclap{
          \raisebox{3pt}{
            \tiny
            \color{darkblue}
            \bf
            \begin{tabular}{c}
              tangential
              \\
              $\mathrm{Sp}(2)$-structure
            \end{tabular}
          }
          }
        }{
          \tau
        }
      }_>>>{\ }="s"
      &&
      B \mathrm{Sp}(2)
      \ar[dl]
      &{\phantom{}}&
      \Omega^\bullet_{\mathrm{dR}}
      \big(
        X^8
      \big)
      \ar@{<-}[dr]_-{
        \scalebox{.7}{$
        \arraycolsep=1.4pt
        \begin{array}{rcl}
          p_2(\nabla)
          &\mapsfrom&
          p_2
          \\
          p_1(\nabla)
          &\mapsfrom&
          p_1
        \end{array}
        $}
      }
      \ar[rr]^-{
        \scalebox{.7}{
        $
        \arraycolsep=1.4pt
        {\begin{array}{rcl}
          \rchi_{8}(\nabla)
          &\mapsfrom&
          \rchi_8
          \\
          \tfrac{1}{2}p_1(\nabla)
          &\mapsfrom&
          \tfrac{1}{2}p_1
        \end{array}}
        $
      }
      }
      &&
      \left(
      {\begin{aligned}
        d\, \rchi_8 & = 0
        \\
        d\, \tfrac{1}{2}p_1 & = 0
      \end{aligned}}
      \right)
      \ar@{<-}[dl]^-{
        \scalebox{.7}{$
        \arraycolsep=1.4pt
        \begin{array}{rcl}
          p_2
          &\mapsto&
          4
          \Big(
            \big(
              \tfrac{1}{4}p_1
            \big)^2
            -
            \rchi_8
          \Big)
          \\
          p_1
          &\mapsto&
          p_1
        \end{array}
        $}
      }
      \\
      &
      B \mathrm{Spin}(8)
      &
      &&
      &
      \left(
      {\begin{aligned}
        d\, p_2 & = 0
        \\
        d\, p_1 & = 0
      \end{aligned}}
      \right)
      \ar@{=>} "s"; "t"
    }
  \end{equation}
  \vspace{-.3cm}

  \noindent
  {\bf (ii)}
  the corresponding Euler 8-form $\rchi_8(\nabla)$ trivializes
  (meaning that the singular M2-brane loci themselves are removed from $X^8$, see \cite[\S 2.5]{FSS19b})
  \begin{equation}
    \label{EulerFormAsPfaffian}
    \Theta_7 \in \Omega_{\mathrm{dR}}^7(X)
    \phantom{AA}
    \mbox{\rm such that}
    \phantom{AA}
    d \Theta_7
    =
    \rchi_8(\nabla)
      :=
    \mathrm{Pf}(R)
    \,,
  \end{equation}
  which means
  (by Example
  \ref{ClassifyingSpaceOfHigherExtendedQuaternionUnitaryGroup})
  that the $\mathrm{Sp}(2)$-structure on $X$ further lifts to
  $\widehat {\mathrm{Sp}(2)}$-structure $\widehat \tau$:

  \vspace{-.6cm}
  \begin{equation}
    \label{hatSp2Structure}
    \xymatrix@C=2.7em{
      X
      \ar[dr]_-{ \tau }^-{\ }="t"
      \ar[rr]^-{
        \widehat \tau
      }_>>>>>{\ }="s"
      &&
      B \widehat{ \mathrm{Sp}(2) }
      \ar[dl]^-{
              \mathrm{hofib}(\rchi_8)
              }
      &&
      \Omega^\bullet_{\mathrm{dR}}
      \big(
        X^8
      \big)
      \ar@{<-}[rr]|-{
        \;
        \scalebox{.7}{$
          \begin{array}{rcl}
            \Theta_7 &\mapsfrom& \theta_7
            \\
            \rchi_8(\nabla) &\mapsfrom& \rchi_8
            \\
            \tfrac{1}{2}p_1(\nabla) &\mapsfrom& \tfrac{1}{2}p_1
               \end{array}
        $}
        \;
      }
      \ar@{<-}[dr]_-{
        \scalebox{.7}{$
          \begin{array}{rcl}
            \rchi_8(\nabla)
            &\mapsfrom&
            \rchi_8
            \\
            \tfrac{1}{2}p_1(\nabla)
            &\mapsfrom&
            \tfrac{1}{2}p_1
                 \end{array}
        $}
      }
      &&
      \left(
      {\begin{aligned}
        d\, {\color{blue}\theta_7} & = \rchi_8
        \\
        d\, \rchi_8 & = 0
        \\
        d\, \tfrac{1}{2}p_1 & = 0
      \end{aligned}}
      \right)
      \ar@{<-_{)}}[dl]
      \\
      &
      B \mathrm{Sp}(2)
      &
      &&
      &
            \mathrm{CE}
      \big(
        \mathfrak{l} B \mathrm{Sp}(2)
      \big)
      \ar@{=>} "s"; "t"
    }
  \end{equation}
  \vspace{-.4cm}

  \noindent Then the general background field data \eqref{SpacetimeFormAssumption}
  (now including the $p_1$-terms, Remark \ref{ShiftedFluxQuantizationCorrectionsToTheHopfWZTerm})
  may be identified with a cocycle in rational
  \emph{$\tau$-twisted Cohomotopy} (see also \cite[Ex. 3.96]{FSS20b}):
  \begin{equation}
    \label{GeneralBackgroundFieldDataAsCohomotopyCocycle}
    \hspace{2mm}
    \scalebox{.9}{
    \raisebox{20pt}{
    \xymatrix@C=13pt@R=20pt{
      X
      \ar[rr]^-{ (G_4, 2G_7) }_>>>>>>>>>{\ }="s"
      \ar[dr]_-{\tau}^{\ }="t"
        &&
        L_{\mathbb{R}}
      \big(
        S^4
        \!\sslash\!
        \mathrm{Sp}(2)
      \!\big),
      \ar[dl]
      &
      {\phantom{AAA}}
      &
      \;\;\;
      \Omega_{\mathrm{dR}}^\bullet(X)
      \ar@{<-}[rr]^-{
        \mbox{
          \tiny
          $
          \begin{array}{lcl}
            \phantom{2}G_4 &\!\!\!\!\!\!\mapsfrom\!\!\!\!\!\!& \omega_4
            \\
            2G_7
            -
            \Theta_7
              &\!\!\!\!\!\!\mapsfrom\!\!\!\!\!\!&
            \omega_{\, 7}
          \end{array}
          $
        }
      }
      \ar@{<-}[dr]_<<<<<<{\!\!
        \mbox{
          \tiny
          $
          \begin{array}{rcl}
            p_1(\nabla) &\!\!\!\!\!\!\mapsfrom\!\!\!\!\!\!& p_1
            \\
            p_2(\nabla) &\!\!\!\!\!\!\mapsfrom\!\!\!\!\!\!& p_2
            \\
            \rchi_8(\nabla) &\!\!\!\!\!\!\mapsfrom\!\!\!\!\!\!& \rchi_8
          \end{array}
          $
        }
        \!\!\!\!\!
        \!\!\!\!\!
        \!\!\!
      }
      &&
      \left(
        \!\!\!\!\!\!
        \mbox{
        \small
        $
        {\begin{array}{ll}
          d\,\omega_4 & \!\!\!\!\!\!\! = 0
          \\
          d\,\omega_{\, 7}
           &\!\!\!\!\!\!\!  =
            - \omega_4 \wedge \omega_4
            + \tfrac{1}{4}p_1 \wedge \tfrac{1}{4}p_1
            \\
            &\!\!\!\!\!\! \phantom{=}
            - \rchi_8
        \end{array}}
        $
        }
        \!\!\!\!\!\!\!\!
      \right)
     \ar@{<-_{)}}[dl]^<<<<<<{
        \mbox{\;
          \tiny
          $
          \begin{array}{lcl}
            p_1 &\!\!\!\!\!\!\mapsto\!\!\!\!\!\!& p_1
            \\
            p_2 &\!\!\!\!\!\!\mapsto\!\!\!\!\!\!& p_2
            \\
            \rchi_8 &\!\!\!\!\!\!\mapsto\!\!\!\!\!\!& \rchi_8
          \end{array}
          $
        }
      }
      \\
      &
      L_{\mathbb{R}}
      B \mathrm{Sp}(2)
      &
      &&
      &
      \mathrm{CE}
      \big(
        \mathfrak{l}
        B \mathrm{Sp}(2)
      \!\big)
      \ar@{=>} "s"; "t"
      \\
      \ar@{}[rr]|-{
        \mathclap{
        \mbox{
          \small
          \color{darkblue}
          \bf
          \begin{tabular}{c}
            cocycle in
            twisted rational 4-Cohomotopy
          \end{tabular}
        }
        }
      }
      &&
      &&
      \ar@{}[rr]|-{
        \mathclap{
        \mbox{
          \small
          \color{darkblue}
          \bf
          \begin{tabular}{c}
            relative
            dg-algebra homomorphism
          \end{tabular}
        }
        }
      }
      &&
     }
     }
    }
  \end{equation}
\end{example}

Our first main Theorem \ref{AnomalyFunctionalAsLiftInCohomotopy}
says that not only
does rational twisted Cohomotopy naturally encode the
background C-field,
via Remark \ref{BackgroundFieldsAsRationalCohomotopy},
but that it also naturally encodes
the gauging \eqref{GaugedExtendedSigmaModelFields}
of the M5-brane sigma-model fields (as in {\cite[Rem. 3.17]{FSS19b}})
as well as the anomaly functional of the 6d Hopf-WZ term
(Def. \ref{TheAnomalyFunctional})
as a homotopy invariant (Lemma \ref{AnomalyFunctionalIsHomotopyInvariant}):

\begin{theorem}[6d Hopf-WZ anomaly functional is lift through $h_{\mathbb{H}}$]
  \label{AnomalyFunctionalAsLiftInCohomotopy}
  In the situation of Def. \ref{BasicSetup},
  consider a closed extended worldvolume $\Sigma := \widetilde \Sigma^7$.
  Then, under the identification of the background field
  with a cocycle $c$ in rational twisted Cohomotopy,
  via Remark \ref{BackgroundFieldsAsRationalCohomotopy},
  we have:
  \vspace{-3mm}
  \begin{enumerate}[{\bf (i)}]
    \item The homotopy classes \eqref{GaugeableFunctions}
    of gaugings $\widetilde H_3$ \eqref{GaugedExtendedSigmaModelFields}
    of an extended sigma-model field $\widetilde f$
    are in bijection to homotopy classes of
    homotopy lifts $\widehat{ c \circ \widetilde f}$
    through the quaternionic Hopf fibration $h_{\mathbb{H}}$
    \eqref{TheQuaternionicHopfFibration}
    of the composite $c \circ \widetilde f$
    of $\widetilde f$
    with the classifying map $c$
    \eqref{GeneralBackgroundFieldDataAsCohomotopyCocycle} of the background C-field:
    \vspace{-2mm}
    \begin{equation}
      \label{GaugingHomotopy}
      \pi_0
      \big(
        \mathrm{Maps}^{\mathrm{ggd}}_{\mathrm{smth}}
        (
          \Sigma, X
        )
      \big)_{\vert_{\widetilde f}}
      \;\;\simeq\;\;
      \left\{
        \raisebox{22pt}{
        \xymatrix@C=24pt@R=2.5em{
          \widetilde \Sigma^7
          \ar@{-->}[rr]^-{
            \widehat{
              c \circ \widetilde f
            }
          }_>>>>>>>>>>{\ }="s"
          \ar[dr]_-{
            c \circ \widetilde f
            \;
          }^{\ }="t"
          &&
          L_{\mathbb{R}}
          \big(
          S^7
          \!\sslash\!
          \mathrm{Sp}(2)
          \!\big)
          \ar[dl]^<<<<<<<<<<{
            \;\;\;\;\;\;
            L_{\mathbb{R}}
            (
            h_{\mathbb{H}}
            \sslash
              \mathrm{Sp}(2)
            )
          }
          \\
          &
          L_{\mathbb{R}}
          \big(
            S^4
            \!\sslash\!
            \mathrm{Sp}(2)
          \!\big)
          \ar@{=>} "s"; "t"
        }
        }
      \right\}_{\!\!\!\!\!\Big/\sim_{
        \!\!\!\!\!\!\!\!\!\!\!\!\!\!\!\!\!
        \mbox{
          \tiny
          \rm
          \def\arraystretch{.5}
          \begin{tabular}{c}
            relative
            \\
            homotopy
          \end{tabular}
        }
      }}
          \end{equation}

          \vspace{-4mm}
    \item
    Under this bijection \eqref{GaugingHomotopy},
    twice the anomaly functional (Def. \ref{TheAnomalyFunctional})
    equals
    the correction by the Euler-potential $\Theta_7$ \eqref{EulerFormAsPfaffian} of the
    integral
    \begin{equation}
      \label{TheEulerCorrectedTerm}
      2\widetilde S\big(\widetilde f, \widetilde H_3\big)
      \;=\;
      \underset{
        \widetilde \Sigma^7
      }{
        \int
      }
      \Big(
      \big(
        \widehat{ c\circ \widetilde f }\;
      \big)^\ast
      (
        \omega_{\, 7}
      )
      \;+\;
      f^*\Theta_7
      \Big)
    \end{equation}

      \vspace{-2mm}
\noindent
    of the pullback of the angular cochain $\widetilde \omega_{\, 7}$
    on the universal 7-spherical fibration \eqref{GeneralBackgroundFieldDataAsCohomotopyCocycle}
    which
    is fiberwise
    the unit volume form on $S^7$ and which trivializes minus the {universal} Euler form:
      \vspace{-2mm}
    \begin{equation}
      \label{TopRightomega7}
      \langle \widetilde \omega_{\, 7}, S^7\rangle
      =
      1
      \,,
      \phantom{AA}
      d \widetilde \omega_{\, 7} = - \rchi_8
      \,.
    \end{equation}
  \end{enumerate}
\end{theorem}
\begin{proof}
  By \cite[Lem. 3.19]{FSS19b}
  the dgc-algebra model for the situation is as shown
  on the right in the following diagram,
  where the generator $\widetilde \omega_{\, 7}$ in the top right
  satisfies \eqref{TopRightomega7} by \cite[Prop. 2.5 (39)]{FSS19b}:
    \vspace{-2mm}
    \begin{equation}
      \label{DiagramExhibitinAnomalyFunctionalAsHopfInvariant}
      \hspace{-.15cm}
      \raisebox{78pt}{
      \xymatrix@C=14pt@R=3.2em{
        \widetilde \Sigma^7
        \ar[dd]_-{ \widetilde f }^>>>>>>>>>>>>>>{\ }="t"
        \ar@{-->}[rr]^-{
          \widehat{ c \circ \widetilde f }
        }_>>>>{\ }="s"
        &&
        L_{\mathbb{R}}
        \big(
          S^7
            \!\sslash\!
          \mathrm{Sp}(2)
        \!\big)
        \ar[dd]|-{
          L_{\mathbb{R}}
          (
            h_{\mathbb{H}}
            \sslash
            \mathrm{Sp}(2)
          )^{\phantom{A}}_{\phantom{A}}
        }
        &\;\;&
        \Omega_{\mathrm{dR}}^\bullet
        \big(\,
          \widetilde \Sigma^7
        \big)
        \ar@{<-}[rrr]^-{
          \mbox{
            \tiny
            $
            \begin{array}{lcl}
              \left({\color{darkblue}2\widetilde S}
              \big(
                \widetilde f,
                \widetilde H_3
              \big){-\int_{\tilde{\Sigma}^7}f^*\Theta_7}\right)
                \cdot
              \mathrm{vol}_{\widetilde \Sigma^7}
              &\!\!\!\!\!\!\!\!\!\mapsfrom\!\!\!\!\!\!\!\!\!&
              \widetilde \omega_{\, 7}
            \end{array}
            $
          }
        }_>>>>>{\ }="s2"
        \ar@{<-}[drrr]_-{
          \mbox{
            \tiny
            $
            \arraycolsep=2.2pt
            \begin{array}{rcl}
              \mathclap{\phantom{\vert^{\vert^{\vert}}}}
              \widetilde H_3
                &\!\!\!\!\mapsfrom\!\!\!\!&
               h_3
              \\
              \widetilde f^\ast \phantom{2}G_4
                &\!\!\!\!\mapsfrom\!\!\!\!&
              \omega_4
              \\
              \widetilde f^\ast ( 2G_7 - \Theta_7)
                  &\!\!\!\!\mapsfrom\!\!\!\!&
              \omega_{\, 7}
            \end{array}
            $
          }
        }^>>>>>{\ }="t2"
        \ar@{<-}[dd]_{
          \widetilde f^\ast
        }
        &&&
        \big(
          \!
          \arraycolsep=2.2pt
          {\begin{array}{lcl}
            d\,\widetilde \omega_{\, 7} & =& - \rchi_8
          \end{array}}
          \!
        \big)
        \ar@{<-}[d]_-{\simeq}^-{
          \mbox{
            \tiny
            $
            \arraycolsep=1.4pt
            \begin{array}{ccc}
              0 & \tfrac{1}{4}p_1 & \widetilde \omega_{\, 7}
              \\
              \mapsup & \mapsup & \mapsup
              \\
              h_3 & \omega_4 & \omega_{\, 7}
            \end{array}
            $
          }
        }
        \\
        &
        &
        &&
        &&&
        \left(
          \!\!\!\!
          {\begin{array}{lcl}
            d\,h_3 &\!\!\!\!\!=\!\!\!\!\!& \omega_4 - \tfrac{1}{4}p_1
            \\
            d\,\omega_4 &\!\!\!\!\!=\!\!\!\!\!& 0
            \\
            d\,\omega_{\, 7}
              &\!\!\!\!\!=\!\!\!\!\!&
              - d h_3 \wedge \big( \omega_4 + \tfrac{1}{4}p_1\big)
            \\
              &&
              - \rchi_8
          \end{array}}
          \!\!\!\!\!\!\!
        \right)
        \ar@{<-^{)}}[d]^-{
          \mbox{
            \tiny
            $
            \arraycolsep=1.4pt
            \begin{array}{cc}
              \omega_4 & \omega_{\, 7}
              \\
              \mapsup & \mapsup
              \\
              \omega_4 & \omega_{\, 7}
            \end{array}
            $
          }
        }
        \\
        X
        \ar[dr]_{ \tau }
        \ar[rr]^-{ c }
        &&
        L_{\mathbb{R}}
        \big(
          S^4
          \!\sslash\!
          \mathrm{Sp}(2)
        \!\big)
        \ar[dl]
        &&
        \Omega_{\mathrm{dR}}^\bullet(X)
        \ar@{<-}[rrr]^-{
          \mbox{
            \tiny
            $
            \arraycolsep=2.2pt
            \begin{array}{lcl}
              \phantom{2}G_4 &\!\!\!\mapsfrom\!\!\!& \omega_4
              \\
              2G_7 - \Theta_7 &\!\!\!\mapsfrom\!\!\!& \omega_{\, 7}
            \end{array}
            $
          }
        }
        \ar@{<-}[dr]_<<<<<<<<<<{
          \mbox{
            \tiny
            $
            \arraycolsep=1.4pt
            \begin{array}{lcl}
              p_1(\nabla) &\mapsfrom& p_1
              \\
              p_2(\nabla) &\mapsfrom& p_2
              \\
              \rchi_8(\nabla) &\mapsfrom& \rchi_8
            \end{array}
            $
          }
          \!\!\!\!\!
        }
        &&&
        \left(
          \!\!\!\!
          {\begin{array}{lcl}
            d\,\omega_4 &\!\!\!\!\!=\!\!\!\!\!& 0
            \\
            d\,\omega_{\, 7}
              &\!\!\!\!\!=\!\!\!\!\!&
              - \omega_4 \wedge \omega_4
              + \big(\tfrac{1}{4}p_1\big)^2
            \\
              &&
              - \rchi_8
          \end{array}}
          \!\!\!\!\!\!
        \right)
        \ar@{<-_{)}}[dll]^<<<<<<<<<<{
          \mbox{
            \tiny
            $
            \arraycolsep=1.4pt
            \begin{array}{lcl}
              p_1 &\mapsto& p_1
              \\
              p_2 &\mapsto& p_2
              \\
              \rchi_8 &\mapsto& \rchi_8
            \end{array}
            $
          }
        }
        \\
        &
        \mathrm{B}
        \mathrm{Sp}(2)
        &
        &
        &&
         \mathrm{CE}
        \big(
          \mathfrak{l}
          B
          \mathrm{Sp}(2)
        \big)
        \ar@{=>}^-{\simeq}_-{ \widetilde H_3 }  "s"; "t"
        \ar@{<=}^-{\simeq}_-{ \eta^\ast } "s2"; "t2"
      }
      }
    \end{equation}
  Here the right vertical morphism
  is the relative minimal Sullivan model (e.g. \cite[Prop. 3.17]{FSS20b})
  of the parametrized quaternionic Hopf fibration,
  which is a cofibration out of a cofibrant object
  (e.g. \cite[Prop. 3.43]{FSS20b}). Since, moreover, every dgc-algebra
  is projectively fibrant (e.g. \cite[Rem. 3.37]{FSS20b}),
  any homotopy as on the left in \eqref{DiagramExhibitinAnomalyFunctionalAsHopfInvariant}
  is represented by a homotopy $\eta^\ast$
  as shown on the right (e.g. \cite[Prop. A.16]{FSS20b}).

  \noindent {\bf (i)}
  The diagonal morphism on the right of
  \eqref{DiagramExhibitinAnomalyFunctionalAsHopfInvariant}
  manifestly exhibits a choice of gauging $\widetilde H_3$
  of $\widetilde f$. So to prove the first claim
  it just remains to see that this
  establishes a bijection on homotopy classes.
  Observe that a homotopy of homotopy lifts is now of this form:
  \vspace{-2mm}
    \begin{equation}
      \hspace{-.8cm}
      \raisebox{50pt}{
      \scalebox{.95}{
      \xymatrix@C=80pt{
        \Omega_{\mathrm{dR}}^\bullet
        \big(
          \widetilde \Sigma^7
        \big)
        \ar@/_2.8pc/@{<-}[drrr]|-{
          \mbox{
            \scalebox{0.7}{
            $
            \begin{array}{lcl}
              \phantom{f^\ast 2}(\widetilde H_3)_1
                &\!\!\!\!\mapsfrom\!\!\!\!&
              h_3
              \\
              \widetilde f^\ast \phantom{2}G_4
                &\!\!\!\!\mapsfrom\!\!\!\!&
              \omega_4
              \\
              \widetilde f^\ast 2G_7 - \widetilde f^\ast \Theta_7 &\!\!\!\!\mapsfrom\!\!\!\!& \omega_{\, 7}
            \end{array}
            $
            }
          }
        }^>>>>>>>>>>>{\ }="t2"
        \ar@/^2.8pc/@{<-}[drrr]|-{
          \mbox{
            \scalebox{0.7}{
            $
            \begin{array}{lcl}
              \phantom{f^\ast 2}(\widetilde H_3)_0
                &\!\!\!\!\mapsfrom\!\!\!\!&
              h_3
              \\
              \widetilde f^\ast \phantom{2}G_4
                &\!\!\!\!\mapsfrom\!\!\!\!&
              \omega_4
              \\
              \widetilde f^\ast 2G_7 - \widetilde f^\ast \Theta_7
                &\!\!\!\!\mapsfrom\!\!\!\!&
              \omega_{\, 7}
            \end{array}
            $
            }
          }
        }_>>>>>>>>>>>>>{\ }="s2"
        \ar@{<-}[dd]_{
          \widetilde f^\ast
        }
        &&&
        \\
        &&&
        \left(
          \!\!\!\!
          {\begin{array}{lcl}
            d\,h_3 &\!\!\!\!\!=\!\!\!\!\!& \omega_4 - \tfrac{1}{4}p_1
            \\
            d\,\omega_4 &\!\!\!\!\!=\!\!\!\!\!& 0
            \\
            d\,\omega_{\, 7}
              &\!\!\!\!\!=\!\!\!\!\!&
              - d h_3 \wedge \big( \omega_4 + \tfrac{1}{4}p_1\big)
              - \rchi_8
          \end{array}}
          \!\!\!\!
        \right)
        \ar@{<-^{)}}[d]^-{
          \mbox{
          \tiny
            $
            \arraycolsep=1.4pt
            \begin{array}{cc}
              \omega_4 & \omega_{\, 7}
              \\
              \mapsup & \mapsup
              \\
              \omega_4 & \omega_{\, 7}
            \end{array}
            $
                                  }
        }
        \\
        \Omega_{\mathrm{dR}}^\bullet(X)
        \ar@{<-}[rrr]^-{
          \mbox{
         \scalebox{0.7}{
            $
            \begin{array}{lcl}
              \phantom{2}G_4 &\!\!\!\mapsfrom\!\!\!& \omega_4
              \\
              2G_7 - \Theta_7 &\!\!\!\mapsfrom\!\!\!& \omega_{\, 7}
            \end{array}
            $
            }
          }
        }
        &&&
        \left(
          \!\!\!\!
          {\begin{array}{lcl}
            d\,\omega_4 &\!\!\!\!\!=\!\!\!\!\!& 0
            \\
            d\,\omega_{\, 7}
              &\!\!\!\!\!=\!\!\!\!\!&
              - \omega_4 \wedge \omega_4
              + \big(\tfrac{1}{4}p_1\big)^2
              - \rchi_8
          \end{array}}
          \!\!\!\!
        \right)
        \,.
        \ar@{<=}_-{ \eta^\ast } "s2"; "t2"
      }
      }
      }
    \end{equation}

    \vspace{-2mm}
\noindent
   Hence,
    since path space objects of de Rham dgc-algebras
    over $X$ are given by de Rham dgc-algebras over $X \times [0,1]$
    (e.g. \cite[Lem. 3.88]{FSS20b}),
    this is equivalently
    (e.g. \cite[Prop. A.16]{FSS20b})
    a dgc-algebra homomorphism making the following
    diagram commute
    under $\mathrm{CE}\big( \mathfrak{l}B \mathrm{Sp}(2) \big)$
    (where $s$ denotes the canonical coordinate function on $[0,1]$):
\vspace{-2mm}
    \begin{equation}
    \label{AlgebraicHomotopyForGugedGaugedTransformation}
    \hspace{-1cm}
    \raisebox{51pt}{
  \xymatrix@C=7em@R=3pt{
    \Omega_{\mathrm{dR}}^\bullet
    \big(\,
      \widetilde \Sigma^7
    \big)
    \ar@{<-}[d]_-{
      \mbox{
        \tiny
        $
        \arraycolsep=1.4pt
        \begin{array}{cc}
          0 & 0
          \\
          \mapsup & \mapsup
          \\
          s & d s
        \end{array}
        $
              }
    }
    \ar@{<-}[drr]^-{\;\;\;
      \mbox{
        \scalebox{0.65}{
        $
                \;\;\;\;\;\;\;\;\;\;
        \arraycolsep=1.4pt
        \begin{array}{lcl}
          \phantom{f^\ast 2}\big(\widetilde H_3\big)_0
          &\mapsfrom &
          h_3
          \\
          \widetilde f^\ast \phantom{2}G_4
          &\mapsfrom&
          \omega_4
          \\
          \widetilde f^\ast 2G_7 - \widetilde f^\ast \Theta_7
          &\mapsfrom&
          \omega_{\, 7}
        \end{array}
        $
        }
      }
    }
    \\
    \Omega_{\mathrm{dR}}^\bullet
    \Big(
      \widetilde \Sigma^7
      \times [0,1]
    \Big)
    \ar@{<-}[rr]|-{\; \eta^\ast }
    &&
    \left(
      \!\!\!
      {\begin{array}{lcl}
        d\,h_3 &=& \omega_4 - \tfrac{1}{4}p_1
        \\
        d\,\omega_4 &=& 0
        \\
        d\,\omega_{\, 7}
          &=&
          - \omega_4 \wedge \omega_4
          + \big( \tfrac{1}{4} p_1\big)^2
          \\
          & & - \rchi_8
      \end{array}}
      \!\!\!
    \right).
    \\
    \Omega_{\mathrm{dR}}^\bullet
    \big(\,
      {\widetilde{\Sigma}^7}
    \big)
    \ar@{<-}[u]^-{
      \mbox{
        \tiny
        $
        \arraycolsep=1.4pt
        \begin{array}{cc}
          s & d s
          \\
          \mapsdown & \mapsdown
          \\
          1 & 0
        \end{array}
        $
      }
    }
    \ar@{<-}[urr]_{\!\!\!\!\!\!\!\!
      \mbox{
        \scalebox{0.65}{
        $
         \arraycolsep=1.4pt
         \begin{array}{lcl}
           \mathclap{\phantom{\vert^{\vert^{\vert}}}}
           \phantom{f^\ast 2}\big(\widetilde H_3\big)_1 & \mapsfrom & h_3
           \\
           \widetilde f^\ast \phantom{2}G_4 & \mapsfrom & \omega_4
           \\
           \widetilde f^\ast 2G_7 - \widetilde f^\ast \Theta_7
             & \mapsfrom &
           \omega_{\, 7}
         \end{array}
        $
        }
      }
    }
    \\
    \big(\widetilde H_3\big)_{[0,1]}
    \ar@{<-|}[rr]
    &&
    h_3
  }
  }
  \end{equation}
  But this homotopy diagram \eqref{AlgebraicHomotopyForGugedGaugedTransformation}
  manifestly exhibits the same data and conditions as
  in \eqref{DataHomotopyGauged} for a homotopy of
  gaugings of a sigma-model field $\widetilde f$:

  \vspace{-.6cm}
  $$
    \xymatrix{
      \big(
        \widetilde f, (H_3)_0
      \big)
      \ar@{=>}[rr]^-{ \left( \mathrm{id}, (\widetilde H_3)_{[0,1]} \right) }
      &&
      \big(
        \widetilde f, (H_3)_1
      \big)
    }
    \,,
  $$
  and hence homotopy classes are equivalent to gauge equivalence classes,
  as claimed.

  \noindent {\bf (ii)}
  Consider in the following  any 7-form
  on $\widetilde \Sigma^7$ of unit volume:
  \begin{equation}
    \label{VolumeForm}
    \mathrm{vol}_{ \widetilde \Sigma^7 }
      \;\in\;
    \Omega_{\mathrm{dR}}^7\big( \widetilde \Sigma^7\big)
    \phantom{AA}
    \mbox{\rm such that}
    \phantom{AA}
    \underset{\widetilde \Sigma^7}{\int}
      \mathrm{vol}_{\widetilde \Sigma^7}
    \;=\;
    1\;.
  \end{equation}
  Again using the above path space objects,
  the homotopy $\eta^\ast$
  on the right in \eqref{DiagramExhibitinAnomalyFunctionalAsHopfInvariant} is a dgc-algebra homomorphism
  that makes the following diagram commute
  under
  $\mathrm{CE}\big( \mathfrak{l}B \mathrm{Sp}(2) \big)$:

\begin{equation}
  \label{ThedgcHomotopy}
  \hspace{-1.5cm}
  \raisebox{54pt}{
  \xymatrix@C=7em@R=2pt{
    \Omega_{\mathrm{dR}}^\bullet
    \big(\,
      \widetilde \Sigma^7
    \big)
    \ar@{<-}[d]_-{
      \mbox{
        \tiny
        $
        \arraycolsep=1.4pt
        \begin{array}{cc}
          0 & 0
          \\
          \mapsup & \mapsup
          \\
          s & d s
        \end{array}
        $
      }
    }
    \ar@{<-}[drr]^>>>>>>>>>>>>>>>>>>>>>>>>>>>>>{\;\;\;
      \mbox{
        \scalebox{0.7}{
        $
        \arraycolsep=1.5pt
        \begin{array}{ccl}
          \left({\color{darkblue} 2\widetilde S}{-\int_{\tilde{\Sigma}^7}f^*\Theta_7}\right)
            \cdot
          \mathrm{vol}_{\widetilde \Sigma^7}
          & \mapsfrom &
          \omega_{\, 7}
          \\
          \tfrac{1}{4} \widetilde f^\ast p_1(\nabla)
            &\mapsfrom&
            \omega_4
          \\
          0 &\mapsfrom& h_3
        \end{array}
        $
        }
      }
    }
    \\
    \Omega_{\mathrm{dR}}^\bullet
    \Big(
      \widetilde \Sigma^7
      \times [0,1]
    \Big)
    \ar@{<-}[rr]|-{\; \eta^\ast }
    &&
    \left(
      \!\!\!
      {\begin{array}{lcl}
        d\,h_3 &=& \omega_4 - \tfrac{1}{4}p_1
        \\
        d\,\omega_4 &=& 0
        \\
        d\,\omega_{\, 7}
          &=&
          - \omega_4 \wedge \omega_4
          + \big( \tfrac{1}{4} p_1\big)^2
          \\
          & & - \rchi_8
      \end{array}}
      \!\!\!
    \right).
    \\
    \Omega_{\mathrm{dR}}^\bullet
    \big(
      S^7
    \big)
    \ar@{<-}[u]^-{
      \mbox{
        \tiny
        $
        \arraycolsep=1.4pt
        \begin{array}{cc}
          s & d s
          \\
          \mapsdown & \mapsdown
          \\
          1 & 0
        \end{array}
        $
      }
    }
    \ar@{<-}[urr]_<<<<<<<<<<<<<<<<<<<<<<<<<<<<<<<<<{\!\!\!\!
      \mbox{
        \scalebox{0.7}{
        $
         \arraycolsep=1.4pt
         \begin{array}{lcl}
           \mathclap{\phantom{\vert^{\vert^{\vert^{\vert}}}}}
           \phantom{f^\ast 2}\widetilde H_3 & \mapsfrom & h_3
           \\
           \widetilde f^\ast \phantom{2}G_4
             & \mapsfrom &
           \omega_4
           \\
           \widetilde f^\ast 2G_7 - \widetilde f^\ast \Theta_7
             & \mapsfrom &
           \omega_{\, 7}
         \end{array}
        $
        }
      }
    }
    }
    }
\end{equation}

\noindent
We claim that such an $\eta^\ast$ is given by:
\vspace{-2mm}
 $$
 \hspace{-2mm}
   \begin{array}{lcl}
         s \widetilde H_3
          &\!\longmapsfrom\!&
          h_3
          \\
      ds \wedge \widetilde H_3
          +
          s
          \cdot
          \widetilde f^\ast\big( G_4 \big)
          +
          (1-s)
          \tfrac{1}{4} \widetilde f^\ast\big( p_1(\nabla) \big)
          &\overset{\eta^\ast}{\!\longmapsfrom\!}&
          \omega_4
          \\
          s
            \cdot
          \Big(
            \widetilde f^\ast\big( 2G_7 \big)
            -
            \widetilde f^\ast \big( \Theta_7 \big)
          \!\! \Big)
          +
          \Big(
            {
              \color{darkblue}
              2\widetilde S
            }
            -
            \int_{\tilde{\Sigma}^7}
            \tilde{f}^*\big( \Theta_7 \big)
          \!\! \Big)\!
          \cdot \!
          (1-s)\cdot \mathrm{vol}_{ \widetilde \Sigma^7 }
          +
          s (1-s) \cdot
          \widetilde H_3
            \wedge
          \widetilde f^\ast
          \big(
            G_4
            -
            \tfrac{1}{4}p_1(\nabla)
          \big)
          +
          d s \wedge Q_6
          &\!\longmapsfrom\!&
          \omega_{\, 7}
        \end{array}
  $$

\noindent
where $Q_6 \in \Omega_{\mathrm{dR}}^6\big(\, \widetilde \Sigma^7 \big)$
is any differential form which satisfies
\vspace{-3mm}
$$
  d Q_6
  \;=\;
  \Big(
    \widetilde H_3
      \wedge
    \widetilde f^\ast
    \big(
      G_4 + \tfrac{1}{4}p_1(\nabla)
    \big)
    +
    \widetilde f^\ast
    \big(
      2G_7
      -
      \Theta_7
    \big)
  \Big)
  -
  \overset{
    =:
    {
      \color{darkblue}2\widetilde S
    }
    -\int_{\tilde{\Sigma}^7}\tilde{f}^*( \Theta_7 )
  }{
  \Bigg(\;
  \overbrace{
  \underset{ \widetilde \Sigma^7 }{\int}
    \Big(
    \widetilde H_3
      \wedge
    \widetilde f^\ast
    \big(
      G_4 + \tfrac{1}{4}p_1(\nabla)
    \big)
    +
    \widetilde f^\ast
    \big(
      2G_7
      -
      \Theta_7
    \big)
  \Big)
  }
 \; \Bigg)
  }
  \cdot
  \mathrm{vol}_{S^7}
  \,.
$$

\vspace{-1mm}
\noindent
This exists by \eqref{VolumeForm} and
because cohomology classes of differential forms in top degree on
compact connected manifolds are in bijection with the values of their integrals
(e.g. \cite[\S 7.3, Thm. 7.5]{Lafontaine15}).
It is clear that $\eta^\ast$ thus defined satisfies
the required boundary conditions of
a homotopy in \eqref{ThedgcHomotopy}.
Hence it only
remains to check that it is indeed a dg-algebra homomorphism,
in that it respects the differentials on the generators.
This is verified by direct computation:
\vspace{-1mm}
$$
\openup-.1\jot
  \begin{aligned}
    d \eta^\ast(\omega_{\, 7})
    & =
    d
    \Bigg(\!\!
      s
      \cdot \!
      \Big(
        \widetilde f^\ast \big( 2G_7 \big)
        -
        \widetilde f^\ast \big( \Theta_7 \big)
     \!\! \Big)
      +
      (1-s)
      \cdot \!
      \left(\! {\color{darkblue} 2\widetilde S }{ - \!\!\int_{\tilde{\Sigma}^7}\tilde{f}^*\Theta_7} \!\! \right)
      \! \cdot \mathrm{vol}_{S^7}
      +
      s (1-s) \cdot \widetilde H_3
        \wedge
      \widetilde f^\ast
      \big(
        G_4
        -
        \tfrac{1}{4}p_1(\nabla)
      \big)
      +
      d s \wedge Q_6
    \!\! \Bigg)
    \\
    & =
    d s \wedge
    \Big(
      \widetilde f^\ast \big( 2G_7 \big)
      -
      \widetilde f^\ast \big( \Theta_7 \big)
    \Big)
    -
    d s \wedge
    \Big(
      {
        \color{darkblue} 2\widetilde S
      }
      -
      \int_{\tilde{\Sigma}^7}\tilde{f}^*\big( \Theta_7 \big)
    \Big)
    \cdot
    \mathrm{vol}_{S^7}
    +
    d s \wedge \widetilde H_3
      \wedge
    \widetilde f^\ast \big( G_4 - \tfrac{1}{4}p_1(\nabla) \big)
    \\
    &
    \phantom{=}\;
    -
    2 s \cdot d s
      \wedge
    \widetilde H_3
      \wedge
    \widetilde f^\ast \big( G_4 - \tfrac{1}{4}p_1(\nabla) \big)
    -
    d s \wedge d Q_6
    \\
    & =
    d s
    \wedge
    \left(
      \big(
        \widetilde f^\ast ( 2G_7 )
        -
        \widetilde f^\ast ( \Theta_7 )
      \big)
      +
      \widetilde H_3
        \wedge
      \widetilde f^\ast
      \big(
        G_4 - \tfrac{1}{4}p_1(\nabla)
      \big)
      -
      \left(
        {\color{darkblue}2\widetilde S}
        -
        \int_{\widetilde \Sigma^7}
        \tilde f^\ast(\Theta_7)
      \right)
        \cdot
        \mathrm{vol}_{S^7}
      -
      d Q_6
    \right)
    \\
    &
    \phantom{=}
    -
    2 s \cdot d s \wedge \widetilde H_3
      \wedge
    \widetilde f^\ast \big( G_4 - \tfrac{1}{4}p_1(\nabla) \big)
    \\
    & =
    d s
    \wedge
    \underset{
      = 0
    }{
      \underbrace{
      \left(
        \widetilde f^\ast
        \big(
          2G_7 - \Theta_7
        \big)
        +
        \widetilde H_3
          \wedge
        \widetilde f^\ast
        \big(
          G_4 {\color{darkblue}+} \tfrac{1}{4}p_1(\nabla)
        \big)
        -
        \left({\color{darkblue} 2\widetilde S }{-\int_{\tilde{\Sigma}^7}\tilde{f}^*\Theta_7}\right) \cdot \mathrm{vol}_{S^7}
        -
        d Q_6
      \right)
      }
    }\\
    &
    \phantom{=}\;
    -
    2 \cdot
    d s \wedge \widetilde H_3
      \wedge
    \widetilde f^\ast \big( s \cdot G_4 + (1-s)\tfrac{1}{4}p_1(\nabla) \big)
    \\
    & =
    - \eta^\ast\big( \omega_4\big) \wedge \eta^\ast\big( \omega_4 \big)
    \\
    & =
    \eta^\ast\big(  d \omega_{\, 7}  \big).
  \end{aligned}
 $$

\vspace{-1mm}
\noindent
Here the crucial non-trivial step is the fourth (third line from below).
In the last two steps we used that all 8-forms
on $\widetilde \Sigma^7$ vanish, so that only the 8-forms
on $\widetilde \Sigma^7 \times [0,1]$ with one factor of
$d s$ survive.

\noindent The verification on the other two generators is immediate:
\vspace{-1mm}
$$
  \begin{aligned}
    d \eta^\ast(h_3)
    & =
    d
    \big(
      s \widetilde H_3
    \big)
    \\
    & =
    d s \wedge \widetilde H_3
    +
    s
      \cdot
    \Big(
      \widetilde f^\ast\big( G_4 \big)
      -
      \tfrac{1}{4} \widetilde f^\ast \big( p_1(\nabla) \big)
    \Big)
    \\
    & =
    \eta^\ast
    (
      \omega_4
    )
    -
    \tfrac{1}{4} \widetilde f^\ast\big( p_1(\nabla) \big)
    \\
    & =
    \eta^\ast
    (
      d h_3
    )
    \,,
\end{aligned}
\qquad \qquad
\begin{aligned}
    d \eta^\ast( \omega_4 )
    & =
    -
    d s \wedge d \widetilde H_3
    \\
    & \phantom{=}\;
    +
    d
    \Big(
      s \cdot \widetilde f^\ast\big( G_4 \big)
      +
      (1-s)\tfrac{1}{4} \widetilde f^\ast\big( p_1(\nabla) \big)
    \Big)
    \\
    & = 0
    \\
    & = \eta^\ast
    (
      d \omega_4
    )
    \,.
    \\
 \end{aligned}
$$

\vspace*{-1.2\baselineskip}
\end{proof}

\newpage

\section{Hypothesis H implies M5 Hopf-WZ anomaly cancellation}
\label{HypothesisHImpliesAnomalyCancellation}

In view of the rational cohomotopical
interpretation of background C-field
(Remark \ref{BackgroundFieldsAsRationalCohomotopy})
and of the 6d Hopf-WZ anomaly functional
(Theorem \ref{AnomalyFunctionalAsLiftInCohomotopy})
it is natural to hypothesize that
the topological sector of the background C-field
should be required to be a cocycle in actual twisted Cohomotopy.
This non-abelian charge-quantization condition
(\cite[\S 5.3]{FSS20b})
is called \hyperlink{HypothesisH}{\it Hypothesis H} in \cite{FSS19b};
we recall the precise statement as
Def. \ref{BackgroundFieldsSatisfyingHypothesisH} below.

\medskip
We observe in Prop. \ref{GWIIsHomotopyInvariantOfBackgroundFieldData}
that,
under \hyperlink{HypothesisH}{\it Hypothesis H}
and in the absence of topological twisting,
Theorem \ref{AnomalyFunctionalAsLiftInCohomotopy}
exhibits the M5 Hopf-WZ anomaly functional
as the homotopy Whitehead integral formula
(see Remark \ref{WhiteheadIntegralFormulasInTheLiterature} below)
for the Hopf invariant (recalled in Def. \ref{TheHopfInvariant} below).
This proves the anomaly cancellation \eqref{ConsistencyCondition}
for the special case
of oriented differences of extended worldvolumes
being the 7-sphere and for vanishing topological twist $\tfrac{1}{4}p_1$
(Remark \ref{ShiftedFluxQuantizationCorrectionsToTheHopfWZTerm}).
Finally we establish a twisted/parametrized generalization of the
integral Hopf invariant in Theorem \ref{AnomalyIsIntegral},
which proves the anomaly cancellation condition
\eqref{ConsistencyCondition} generally.

\begin{defn}[{\it Hypothesis H} {\cite{FSS19b}}]
 \label{BackgroundFieldsSatisfyingHypothesisH}
 In the situation of Def. \ref{BasicSetup} we say that:

\vspace{-1mm}
 \item {\bf (i)} the background fields $(G_4, 2G_7)$  \eqref{SpacetimeFormAssumption}
 {\it satisfy \hyperlink{HypothesisH}{\it Hypothesis H}}
 if they are classified as in \cite[Def. 3.5]{FSS19b}
 by an actual cocycle $c$ in twisted Cohomotopy \cite[\S 2.1]{FSS19b},
 hence if their classifying map in rational twisted Cohomotopy
 from Remark \ref{BackgroundFieldsAsRationalCohomotopy}
 factors, up to homotopy, through the homotopy quotient
 $S^4 \!\sslash\! \mathrm{Sp}(2)$ of the
 4-sphere canonically acted on by
 $\mathrm{Sp}(2) \simeq \mathrm{Spin}(5)$
 (see \cite[Prop. 2.1]{FSS20a}), followed by the rationalization map
 \eqref{Rationalization};

\vspace{-1mm}
 \item  {\bf (ii)}  the (extended or not) higher gauged
  sigma-model fields $(\widetilde f, \widetilde H_3)$  \eqref{GaugedExtendedSigmaModelFields}
  {\it satisfy \hyperlink{HypothesisH}{\it Hypothesis H}}
  if the corresponding lift \eqref{DiagramExhibitinAnomalyFunctionalAsHopfInvariant}
  through the rationalized parametrized
  quaternionic Hopf fibration, which classifies them
  by Theorem \ref{AnomalyFunctionalAsLiftInCohomotopy},
  factors as a lift through the actual parametrized
  quaternionic Hopf fibration $h_{\mathbb{H}} \!\sslash\! \mathrm{Sp}(2)$
  (\cite[Prop. 2.22]{FSS19b}):
\vspace{-1mm}
 \begin{equation}
   \label{CohomotopyCocycle}
   \hspace{-4mm}
   \raisebox{45pt}{
   \xymatrix@C=6em@R=30pt{
     \widetilde \Sigma^7
     \ar[dd]_-{\widetilde f}^>>>>>>{\ }="t"
     \ar@{-->}[rr]|-{\;\;
       \widehat { c \circ \widetilde f }
     \;\;}_<<<<<<<<<<<<<<<<{\ }="s"
     \ar@/^2pc/[rrrr]|-{
       \overset{
         \raisebox{3pt}{
           \tiny
           \color{darkblue}
           \bf
           rational Hopf-WZ term
         }
       }{
       \;
       2 \widetilde S
       (
         \widetilde H_3,
         \,
         G_4\,
         2G_7
       )
       \;
       }
     }
     &&
     S^7
     \!\sslash\!
       \mathrm{Sp}(2)
     \ar[rr]|-{
       \;\;
       \mbox{
         \tiny
         \color{greenii}
         \bf
         rationalization
       }
       \;\;
     }
     \ar[dd]|-{
       \overset{
         \raisebox{3pt}{
           \tiny
           \color{greenii}
           \bf
           \begin{tabular}{c}
             $\mathrm{Sp}(2)$-parametrized
             \\
             quaternionic Hopf fibration
           \end{tabular}
         }
       }{
         h_{\mathbb{H}}
         \sslash
         \mathrm{Sp}(2)
       }
     }
     &&
     L_{\mathbb{R}}
     \big(
       S^7
       \!\sslash\!
       \mathrm{Sp}(2)
     \!\big)
     \ar[dd]^-{
       L_{\mathbb{R}}
       (
       h_{\mathbb{H}}
       \sslash
       \mathrm{Sp}(2)
       )
     }
     \\
     \ar@{}[rr]|-{
       \mbox{
         \tiny
         \color{darkblue}
         \bf
         \begin{tabular}{c}
           lift to actual
           \\
           twisted Cohomotopy
         \end{tabular}
       }
     }
     &&
     \ar@{}[rr]|-{
       \mbox{
         \tiny
         \color{darkblue}
         \bf
         \begin{tabular}{c}
           rational
           \\
           twisted Cohomotopy
         \end{tabular}
       }
     }
     &&
     \\
     X
     \ar[dr]_-{ \tau }
     \ar@/_2pc/[rrrr]_>>>>>>>>>>>>>>>>>>>>>{
       \underset{
         \raisebox{-3pt}{
           \tiny
           \color{darkblue}
           \bf
           \begin{tabular}{c}
             cocycle in rational
             \\
             twisted Cohomotopy
           \end{tabular}
         }
       }{
         (G_4, 2G_7)
       }
     }
     \ar@{-->}[rr]|-{\; c \; }
     &&
     S^4
     \!\sslash\!
     \mathrm{Sp}(2)
     \ar[rr]|-{
       \;\;
       \mbox{
         \tiny
         \color{greenii}
         \bf
         rationalization
       }
       \;\;
     }
     \ar[dl]|-{\;\;\;\;\;\;\;\;\; { {\phantom{A} \atop \phantom{A}} \atop \phantom{A} } }
     &&
     L_{\mathbb{R}}
     \big(
       S^4
       \!\sslash\!
       \mathrm{Sp}(2)
     \big)
     \\
     &
     B
     \mathrm{Sp}(2)
     \ar@{=>}^-{ \widetilde H_3 } "s"; "t"
   }
   }
 \end{equation}
\end{defn}

\noindent {\bf Hopf-WZ term in terms of Fivebrane-extended $\mathrm{Sp}(2)$-structure }
For transparent formulation of the proof of the following integrality theorem (Theorem \ref{AnomalyIsIntegral} below),
it is useful to re-cast the result of Theorem \ref{AnomalyFunctionalAsLiftInCohomotopy}
in terms of the
Fivebrane-extended $\widehat{\mathrm{Sp}(2)}$-structure
from Example \ref{ClassifyingSpaceOfHigherExtendedQuaternionUnitaryGroup}:

\begin{defn}[Quaternionic Hopf fibration parametrized over Fivebrane-extended $\mathrm{Sp}(2)$]
  \label{QuaternionicHopfFibrationParametrizedOverHigherExtendedSp2}
Consider the homotopy pullback (e.g. \cite[Def. A.23]{FSS20b})
of the $\mathrm{Sp}(2)$-parametrized quaternionic Hopf fibration
\eqref{CohomotopyCocycle} along the Fivebrane-extension
$\widehat {\mathrm{Sp}(2)} \to \mathrm{Sp}(2)$ (Example \ref{ClassifyingSpaceOfHigherExtendedQuaternionUnitaryGroup}).
By the pasting law and the homotopy-restriction map on $\infty$-actions
(see \cite[Prop. 2.23, 2.85]{SS20b}),
we may denote this as follows:
\vspace{-1mm}
\begin{equation}
  \label{PullbackToUniversalSpaceOnWhichChi8Trivializes}
  \raisebox{20pt}{
  \xymatrix@R=8pt@C=4em{
    &
    S^7 \!\sslash\! \widehat {\mathrm{Sp}(2)}
    \ar@{}[ddrr]|-{
      \mbox{
        \tiny
        (pb)
      }
    }
    \ar[rr]
    \ar[dd]_-{
      \mathllap{
        \mbox{
          \tiny
          \color{greenii}
          \bf
          \begin{tabular}{c}
            $\widehat{\mathrm{Sp}(2)}$-parametrized
            \\
            quaternionic Hopf fibration
          \end{tabular}
        }
      }
      \scalebox{0.6}{$
        h_{\mathbb{H}} \sslash \widehat{\mathrm{Sp}(2)}
      $}
    }
    &&
    S^7 \!\sslash\! \mathrm{Sp}(2)
    \ar[dd]^-{\scalebox{0.6}{$
      h_{\mathbb{H}} \sslash \mathrm{Sp}(2)
      $}
      \mathrlap{
        \mbox{
          \tiny
          \color{darkblue}
          \bf
          \begin{tabular}{c}
            $\mathrm{Sp}(2)$-parametrized
            \\
            quaternionic Hopf fibration
          \end{tabular}
        }
      }
    }
    \\
    \\
    &
    S^4 \!\sslash\! \widehat{\mathrm{Sp}(2)}
    \ar@{}[ddrr]|-{
      \mbox{
        \tiny
        (pb)
      }
    }
    \ar[dd]_-{ \widehat \rho_{S^4} }
    \ar[rr]
    &&
    S^4 \!\sslash\! \mathrm{Sp}(2)
    \ar[dd]^-{ \rho_{S^4} }
    \\
    \\
    &
    \mathllap{
    \mbox{
      \tiny
      \color{darkblue}
      \bf
      \begin{tabular}{c}
        classifying space for
        \\
        Fivebrane-extended
        $\mathrm{Sp}(2)$-structure
      \end{tabular}
    }
    }
    B \widehat {\mathrm{Sp}(2)}
    \ar[rr]_-{\scalebox{0.7}{$ \mathrm{hofib}(\rchi_8)$} }
    &&
    B \mathrm{Sp}(2)
    \mathrlap{
    \mbox{
      \tiny
      \color{darkblue}
      \bf
      \begin{tabular}{c}
        classifying space
        for $\mathrm{Sp}(2)$-structure
        \\
        (BPS M2-brane backgrounds)
      \end{tabular}
    }
    }
    &
  }
  }
\end{equation}
\vspace{-.4cm}

Notice, by Example \ref{ClassifyingSpaceOfHigherExtendedQuaternionUnitaryGroup},
that the minimal relative Sullivan model for the
$\widehat{\mathrm{Sp}(2)}$-parametrized quaternionic Hopf fibration
on the left of \eqref{PullbackToUniversalSpaceOnWhichChi8Trivializes}
is just like that of the $\mathrm{Sp}(2)$-parametrized
Hopf fibration on the right of \eqref{PullbackToUniversalSpaceOnWhichChi8Trivializes},
as given in \eqref{DiagramExhibitinAnomalyFunctionalAsHopfInvariant},
except that all entries have the generator
$\theta_7$ adjoined, with $d \theta_7 = \rchi_8$
(universal rational Fivebrane-structure).

\end{defn}

\begin{remark}[The Hopf-WZ term in terms of Fivebrane $\widehat{\mathrm{Sp}(2)}$-structure ]
 \label{TheHopfWZTermOverBFivebrane}
 In terms of the $\widehat{\mathrm{Sp}(2)}$-parametrized
 quaternionic Hopf fibration (Def. \ref{QuaternionicHopfFibrationParametrizedOverHigherExtendedSp2}),
 the content of Theorem \ref{AnomalyFunctionalAsLiftInCohomotopy}
 becomes the following transparent statement:

\noindent {\bf (i)} By the assumption \eqref{EulerFormAsPfaffian}
 that the Euler 8-class of $X^8$ is equipped with a
 trivialization, hence that we have
 $\widehat{\mathrm{Sp}(2)}$-structure $\widehat \tau$
 \eqref{hatSp2Structure}
 on $X^8$, we may pull back the situation in
 \eqref{DiagramExhibitinAnomalyFunctionalAsHopfInvariant}
 along
 ${ B \widehat{\mathrm{Sp}(2)} \xrightarrow{\mathrm{hofib}(\rchi_8)}  B \mathrm{Sp}(2) }$
 and regard the twisted Cohomotopy classes $c$
 and $\widehat{ c \circ \widetilde f }$ as having local coefficients
 in the $\widehat{\mathrm{Sp}(2)}$-parametrized quaternionic Hopf fibration
 from Def. \ref{QuaternionicHopfFibrationParametrizedOverHigherExtendedSp2}.

\noindent {\bf (ii)} In this formulation, Theorem \ref{AnomalyFunctionalAsLiftInCohomotopy},
 says equivalently that there is a single rational 7-class
 on the total space of the universal $\widehat{\mathrm{Sp}(2)}$-parametrized
 7-sphere fibration, represented by the sum of the
 generators \eqref{TopRightomega7} and \eqref{HomotopyFiberSpace}:

 \vspace{-.6cm}
 \begin{equation}
   \label{UniversalRationalHopfWZFlux}
   \overset{
     \mathclap{
     \hspace{24pt}
     \rotatebox[origin=c]{34}{
       \tiny
       \color{greenii}
       \bf
       \def\arraystretch{.8}
       \begin{tabular}{c}
         universal
         \\
         Hopf-WZ term
       \end{tabular}
     }
     }
   }{
     2 \widetilde {\mathbf{S}}
   }
   \;:=\;
   \overset{
     \mathclap{
     \hspace{24pt}
     {
     \rotatebox[origin=c]{34}{
       \tiny
       \color{darkblue}
       \bf
       \def\arraystretch{.8}
       \begin{tabular}{c}
         fiberwise
         \\
         volume form
         on $S^7$
       \end{tabular}
     }
     }
     }
   }{
     \widetilde \omega_{\, 7}
   }
   \;\;
     +
   \;\;
   \overset{
     \mathclap{
     {
     \hspace{27pt}
     \rotatebox[origin=c]{34}{
       \tiny
       \color{darkblue}
       \bf
       \def\arraystretch{.8}
       \begin{tabular}{c}
         universal
         \\
         Fivebrane structure
       \end{tabular}
     }
     }
     }
   }{
     \theta_7
   }
   \;\;\;\;
   \in
   \;
   \mathrm{CE}
   \Big(
     \mathfrak{l}
     \big(
       S^7 \!\sslash\! \widehat{\mathrm{Sp}(2)}
     \big)
   \Big)
   \,,
   \;\;\;\;\;\;\;\;\;
   \big[
     \widetilde {\mathbf{S}}
   \big]
   \;\in\;
   H^7
   \Big(
     S^7 \!\sslash\! \widehat{\mathrm{Sp}(2)}
     ;
     \;
     \mathbb{R}
   \Big)
 \end{equation}
 which is the universal avatar of the Hopf-WZ term
 (Def. \ref{DefinitionOfWZWTermByFieldExtensionToCoboundary}, under
 $H^7_{\mathrm{dR}}\big( \widetilde \Sigma^7 \big) \,\simeq\, \mathbb{R}$):

\vspace{-6mm}
\begin{equation}
  \label{Rational7ClassOnE7}
  \hspace{-3mm}
  \xymatrix@C=12pt@R=16pt{
    \widetilde \Sigma^7
    \ar[rr]^-{\scalebox{0.7}{$
      \widehat{
        c \,\circ\, \tilde f
      }
      $}
    }_>>>{\ }="s"
    \ar[d]_-{
      \widetilde f
    }^>>>{\ }="t"
    &&
    L_{\mathbb{R}}
    \big(
      S^7 \!\sslash\! \widehat{\mathrm{Sp}(2)}
    \big)
    \mathrlap{\;,}
    \ar[d]|-{
      \mathclap{\phantom{\vert^{\vert^{\vert}}}}
     L_{\mathbb{R}}
      (
        h_{\mathbb{H}}
        \sslash
        \widehat {\mathrm{Sp}(2)}
      )
           \mathclap{\phantom{\vert_{\vert_{\vert}}}}
    }
    &&
    \Omega^\bullet_{\mathrm{dR}}
    \big(\,
      \widetilde \Sigma^7
    \big)
    \ar@{<-}[rr]^-{
      \overset{
        \mathclap{
          \raisebox{3pt}{
            \tiny
            \color{darkblue}
            \bf
            Hopf-WZ term
          }
        }
      }{
        2 \widetilde S
        (
          \widetilde H_3, G_4, G_7
        )
      }
      \;\;\mapsfrom\;\;
      \overset{
        \mathclap{
        \raisebox{3pt}{
          \tiny
          \color{greenii}
          \bf
          \begin{tabular}{c}
            universal
            \\
            Hopf-WZ term
          \end{tabular}
        }
        }
      }{
        \big[
        2{\widetilde {\mathbf{S}}}
        :=
        \widetilde \omega_{\, 7} + \theta_7
        \big]
      }
    }_>>>{\ }="s2"
    \ar@{<-}[d]_-{ \widetilde f^\ast }^>>>{\ }="t2"
    &&
    \left(
      {\begin{aligned}
      d\, \theta_7 & = \phantom{+} \rchi_8
      \\
      d\, \widetilde \omega_{\, 7} & = - \rchi_8
      \end{aligned}}
    \right)
    \ar@{<-}[d]
    \\
    X^8
    \ar[dr]_-{
      \underset{
        \mathllap{
          \mbox{
          \tiny
          \color{darkblue}
          \bf
          \begin{tabular}{c}
            Fivebrane
            \\
            structure
          \end{tabular}
          }
        }
      }{
        \widehat \tau
        \;\;\;
      }
    }^>>>>>>>>>{\ }="t3"
    \ar[rr]|-{
      \;\;c\;\;
    }^-{
      \mbox{
        \tiny
        \color{darkblue}
        \bf
        \begin{tabular}{c}
          cocycle in rational
          \\
          twisted 4-Cohomotopy
        \end{tabular}
      }
    }_>>>{\ }="s3"
    &&
    L_{\mathbb{R}}
    \big(
      S^4 \!\sslash\! \widehat{\mathrm{Sp}(2)}
    \big)
    \ar[dl]
    &&
    \Omega^\bullet_{\mathrm{dR}}\big(X^8\big)
    \ar@{<-}[rr]|-{
      \;\;
      \overset{
        \raisebox{3pt}{
          \tiny
          \color{darkblue}
          \bf
          fluxes
        }
      }{
      \scalebox{.7}{$
        \arraycolsep=1.4pt
        {\begin{array}{rcl}
          2G_7 - \Theta_7 &\mapsfrom& \omega_{\, 7}
          \\
          G_4 &\mapsfrom& \omega_4
          \\
          {\phantom{a}}
        \end{array}}
      $}
      }
      \;\;
    }
    \ar@{<-}[dr]
    &&
    \left(
    {\begin{aligned}
      d\, \theta_7 & = \rchi_8
      \\
      d\, \omega_{\, 7}
        & =
        - \omega_4 \wedge \omega_4
        + \big( \tfrac{1}{4}p_1\big)^2
        \\
        &
        \phantom{=}\,
        - \rchi_8
      \\
      d\, \omega_4 & = 0
    \end{aligned}}
    \right)
    \ar@{<-^{)}}[dl]
    \\
    &
    B \widehat{\mathrm{Sp}(2)}
    &&
    &&
    \left(
    {\begin{aligned}
      d\, \theta_7 & = \rchi_8
      \\
      d\, \rchi_8 & = 0
      \\
      d\, \tfrac{1}{2}p_1 & = 0
    \end{aligned}}
    \right)
    \ar@{=>}_{  } "s"; "t"
    \ar@{=>}^{ \widetilde H_3 } "t2"; "s2"
    \ar@{=>}_{  } "s3"; "t3"
  }
\end{equation}
\end{remark}

\noindent {\bf The plain Hopf invariant via homotopy Whitehead integral.}
Before analyzing the implications of
\hyperlink{HypothesisH}{\it Hypothesis H}
in the general twisted/shifted case
(Remark \ref{ShiftedFluxQuantizationCorrectionsToTheHopfWZTerm}),
we recall the
definition of the plain Hopf invariant
and show how this is computed by the untwisted Hopf-WZ term,
as a homotopy Whitehead integral.
\begin{defn}[Hopf invariant (e.g. {\cite[\S 4]{MosherTangora08}})]
  \label{TheHopfInvariant}
  For $k \in \mathbb{N}$ with $k \geq 1$, let
  \vspace{-2mm}
  \begin{equation}
    \label{MapPhiBetweenSpheres}
    \xymatrix{
      S^{4k-1}
      \ar[r]^-{ \phi }
      &
      S^{2k}
    }
  \end{equation}

    \vspace{-2mm}
\noindent  be a continuous function between higher dimensional spheres, as shown.
  Then the homotopy cofiber space of $\phi$ has integral cohomology given by
       \vspace{-2mm}
  $$
    H^p
    \big(
      \mathrm{cofib}
      (
        \phi
      ), \Z
    \big)
    \;\simeq\;
    \left\{
      \begin{array}{ccc}
        \mathbb{Z} &\vert& p \in \{ 2k , 4k\}
        \\
        0 &\vert& \mbox{otherwise}
      \end{array}
    \right.
     $$

    \vspace{-2mm}
\noindent
Hence, with generators denoted
\vspace{-2mm}
  $$
    \omega_{2k}
    \;:=\;
    \pm 1
    \in
    \mathbb{Z}
    \;\simeq\;
    H^{2k}
    \big(
      \mathrm{cofib}(\phi), \Z
    \big)
    \,,
    \phantom{A}
    \omega_{4k}
    \;:=\;
    \pm 1
    \in
    \mathbb{Z}
    \;\simeq\;
    H^{4k}\big(
      \mathrm{cofib}(\phi), \Z
    \big),
  $$
  there exists a unique integer
  \begin{equation}
    \label{HopfInvariant}
    \mathrm{HI}(\phi)
    \;\in\;
    \mathbb{Z}
    \,,
    \phantom{AA}
    \mbox{\rm such that}
    \phantom{AA}
    \omega_{2k} \cup \omega_{2k}
    \;=\;
    \mathrm{HI}(\phi)
    \cdot
    \omega_{4k}
  \end{equation}
  relating the cup-product square of the first to a multiple of the second.
  This integer is called the \emph{Hopf invariant} $\mathrm{HI}(\phi)$ of $\phi$.
  It depends on the choice of generators only up to a sign.
\end{defn}

We make the following basic observation:

\begin{lemma}[Recognition of Hopf invariants from Sullivan models]
  \label{RecognitionOfHopfInvariantFromSullivanModel}
  The unique coefficient in the minimal Sullivan model for a map $\phi$ of spheres
  as in \eqref{MapPhiBetweenSpheres}
  is the Hopf invariant $\mathrm{HI}(\phi)$ (Def. \ref{TheHopfInvariant}):
  \vspace{-2mm}
  \begin{equation}
    \label{CoefficientOfCEphiIsHopfInvariant}
    \xymatrix@R=-5pt@C=41pt{
      S^{4k-1}
      \ar[rrr]^-{ \phi }
      &&&
      S^{2k}
      \\
      \scalebox{0.85}{$
      \left(
      {\begin{aligned}
        d\, \omega_{4k-1} & = 0 \!\!\!\!\!\!\!
      \end{aligned}}
      \right)
      $}
      \ar@{<-}[rrr]|-{
        \tiny
        \;\;
        \begin{aligned}
          {\color{darkblue} \mathrm{HI}(\phi)}
          \cdot
          \omega_{4k-1}
          &
          \mapsfrom \omega_{4k-1}
          \\
          0 & \mapsfrom \omega_{2n}
          \\
        \end{aligned}
        \;\;
      }
      &&&
      \scalebox{0.85}{$
      \left(\!\!
      {\begin{array}{lcl}
        d \, \omega_{4k-1} & = &  - \omega_{2k} \wedge \omega_{2k}
        \\
        d \, \omega_{4k} & = &  0
      \end{array}}
      \!\!\right).
      $}
    }
  \end{equation}
\end{lemma}
\begin{proof}
  By cofibrant replacement in the classical model structure on
  topological spaces (e.g. \cite[Ex. A.7]{FSS20b}),
  the homotopy cofiber is represented by the ordinary pushout of topological spaces
  as shown in the following diagram on the left.
  By the fundamental theorem of dg-algebraic rational homotopy theory
  \cite[Thm. 9.4]{BousfieldGugenheim76} (see \cite[Prop. 3.60]{FSS20b}),
  this is dg-algebraically represented by the
  pullback of dgc-algebras as shown on the right,
  which is computed as a fiber product
  of underlying cochain complexes (e.g. \cite[Prop. 3.41]{FSS20b}):

  \vspace{-.5cm}
  \begin{equation}
    \label{HopfRat}
    \hspace{-3mm}
  \raisebox{40pt}{
  \xymatrix@C=8pt{
   &
    \mathrm{cofib}(\phi)
    \ar@{}[dd]|-{ \mbox{\tiny (po)} }
    \\
    S^{2k}
    \ar[ur]
    &&
    D^{4k}
    \ar[ul]
    \\
    &
    S^{4k-1}
    \ar[ul]^-{\phi}
    \ar[ur]_-{  }
  }
  }
  \hspace{.7cm}
  \raisebox{78pt}{
  \xymatrix@C=10pt@R=2em{
    &
    \scalebox{.85}{$
    \left(
    \begin{aligned}
      d\, \omega_{4k-1}
        & =\;
        - \omega_{2k} \wedge \omega_{2k}
        \\
        & \phantom{=}\;\;
        +
        {\color{darkblue} h} \cdot \omega_{4k}
      \\
      d\, \omega_{4k\phantom{-1}} & =\; 0
      \\
      d\, \omega_{2k\phantom{-1}} & =\; 0
      \\
      \omega_{2k}
        \wedge
      \omega_{4k}
        & = \; 0
    \end{aligned}
    \right)
    $}
    \ar@{}[dd]|-{ \mbox{\tiny \rm (pb)} }
    \ar[dl]_>>>>>{
      \mbox{
        \tiny
        $
        \arraycolsep=1.4pt
        {\begin{array}{lcl}
          \omega_{2k}
           & \mapsfrom &
          \omega_{2k}
          \\
          \omega_{4k-1}
            & \mapsfrom &
          \omega_{4k-1}
          \\
          0
            & \mapsfrom &
          \omega_{4k}
        \end{array}}
        $
      }
    }
    \ar[dr]^>>>>>{
      \mbox{
        \tiny
        $
        \arraycolsep=1.4pt
        {\begin{array}{lcl}
          \omega_{2k} & \mapsto & 0
          \\
          \omega_{4k-1} & \mapsto & {\color{darkblue} h} \cdot \omega_{4k-1}
          \\
          \omega_{4k} & \mapsto & \omega_{4k}
        \end{array}}
        $
      }
    }
    \\
    \scalebox{.85}{$
    \left(
      \!\!\!
      {\begin{array}{lcl}
        d\,\omega_{4k-1}
         & \!\!\!\! = & \!\!\!\!
        - \omega_{2k} \wedge \omega_{2k}
        \\
        d\,\omega_{2k} & \!\!\!\! = & \!\!\!\! 0
      \end{array}}
      \!\!\!
    \right)
    $}
    \ar[dr]_{
      \mbox{
        \tiny
        $
        \arraycolsep=1.4pt
        {\begin{array}{lcl}
          \omega_{2k}
            & \mapsto &
          0
          \\
          \omega_{4k-1}
            & \mapsto &
          {\color{darkblue}{h}}
          \cdot
          \omega_{4k-1}
        \end{array}}
        $
      }
    }
    &&
    \scalebox{.85}{$
    \left(
      \!\!\!
      {\begin{array}{ll}
        d\,\omega_{4k-1} & \!\!\!\!= \omega_{4k}
        \\
        d\,\omega_{4k} & \!\!\!\! = 0
      \end{array}}
      \!\!\!
    \right)
    $}
    \ar[dl]^<<<<<<<{
      \mbox{ \;\;
        \tiny
        $
        \arraycolsep=1.4pt
        {\begin{array}{lcl}
          \omega_{4k-1}
            &  \mapsfrom &
          \omega_{4k-1}
          \\
          0
            & \mapsfrom &
          \omega_{4k}
        \end{array}}
        $
      }
    }
    \\
    &
    \scalebox{.85}{$
    \big(
      d\,\omega_{4k-1}  = 0
    \big)
    $}
    &
    {\phantom{AAAAAAAAAAAAAAA}}
  }
  }
\end{equation}

\vspace{-1mm}
\noindent
 The first differential relation shown in the fiber product dgc-algebra
 implies on cochain cohomology the relation
 $
 [\omega_2 \wedge \omega_2] \,=\, {\color{darkblue}h} \cdot \big[ \omega_4 \big]
 .
 $
  By the fact that the cochain cohomology of $A$
  computes the rational cohomology groups
  of $\mathrm{cofib}(\varphi)$
  (see \cite[Prop. 3.57]{FSS20b}),
  comparison with \eqref{HopfInvariant} shows that $h = \mathrm{HI}(\phi)$.
\end{proof}

Using this, we obtain the following consequence of Theorem \ref{AnomalyFunctionalAsLiftInCohomotopy}:

\begin{prop}[Recovering the homotopy Whitehead formula]
  \label{GWIIsHomotopyInvariantOfBackgroundFieldData}
    In the situation of Def. \ref{BasicSetup},
    consider the special case where:
    \vspace{-1mm}
       \item
  {\bf (i)}
    the background C-field \eqref{SpacetimeFormAssumption} satisfies
              \hyperlink{HypothesisH}{\it Hypothesis H}
       (Def. \ref{BackgroundFieldsSatisfyingHypothesisH});
       \vspace{-1mm}
      \item {\bf (ii)} the extended worldvolume is the 7-sphere
        $\Sigma := \widetilde \Sigma^7 := S^7$
        (as in Lemma \ref{GaugebleSigmaModelMapsOn7Sphere}
        and Example \ref{CoboundariesForS3xS3});
        \vspace{-1mm}
      \item {\bf (iii)}
        {the $\mathrm{Spin}(8)$-bundle over $X$ is trivial, as well as the $\mathrm{Spin}(8)$-connection $\nabla$, and the trivial trivialization $\Theta_7=0$ of $\chi_8(\nabla)$ \eqref{EulerFormAsPfaffian} is chosen.
        }

    \noindent
    Then twice the Hopf-WZ anomaly functional $2\widetilde S$
    (Def. \ref{TheAnomalyFunctional}, Lemma \ref{AnomalyFunctionalIsHomotopyInvariant})
    is equal to the Hopf invariant
    $\mathrm{HI}\big( c \circ \widetilde f \;\big)$
    (Def. \ref{TheHopfInvariant})
    of the composite
          \vspace{-2mm}
    $$
      \xymatrix{
        S^7
        \ar[r]^-{ \widetilde f }
        &
        X
        \ar[r]^-{ c }
        &
        S^4
      }
    $$

    \vspace{-2mm}
\noindent
    of the extended sigma-model field
    $\widetilde f$ \eqref{GaugedExtendedSigmaModelFields}
    with the (untwisted)
    Cohomotopy cocycle $c$ \eqref{CohomotopyCocycle}
    that classifies the background fields:
    \vspace{-1mm}
    \begin{equation}
      \label{HomotopyWhiteheadComputesHopfInvariant}
      2\widetilde S
      \big(
        \widetilde f, \widetilde H_3
      \big)
      \;=\;
      \int_{{}_{S^7}}
      \Big(
        \widetilde H_3 \wedge \widetilde f^\ast G_4
        +
        \widetilde f^\ast 2G_4
      \Big)
      \;=\;
      \mathrm{HI}
      \big(
        c \circ \widetilde f \;
      \big)
      \;\in\;
      \mathbb{Z}\;.
    \end{equation}
\end{prop}

\newpage

\begin{proof}
  Under the given assumptions,
  diagram \eqref{DiagramExhibitinAnomalyFunctionalAsHopfInvariant}
  in Theorem \ref{AnomalyFunctionalAsLiftInCohomotopy}
  reduces to the following:
    \vspace{-2mm}
    \begin{equation}
      \label{TheWhiteheadIntegralFromTheHopfWZTerm}
      \raisebox{50pt}{
      \xymatrix@C=3.5em{
        \mathllap{
          \widetilde \Sigma^7
          =
          \;
        }
        S^7
        \ar[dd]_-{ \widetilde f }
        \ar[rr]^-{
          {\color{darkblue}
            \mathrm{HI}
            (
              c \circ \widetilde f
            \;)
          }
        }_>>>>>{\ }="s"
        &&
        S^7
        \ar[dd]^-{ h_{\mathbb{H}} }
        &&
        \Omega_{\mathrm{dR}}^\bullet
        (
          S^7
        )
        \ar@{<-}[rr]^-{
          \mbox{
            \tiny
            $
            \arraycolsep=2.2pt
            \begin{array}{lcl}
              {\color{darkblue}2\widetilde S}
                \cdot
              \mathrm{vol}_{S^7}
              &\!\!\!\!\mapsfrom\!\!\!\!&
              \widetilde \omega_{\, 7}
            \end{array}
            $
          }
        }_>>>>>{\ }="s2"
        \ar@{<-}[drr]_-{
          \mbox{
            \tiny
            $
            \arraycolsep=2.2pt
            \begin{array}{rcl}
              \widetilde H_3 &\!\!\!\!\mapsfrom\!\!\!\!& h_3
              \\
              \widetilde f^\ast G_4 &\!\!\!\!\mapsfrom\!\!\!\!& \omega_4
              \\
              \widetilde f^\ast 2G_7 &\!\!\!\!\mapsfrom\!\!\!\!& \omega_{\, 7}
            \end{array}
            $
          }
        }^>>>>>{\ }="t2"
        \ar@{<-}[dd]_{
          \widetilde f^\ast
        }
        &&
        \left(
          \!\!\!
          {\begin{array}{lcl}
            d\,\widetilde \omega_{\, 7} &=& 0
          \end{array}}
          \!\!\!
        \right)
        \ar@{<-}[d]_-{\simeq}^-{
          \mbox{
            \tiny
            $
            \arraycolsep=1.4pt
            \begin{array}{ccc}
              0 & 0 & \widetilde \omega_{\, 7}
              \\
              \mapsup & \mapsup & \mapsup
              \\
              h_3 & \omega_4 & \omega_{\, 7}
            \end{array}
            $
          }
        }
        \\
               &&
        &&
        &&
        \left(
          \!\!\!
          {\begin{array}{lcl}
            d\,h_3 & =& \omega_4
            \\
            d\,\omega_4 &=& 0
            \\
            d\,\omega_{\, 7} &=& -\omega_4 \wedge \omega_4
          \end{array}}
          \!\!\!
        \right)
        \ar@{<-_{)}}[d]^-{
          \mbox{
            \tiny
            $
            \arraycolsep=1.4pt
            \begin{array}{cc}
              \omega_4 & \omega_{\, 7}
              \\
              \mapsup & \mapsup
              \\
              \omega_4 & \omega_{\, 7}
            \end{array}
            $
          }
        }
        \\
        X
        \ar[rr]_-{ c }^<<<<<<{\ }="t"
        &&
        S^4
        &&
        \Omega_{\mathrm{dR}}^\bullet(X)
        \ar@{<-}[rr]^-{
          \mbox{
            \tiny
            $
            \arraycolsep=2.2pt
            \begin{array}{lcl}
              G_4 &\!\!\!\mapsfrom\!\!\!& \omega_4
              \\
              G_7 &\!\!\!\mapsfrom\!\!\!& \omega_{\, 7}
            \end{array}
            $
          }
        }
        &&
        \left(
          \!\!\!
          {\begin{array}{lcl}
            d\,\omega_4 &=& 0
            \\
            d\,\omega_{\, 7} &=& - \omega_4 \wedge \omega_4
          \end{array}}
          \!\!\!
        \right)
        \ar@{=>}^-{\simeq}_-{ \widetilde H_3 }  "s"; "t"
        \ar@{<=}^-{\simeq}_-{ \eta^\ast } "s2"; "t2"
      }
      }
    \end{equation}

    \vspace{-2mm}
\noindent     Now observe (e.g., \cite[Lem. 3.18]{FSS19b})
    that, with the given normalization of the differential
    relation $d \omega_{\, 7} = - \omega_4 \wedge \omega_4$ on the right,
    the generator $\omega_{\, 7}$ is (maps to) the rational image of the
    {\it integral} generator
    $1 \,\in\, \mathbb{Z} \,\simeq\, H^7\big( S^7;\, \mathbb{Z}\big)$.
    Since in $\Omega^\bullet_{\mathrm{dR}}\big( S^7 \big)$ this integral
    generator is represented by the volume form
    (instead of some multiple of it),
    Lemma \ref{RecognitionOfHopfInvariantFromSullivanModel}
    applied to the diagonal map in
    \eqref{TheWhiteheadIntegralFromTheHopfWZTerm}
    implies the claim \eqref{HomotopyWhiteheadComputesHopfInvariant}.
\end{proof}

\begin{remark}[Whitehead integral formulas in the literature]
\label{WhiteheadIntegralFormulasInTheLiterature}
The statement of Prop. \ref{GWIIsHomotopyInvariantOfBackgroundFieldData}
is essentially that of
\cite[p. 17]{Haefliger78}
\cite[\S 14.5]{GrMo13},
the integrand being the \emph{functional cup product}-expression of \cite{Steenrod49},
recalled as a \emph{homotopy period}-expression in
\cite[Ex. 1.9]{SinhaWalter08};
but the proof as a special case
of Theorem \ref{AnomalyFunctionalAsLiftInCohomotopy}
is new and more conceptual.
In the special case that $G_7 = 0$
(which is given if
the classifying map $X \overset{c }{\to} S^4$ \eqref{CohomotopyCocycle}
is a smooth function)
the statement of Prop. \ref{GWIIsHomotopyInvariantOfBackgroundFieldData}
further reduces to that of the classical \emph{Whitehead integral formula}
for the Hopf invariant
\cite{Whitehead47} (see \cite{Haefliger78}\cite[Prop. 17.22]{BT82}).
    \end{remark}

Our second Theorem
generalizes this integral situation to
arbitrary oriented differences $\widetilde \Sigma^7$
of extended worldvolumes
and to non-trivial topological twists:

\begin{theorem}[6d Hopf-WZ anomaly functional is integral]
  \label{AnomalyIsIntegral}
Let target space $X = X^8$ be
an M2-brane background for M-theory on 8-manifolds
(Example \ref{BackgroundFieldsAsRationalCohomotopy}),
hence a smooth spin 8-manifold
equipped with tangential $\mathrm{Sp}(2)$-structure
\eqref{TangentialSp2Structure}
and with vanishing Euler class \eqref{hatSp2Structure}.

Then
\hyperlink{HypothesisH}{\it Hypothesis H}
(Def. \ref{BackgroundFieldsSatisfyingHypothesisH})
implies that
the general 6d Hopf-Wess-Zumino anomaly functional of the M5-brane
(Def. \ref{TheAnomalyFunctional}, Lemma \ref{AnomalyFunctionalIsHomotopyInvariant})
takes values in the integers
\vspace{-3mm}
  \begin{equation}
    \label{GWIFunctionInt}
    \xymatrix@R=-6pt@C=4em{
      \pi_0
      \Big(
        \mathrm{Maps}_{\mathrm{smth}}^{\mathrm{ggd}}
        \big(\,
          \widetilde \Sigma^7
          ,
          X
        \big)
      \Big)
      \ar@/^1.4pc/[rrr]^<<<<<<<<<<<<<<<<<<<<<<{ {\color{darkblue}2\widetilde S} }
      \ar[rr]^{
      }
      &&
     \quad \mathbb{Z}
      \quad \ar@{^{(}->}[r]
      &
      \;\; \mathbb{R}
      \\
      \big[
        \widetilde f, \widetilde H_3
      \big]
      \ar@{}[rr]|-{\quad \longmapsto}
      &&
      \underset{\widehat \Sigma^7 }{\bigintsss}
      \left(
        \widetilde H_3
          \wedge
        \widetilde f^\ast
        \big(
          G_4
          +
          \tfrac{1}{4}
          p_1(\nabla)
        \big)
        +
        \widetilde f^\ast 2G_7
      \right)
    }
  \end{equation}

    \vspace{-2mm}
\noindent
  and hence that the exponentiated action
  \eqref{ConsistencyCondition} of
  the 6d WZ term of the M5-brane (Def. \ref{DefinitionOfWZWTermByFieldExtensionToCoboundary})
  is well-defined.
\end{theorem}
\begin{proof}
Recall that,

\noindent
\begin{itemize}
\vspace{-.2cm}
\item[{\bf (a)}]
by Theorem \ref{AnomalyFunctionalAsLiftInCohomotopy},
the anomaly term is characterized as a homotopy lift
through the rationalization \newline
$L_{\mathbb{R}}\big(h_{\mathbb{H}} \!\sslash\! \mathrm{Sp}(2)\big)$
of the parametrized quaternionic Hopf fibration;

\vspace{-.2cm}
\item[{\bf (b)}]
by \hyperlink{HypothesisH}{\it Hypothesis H} this comes from a
lift through the actual $\mathrm{Sp}(2)$-parametrized Hopf fibration
$h_{\mathbb{H}} \!\sslash\! \mathrm{Sp}(2)$ \eqref{CohomotopyCocycle};

\vspace{-.2cm}
\item[{\bf (c)}]
by \eqref{hatSp2Structure}
this, in turn,
comes from a lift through the $\widehat{\mathrm{Sp}(2)}$-parametrized
Hopf fibration
$h_{\mathbb{H}} \!\sslash\! \widehat{\mathrm{Sp}(2)}$
(Def. \ref{QuaternionicHopfFibrationParametrizedOverHigherExtendedSp2}).
\end{itemize}

\vspace{-2mm}
\noindent Therefore, it is now sufficient to show that the
rational class \eqref{Rational7ClassOnE7}
of the universal Hopf-WZ term
on the total space of the universal $\widehat{\mathrm{Sp}(2)}$-parametrized
7-sphere fibration \eqref{PullbackToUniversalSpaceOnWhichChi8Trivializes}
is the rational image of an integral cohomology class.

\newpage

First, consider the Gysin sequence (see e.g. \cite[15.30]{Switzer75})
for the universal 4-spherical fibration
in the bottom right of \eqref{PullbackToUniversalSpaceOnWhichChi8Trivializes}
\vspace{-2mm}
$$
  \xymatrix{
    S^4
      \ar[rr]^-{\scalebox{0.7}{$\mathrm{hofib}(\rho_{S^4})$} }
    &&
    S^4
    \!\sslash\!
    \mathrm{Sp}(2)
      \ar[rr]^-{ \rho_{S^4} }
    &&
    B \mathrm{Sp}(2)
  }
  \,.
$$

\vspace{-1mm}
\noindent Observe that the integral cohomology groups of the classifying space
(see e.g. \cite{Pittie91}\cite[(12)]{Kalkkinen06})
\vspace{-2mm}
\begin{equation}
  \label{IntegralCohomologyOfSp2}
  H^\bullet
  \big(
    B \mathrm{Sp}(2),
    \mathbb{Z}
  \big)
  \;\simeq\;
  \mathbb{Z}
  \big[
    \tfrac{1}{2}p_1,
    \rchi_8
  \big]
\end{equation}

\vspace{-2mm}
\noindent
are non-torsion groups concentrated in even degrees.
Hence the 5-class controlling this Gysin sequence vanishes,
and so the long exact sequence breaks up into short exact sequences of
this form:
\vspace{-3mm}
\begin{equation}
  \label{GysinSES}
  \xymatrix@R=9pt{
    0
    \ar[r]
    &
    H^\bullet
    \big(
      B \mathrm{Sp}(2)
      ;
      \,
      \mathbb{Z}
    \big)
    \ar[r]^-{ \rho^\ast_{S^4} }
    &
    H^\bullet
    \big(
      S^4 \!\sslash\! \mathrm{Sp}(2)
      ;
      ,
      \mathbb{Z}
    \big)
    \ar[r]^-{ \int_{S^4} }
    &
    H^{\bullet -4}
    \big(
      B \mathrm{Sp}(2)
      ;
      \,
      \mathbb{Z}
    \big)
    \ar[r]^{  }
    &
    0
  }.
\end{equation}

\vspace{-2mm}
\noindent
Moreover, since the integral cohomology groups \eqref{IntegralCohomologyOfSp2}
have no torsion,
these short exact sequences imply that also
\vspace{-2mm}
\begin{equation}
  \label{TorsionVanishes}
  H^\bullet
  \big(
    S^4 \!\sslash\! \mathrm{Sp}(2)
    ;
    \,
    \mathbb{Z}
  \big)
\end{equation}

\vspace{-2mm}
\noindent
are non-torsion groups.

Now observe, by \cite[Prop. 3.13]{FSS19b}, that
\vspace{-1mm}
\begin{equation}
  \label{Integralomega4}
  \widetilde \Gamma_4
  \;:=\;
  \omega_4 + \tfrac{1}{4}p_1
  \;\in\;
  H^4\big(
    S^4 \!\sslash\! B \mathrm{Sp}(2);
    \,
    \mathbb{Z}
  \big)
  \longrightarrow
  H^4\big(
    S^4 \!\sslash\! B \mathrm{Sp}(2)
    ;
    \,
    \mathbb{R}
  \big)
\end{equation}

\vspace{-2mm}
\noindent
is an integral class,
being the universal integral shifted C-field flux density \eqref{G4Plus}.
Hence, by \eqref{TorsionVanishes}, the rational trivialization
from \eqref{DiagramExhibitinAnomalyFunctionalAsHopfInvariant}
$$
  \begin{aligned}
    d \omega_{\, 7}
    & =
    - \omega_4 \wedge \omega_4
    +
    \big( \tfrac{1}{4} p_1 \big)^2
    -
    \rchi_8
    \\
    & =
    -
    \big(
      \widetilde \Gamma_4
      \wedge
      \widetilde \Gamma_4
      -
      \tfrac{1}{2}p_1 \wedge \widetilde \Gamma_4
    \big)
    -
    \rchi_8
  \end{aligned}
$$
implies that also the following integral cohomology class vanishes:
\vspace{-2mm}
\begin{equation}
  \label{IntegralVanishingQuadratic}
  \left[
    \widetilde \Gamma_4
    \cup
    \widetilde \Gamma_4
    -
    \tfrac{1}{2}p_1
    \cup
    \widetilde \Gamma_4
    +
    \rchi_8
  \right]
  \;=\;
   0
   \;\;\in\;\;
   H^8
   \big(
     S^4
     \!\sslash\!
     \mathrm{Sp}(2)
     ;
     \,
     \mathbb{Z}
   \big)
   \,.
\end{equation}

\vspace{-2mm}
\noindent
Consider next the integral Gysin sequence corresponding to the
3-spherical fibration which is the parametrized
quaternionic Hopf fibration in the top right of \eqref{PullbackToUniversalSpaceOnWhichChi8Trivializes}:
      \vspace{-3mm}
\begin{equation}
  \label{E4IntegralVanishingQuadratic}
  \xymatrix{
    S^3
    \ar[rr]^-{\scalebox{0.6}{$
      \mathrm{hofib}(
        h_{\mathbb{H}}
          \sslash
        G
     )
     $}
    }
    &&
    S^7
      \!\sslash\!
    \mathrm{Sp}(2)
    \ar[rr]^-{\scalebox{0.6}{$
      h_{\mathbb{H}}
        \sslash
      \mathrm{Sp}(2)
      $}
    }
    &&
    S^4
    \!\sslash\!
    \mathrm{Sp}(2)
  }.
\end{equation}

    \vspace{-2mm}
\noindent
Since, rationally,
$S^7 \!\sslash\! \mathrm{Sp}(2)$
is obtained from $S^4 \!\sslash\! \mathrm{Sp}(2)$
by adjoining the relation
\begin{equation}
  \label{dh3}
  \begin{aligned}
    d h_3
    & =
    \omega_4 - \tfrac{1}{4}p_1
    \\
    & = \widetilde \Gamma_4 - \tfrac{1}{2}p_1
    \,,
  \end{aligned}
\end{equation}

    \vspace{-2mm}
\noindent
by \eqref{DiagramExhibitinAnomalyFunctionalAsHopfInvariant},
it follows from \cite[Prop. 2.5 (44)]{FSS19b}
that $\omega_4 - \tfrac{1}{4}p_1$ is the rational image of the
integral Euler class
of \eqref{E4IntegralVanishingQuadratic}. Consequently,
\begin{equation}
  \label{EulerIntegral4ClassFor3SphereFibration}
  \widetilde \Gamma_4
  -
  \tfrac{1}{2}p_1
  \;\in\;
  H^4
  \big(
    S^4
    \!\sslash\!
    \mathrm{Sp}(2)
    ;
    \,
    \mathbb{Z}
  \big)
\end{equation}
is the integral 4-class controlling the Gysin sequence of
\eqref{E4IntegralVanishingQuadratic},
which therefore reads:
      \vspace{-2mm}
\begin{equation}
  \label{3SphereGysinSequenceUniversal}
  \xymatrix@C=1.5em{
    \cdots
    \ar[r]
    &
    H^7
    \big(
      S^4
      \!\sslash\!
      \mathrm{Sp}(2)
      ;
      \,
      \mathbb{Z}
    \big)
    \ar[rr]^-{
            \scalebox{0.6}{$(
        h_{\mathbb{H}} \sslash \mathrm{Sp}(2)
      )^\ast
      $}
    }
    &&
    H^7
    \big(
      S^7
      \!\sslash\!
      \mathrm{Sp}(2)
      ;
      \,
      \mathbb{Z}
    \big)
    \ar[r]^-{
      \int_{S^3}
    }
    &
    H^4
    \big(
      S^4
      \!\sslash\!
      \mathrm{Sp}(2)
      ;
      \,
      \mathbb{Z}
    \big)
    \ar[rr]^-{\scalebox{0.6}{$
      \cup
      (
        \widetilde \Gamma_4
        -
        \frac{1}{2}p_1
      )
      $}
    }
    &&
    H^8
    \big(
      S^4
      \!\sslash\!
      \mathrm{Sp}(2)
      ;
      \,
      \mathbb{Z}
    \big)
    \ar[r]
    &
    \cdots
  }
\end{equation}

Now consider pulling back this situation
along the homotopy fiber of $\rchi_8$ \eqref{HomotopyFiberSpace}
to yield the sequence of spherical fibrations
on the left of \eqref{PullbackToUniversalSpaceOnWhichChi8Trivializes}.
After this pullback, the Euler class
summand in \eqref{IntegralVanishingQuadratic} disappears,
and we obtain this vanishing class:
\begin{equation}
  \label{IntegralVanishingQuadraticOverSpacetime}
  \left[
    \widetilde \Gamma_4
    \cup
    \widetilde \Gamma_4
    -
    \tfrac{1}{2}p_1(T X)
    \cup
    \widetilde \Gamma_4
  \right]
  \;=\;
   0
   \;\in\;
   H^8
   \big(
     S^4 \!\sslash\! \widehat{\mathrm{Sp}(2)}
     ;
     \,
     \mathbb{Z}
   \big)
   \,.
\end{equation}
With this, the integral Gysin sequence of the
3-spherical fibration
in the top left of \eqref{PullbackToUniversalSpaceOnWhichChi8Trivializes}
\vspace{-2mm}
$$
  \xymatrix{
    S^3
    \ar[rr]^-{\scalebox{0.6}{$ \mathrm{hofib}( h_{\mathbb{H}} \sslash G )$} }
    &&
    S^7 \!\sslash\! \widehat{\mathrm{Sp}(2)}
    \ar[rr]^-{\scalebox{0.6}{$ h_{\mathbb{H}} \sslash \widehat{\mathrm{Sp}(2)}$} }
    &&
    S^4 \!\sslash\! \widehat{\mathrm{Sp}(2)}
  }
$$

\vspace{-2mm}
\noindent
 is seen to be of the following form:
\vspace{-2mm}
\begin{equation}
  \label{OverSpacetime3SphereGysinSequence}
  \raisebox{24pt}{\small
  \xymatrix@C=10pt@R=-6pt{
    \mathllap{\cdots}
    \ar[r]
    &
    H^7
    \big(
      S^4 \!\sslash\! \widehat{\mathrm{Sp}(2)}
      ;
      \,
      \mathbb{Z}
    \big)
    \ar[rrr]^-{\scalebox{0.6}{$
      (
          h_{\mathbb{H}} \sslash \widehat{\mathrm{Sp}(2)}
      )^\ast
      $}
    }
    &&&
    H^7
    \big(
      S^7 \!\sslash\! \widehat{\mathrm{Sp}(2)}
      ;
      \,
      \mathbb{Z}
    \big)
    \ar[r]^-{
      \int_{S^3}
    }
    &
    H^4
    \big(
      S^4 \!\sslash\! \widehat{\mathrm{Sp}(2)}
      ;
      ,
      \mathbb{Z}
    \big)
    \ar[rrr]^-{\scalebox{0.6}{$
      \cup
      \left(
        \widetilde \Gamma_4 - \frac{1}{2}p_1(T X)
      \right)
      $}
    }
    &&&
    H^8
    \big(
      S^4 \!\sslash\! \widehat{\mathrm{Sp}(2)}
      ;
      \,
      \mathbb{Z}
    \big)
    \ar[r]
    &
    \mathrlap{\cdots}
    \\
    &&
    &&
    {\color{darkblue} 2\widetilde{\mathbf{S}}}
    \ar@{|->}[r]
    &
    \widetilde \Gamma_4
    \ar@{|->}[rrr]
    &&&
    \underset{
      = 0
    }{
      \underbrace{
        \widetilde \Gamma_4
        \cup
        \widetilde \Gamma_4
        -
        \tfrac{1}{2}p_1(T X)
        \cup
        \widetilde \Gamma_4
      }
    }
  }
  }
\end{equation}

\vspace{-3mm}
\noindent
Here, in the bottom row, we have observed that the
image of $\widetilde \Gamma_4$ \eqref{Integralomega4} under
forming cup product with the 4-class \eqref{EulerIntegral4ClassFor3SphereFibration}
is just the vanishing class \eqref{IntegralVanishingQuadraticOverSpacetime},
which by exactness of the Gysin sequence implies that there
exists an  integral 7-class
\vspace{-2mm}
\begin{equation}
  \label{Integral7Class}
  {\color{darkblue}2 \widetilde{\mathbf{S}}}
  \;\in\;
  H^7\big(
    S^7 \sslash \widehat{\mathrm{Sp}(2)}
    ;
    \,
    \mathbb{Z}
  \big)
  \,,
\end{equation}

\vspace{-2mm}
\noindent
whose integration over the $S^3$-fibers is
$\widetilde \Gamma_4$, as shown.
Since, by \eqref{dh3} and \cite[Prop. 2.5 (45)]{FSS19b},
the fiberwise volume form is $h_3$,
this is, rationally, the same fiber integration as that of
\eqref{Rational7ClassOnE7}, which by the exactness of the
Gysin sequence \eqref{OverSpacetime3SphereGysinSequence},
now with rational coefficients
\vspace{-2mm}
\begin{equation}
  \label{OverSpacetime3SphereRationalGysinSequence}
  \hspace{-5mm}
  \xymatrix@R=-2pt{
    \cdots
    \ar[r]
    &
    H^7
    \big(
      S^4 \!\sslash\! \widehat{\mathrm{Sp}(2)}
      ;
      \,
      \mathbb{R}
    \big)
    \ar[rr]^-{\scalebox{0.6}{$
      (
          h_{\mathbb{H}} \sslash \widehat{\mathrm{Sp}(2)}
      )^\ast
      $}
    }
    &&
    H^7
    \big(
      S^7 \!\sslash\! \widehat{\mathrm{Sp}(2)}
      ;
      \,
      \mathbb{R}
    \big)
    \ar[r]^-{
      \int_{S^3}
    }
    &
    H^4
    \big(
      S^4 \!\sslash\! \widehat{\mathrm{Sp}(2)}
      ;
      \,
      \mathbb{R}
    \big)
    \ar[r]
    &
    \cdots
    \\
    &
    D
    \ar@{|->}[rr]
    &&
          \left(
                {h_3\wedge ( \omega_4 + \frac{1}{4}p_1 )
    +
        \omega_{\, 7} + \theta_7
      }
            \atop
      { - { {\color{darkblue}2\widetilde{\mathbf{S}}} }}
           \right)
        \ar@{|->}[r]
    &
    0
  }
\end{equation}

\vspace{-3mm}
\noindent
implies that the integral class
${\color{darkblue}2\widetilde{\mathbf{S}}}$  \eqref{Integral7Class}
differs from the rational class \eqref{Rational7ClassOnE7} by
a 7-class $D$ pulled back from $S^4 \!\sslash\! \widehat{\mathrm{Sp}(2)}$, as shown.
But by
\eqref{IntegralCohomologyOfSp2}
and
\eqref{DiagramExhibitinAnomalyFunctionalAsHopfInvariant}
there is no non-trivial 7-class on $S^4 \!\sslash\! \widehat{\mathrm{Sp}(2)}$.
Hence the equality
\vspace{-2mm}
$$
  {\color{darkblue}
  2\widetilde{\mathbf{S}}
  }
  \;=\;
  h_3 \wedge \widetilde \Gamma_4
  +
  (\omega_{\, 7} + \theta_7)
$$

\vspace{-2mm}
\noindent
 holds,
and so the anomaly integrand
\eqref{Rational7ClassOnE7} is indeed the rational
image of an integral class \eqref{Integral7Class} and hence
has itself integral periods:
\vspace{-1mm}
\begin{equation}
  \label{IntegralOfRationalImageOfIntegralClass}
  \hspace{-3cm}
  \xymatrix@R=2pt{
    {\color{darkblue}2\widetilde{\mathbf{S}}}
    \ar@{|->}[ddddd]
    &
    H^7\big(
      S^7 \!\sslash\! \widehat{\mathrm{Sp}(2)}
      ;
      \,
      \mathbb{Z}
    \big)
    \ar[rr]^-{ \scalebox{0.6}{$
      \big( \,\widehat{ c \circ \widetilde f } \;\big)^\ast
    $}
    }
    \ar[dddd]
    &&
    H^7\big(
      \widetilde \Sigma^7;
      \,
      \mathbb{Z}
    \big)
    \ar[dddd]
    \ar[rr]^-{ \int_{\widetilde \Sigma^7} }
    &&
    \mathbb{Z}
    \mathclap{\phantom{\vert_{\vert_{\vert}}}}
    \ar@{^{(}->}[dddd]
    \\
    \\
    \\
    \\
    &
    H^7\big(
      S^7 \!\sslash\! \widehat{\mathrm{Sp}(2)}
      ;
      \,
      \mathbb{R}
    \big)
    \ar[rr]^-{\scalebox{0.6}{$
      \big( \,\widehat{ c \circ \widetilde f } \;\big)^\ast
      $}
    }
    &&
    H^7\big(
      \widetilde \Sigma^7;
      \,
      \mathbb{R}
    \big)
    \ar[rr]^-{ \int_{\widetilde \Sigma^7} }
    &&
    \mathbb{R}
    \\
    {
      { h_3\wedge( \omega_4 + \frac{1}{4}p_1 ) }
      \atop
      { + \omega_{\, 7} + \theta_7 }
    }
    \ar@{|->}[rrr]
    & &&
    {
      { \widetilde H_3 \wedge \widetilde f^\ast \widetilde G_4 }
      \atop
      { + \widetilde f^\ast 2 G_7 }
    }
    \ar@{|->}[rr]
    &&
    2
    \widetilde S^{\; \rm M5}_{\mathrm{WZ}}
    \big(
      \widetilde f, \widetilde H_3
    \big)
    \mathrlap{
      \; =
      \widetilde S^{\; 1 \, \rm M5}_{\mathrm{WZ}}
      \big(
        \widetilde f, \widetilde H_3
      \big).
    }
  }
\end{equation}

\vspace{-7mm}
\end{proof}

\medskip

\noindent {\bf Concluding remarks.} We have
established that charge quantization of the C-field in
J-twisted Cohomotopy implies,
beyond integrality of the $\tfrac{1}{4}p_1$-shifted 4-flux,
also integrality of the 7-flux which
cancels the
Hopf-WZ anomaly term
and makes the Page charge \eqref{DifferenceBetweenAnyTwoExtendedActionFunctionals}
integral:

\vspace{-5mm}
\begin{equation}
  \label{SummaryDiagram}
  \hspace{-3mm}
  \raisebox{40pt}{
  \xymatrix@C=3.7em{
    \widetilde \Sigma^7
    \ar[dd]_-{
      \widetilde f
    }
    \ar[rrrr]|-{  \;\widehat{ c \circ \widetilde f }\;\, }
    \ar@/^2pc/[rrrrrrrr]^-{
      \overset{
        \raisebox{3pt}{
          \tiny
          {\color{darkblue}
          \bf
            integral
            Hopf-WZ anomaly/Page charge
          }
          (Thm. \ref{AnomalyIsIntegral})
        }
      }{
        \big[
          2\widetilde S \;:=\; \widetilde H_3 \wedge \widetilde G_4 + 2G_7
        \big]
      }
    }
    &&
    \ar@{}[dd]|-{
      \mbox{
        \tiny
        \begin{tabular}{c}
          \color{greenii}
          \bf
          cocycle in
          \\
          \color{greenii}
          \bf
          twisted Cohomotopy
          \\
          (Def. \ref{BackgroundFieldsSatisfyingHypothesisH})
        \end{tabular}
      }
    }
    &&
    S^7 \!\sslash\! \widehat{\mathrm{Sp}(2)}
    \ar[dd]|-{ \scalebox{0.65}{$
      \phantom{\vert^{\vert^{\vert^{\vert}}}}
      \underset{
        \raisebox{-3pt}{
          \tiny
          \begin{tabular}{c}
            \color{greenii}
            \bf
            parametrized quaternionic Hopf fibration
            \\
            {\color{greenii}
            \bf
            over Fivebrane-extended $\mathrm{Sp}(2)$}
            (Def. \ref{QuaternionicHopfFibrationParametrizedOverHigherExtendedSp2})
          \end{tabular}
        }
      }{
        h_{\mathbb{H}} \!\sslash\! \widehat{\mathrm{Sp}(2)}
      }
      \phantom{\vert_{\vert_{\vert_{\vert}}}}
      $}
    }
    \ar[rrrr]^-{
      2\widetilde{\mathbf{S}}
      \;:=\;
      \widetilde \omega_7 + \theta_7
    }_-{
      \underset{
        \raisebox{-3pt}{
          \tiny
          \begin{tabular}{c}
            {
            \color{darkblue}
            \bf
            universal integral Hopf-WZ/Page term
            }
            \eqref{UniversalRationalHopfWZFlux}
          \end{tabular}
        }
      }{
        =\; \omega_7 + h_3 \wedge \widetilde \omega_4 + \theta_7
      }
    }
    &&&&
    B^7 \mathbb{Z}
    \\
    \\
    X^8
    \ar[rrrr]|-{ \;c\; }
    \ar@/_2pc/[rrrrrrrr]_-{
      \underset{
        \raisebox{-3pt}{
          \tiny
          {
          \color{darkblue}
          \bf
          integral shifted 4-flux
          }
          \cite[Prop. 3.13]{FSS19b}
        }
      }{
        \big[
          \widetilde G_4 \;:=\; G_4 +
          \scalebox{.7}{$\tfrac{1}{4}$}p_1(\nabla)
        \big]
      }
    }
    &&&&
    S^4 \!\sslash\! \widehat{\mathrm{Sp}(2)}
    \ar[rrrr]^-{
      \overset{
        \raisebox{3pt}{
          \tiny
          \color{darkblue}
          \bf
          \begin{tabular}{c}
            universal
            shifted integral 4-flux
          \end{tabular}
        }
      }{
        \widetilde \Gamma_4
          \;:=\;
        \omega_4 + \scalebox{.7}{$\tfrac{1}{4}$}p_1
      }
    }
    &&&&
    B^4 \mathbb{Z}
  }
  }
\end{equation}
\vspace{-.2cm}

\noindent
This lends further evidence to
\hyperlink{HypothesisH}{\it Hypothesis H} that the M-theory
C-field is charge-quantized in twisted Cohomotopy theory.
Notice that twisted Cohomotopy, while different, is not unrelated to K-theory;
the comparison is discussed in \cite{BMSS19}\cite{SS19a}\cite{BSS19}.
The ability of our Cohomotopy approach to deal with electric and magnetic currents
simultaneously is noteworthy and deserves further exploration elsewhere.


\vspace{1cm}
\noindent Domenico Fiorenza,  {\it Dipartimento di Matematica, La Sapienza Universita di Roma, Piazzale Aldo Moro 2, 00185 Rome, Italy.}

 \medskip
\noindent Hisham Sati, {\it Mathematics, Division of Science, New York University Abu Dhabi, UAE.}

 \medskip
\noindent Urs Schreiber,  {\it Mathematics, Division of Science, New York University Abu Dhabi, UAE, on leave from Czech Academy of Science, Prague.}

\end{document}